\def\defit{\bf }

\def\Bbb{\mathbb}

\def\Bbb{\mathbb}
\def\reals{\Bbb R}

\def\cal{\mathcal}

\documentclass[12pt]{amsart}  
\usepackage{amssymb}    

\usepackage{graphicx,amssymb,amsthm,verbatim,amsfonts}

\theoremstyle{plain}                    
\newtheorem{thm}{Theorem}[section]
\newtheorem{cor}[thm]{Corollary}
\newtheorem{lemma}[thm]{Lemma}

\newcounter{ques}
   
\numberwithin{equation}{section}

\begin{document}
\baselineskip=18pt


%

\title [Nonobtuse Triangulations of PSLGs]
          {Nonobtuse triangulations of PSLGs}

\author {Christopher J. Bishop}
\address{C.J. Bishop\\
         Mathematics Department\\
         SUNY at Stony Brook \\
         Stony Brook, NY 11794-3651}
\email {bishop@math.sunysb.edu}
\thanks{The  author is partially supported by NSF Grant DMS 13-05233.  } 
\date{March  2011; revised Nov 2014; re-revised Jan 2016}
\maketitle


\begin{abstract}
We show that any planar straight line 
graph  with $n$ vertices  has a conforming triangulation 
by $O(n^{2.5})$ nonobtuse triangles (all angles $\leq 90^\circ$),
 answering the question of whether any polynomial bound exists. 
A nonobtuse triangulation is Delaunay, so this result also improves 
a previous $O(n^3)$ bound of Eldesbrunner and Tan for 
conforming Delaunay triangulations of PSLGs.
In the special case that  the PSLG is the triangulation of a simple polygon, 
we will show that only 
$O(n^2)$ triangles are needed,
 improving an $O(n^4)$ bound of Bern and Eppstein.
 We also 
show that for any $\epsilon >0$, every PSLG has  a conforming 
triangulation with $O(n^2 /\epsilon^2)$ elements   and with all 
angles bounded above by $90^\circ + \epsilon$. 
This  improves a result of S.  Mitchell when $\epsilon = 
\frac 38 \pi =67.5^\circ $ and Tan 
 when $\epsilon = \frac 7{30} \pi = 42^\circ $.
\end{abstract}

\clearpage


\setcounter{page}{1}
\renewcommand{\thepage}{\arabic{page}}
\section{Introduction} \label{Intro}

A  planar straight line graph  $\Gamma$ 
(or PSLG from now on) is  the 
disjoint union of a finite number
 (possibly none) of  non-intersecting 
open line segments (called the edges of $\Gamma$)
  together with a disjoint finite point set  (the vertices 
of $\Gamma$) 
that includes all the endpoints of the line segments,
but may include other points as well. 

If $V$ is a  finite  point set in the plane, 
a {\defit triangulation} of $V$ is a PSLG with vertex 
set $V$ and a maximal set of edges.
If $\Gamma$ is a PSLG with vertex set 
$V$, then a conforming triangulation for 
$\Gamma$ is a triangulation of a point 
set $V'$ that contains $V$ and such that  
the union of  the vertices and edges of the triangulation 
covers  $\Gamma$.  We allow $V'$ to be strictly 
larger than $V$; in this case $V'\setminus V$ are
called the  {\defit Steiner points}.
 We want to build conforming triangulations 
for $\Gamma$ that have  small 
complexity (the number of triangles used)
 and good  geometry (the shape 
of each triangle; no angles too large or too small),
 but these two goals are often in conflict.
In this paper, we are  interested in  finding the best 
 angle bounds on the triangles  that allow us to    
polynomially bound the number of triangles needed  
 in  terms of $n$,  the number of vertices of  $\Gamma$.

If we triangulate a  $1 \times r$ rectangle into a fixed 
number of elements, it is easy to check that some angles
must tend to zero as $r \to \infty$, so there is 
no  uniform, strictly positive  lower angle bound possible,
if we want 
the number of triangles  to be bounded only in
terms of the  number of vertices of the given PSLG.
 Since the angles 
of a triangle sum to $180^\circ$, if we had an upper 
bound of $90^\circ - \epsilon$  on the angles of 
a triangulation, then we also have a $2 \epsilon$ lower 
angle bound. 
Therefore, no upper bound strictly less than $90^\circ$
is possible.
Thus {\defit nonobtuse  triangulation}
(all angles $\leq 90^\circ$) is the 
best we can hope for.

In 1960 Burago and Zalgaller \cite{BZ60} showed that any polyhedral 
surface has an {\defit acute triangulation}
 (all angles $< 90^\circ$),
 but without giving a  bound on
the number of triangles needed. This was used as a technical 
lemma in their proof of a polyhedral version of the Nash embedding 
theorem. 
In 1984 Gerver \cite{Gerver} 
 used the Riemann mapping theorem to show that if a polygon's  
angles all exceed $36^\circ$, then there 
exists a dissection of it  into triangles with maximum angle $72^\circ$ (in 
a dissection, adjacent triangles need not meet along an entire
edge). 
In 1988 Baker,  Grosse  and  Rafferty  \cite{BGR}
again proved that any  polygon has  a nonobtuse triangulation,
and their construction also gives a  lower
angle bound. As noted above, in this case no complexity
  bound in terms of $n$ 
alone is possible, although there is a sharp bound  in terms of 
integrating the local feature size over the polygon. 
For details, see 
\cite{BEG-1994},
\cite{Ruppert-1993}
 or the survey \cite{Edelsbrunner-2000}.

A linear bound for nonobtuse triangulation of point sets was given 
by Bern, Eppstein and Gilbert  \cite{BEG-1994},
and  Bern and Eppstein \cite{BE92}  gave a quadratic 
bound for simple polygons with holes 
(this is a polygonal region  where every boundary component
 is a simple closed curve or an 
isolated point).  Bern, Dobkin and Eppstein \cite{BDE} improved this 
to $O(n^{1.85})$ for convex domains.
Bern,   S.  Mitchell  and  Ruppert
  \cite{BMR95}   gave a  $O(n)$  algorithm for nonobtuse 
triangulation of simple polygons with holes   in 1994 and 
their construction uses  only right triangles.
We shall make use of their result in this paper.
These and related results are discussed in the surveys 
 \cite{BE-OptMesh} and  \cite{Bern-Plassmann}.
Other papers that deal with algorithms for finding nonobtuse 
and acute triangulations include 
  \cite{Erten-Ungor-2007},
\cite{KS2004},
\cite{Li-Zhang},
and \cite{MS92}.
Giving a polynomial upper  bound for the
complexity of  
nonobtuse triangulation  of PSLGs has remained open 
(e.g., see Problem 3 of  \cite{BE-OptMesh}).
We give such a bound by proving:

\begin{thm} \label{NonObtuse} 
Every PSLG with $n$ vertices 
has a  $O(n^{2.5})$ conforming  nonobtuse triangulation.
\end{thm}

Maehara \cite{Maehara-2002} showed that any
 nonobtuse triangulation 
using $N$  triangles can be refined to an 
acute triangulation (all angles $< 90^\circ$) 
with $O(N)$ elements. 
 A different proof  was given by 
Yuan  \cite{Yuan-2005}.
In our  proof of Theorem \ref{NonObtuse}   the triangulation 
will consist of all right triangles, but the arguments of 
Maehara or Yuan then show the theorem also holds with 
an acute triangulation, at the cost of a
larger constant in the $O(n^{2.5})$. 
As noted above, simple examples
 give a quadratic lower bound for PSLGs (see \cite{BE92}),
so a gap remains between 
our  upper  bound and the worst known  example.
However, this gap can be eliminated in some special cases, e.g., 

\begin{thm} \label{Refine Triangulation}
A triangulation of a simple $n$-gon has a
$O(n^2)$ nonobtuse  refinement.
\end{thm}

This improves a $O(n^4)$ bound given by Bern and Eppstein 
\cite{BE92}. 
 We can also approach the quadratic lower bound 
if we consider ``almost nonobtuse'' triangulations:

\begin{thm} \label{Triangles} 
Suppose $\theta>0$.
Every PSLG with $n$ vertices has a  conforming
 triangulation   with 
$O(n^2/\theta^2)$  elements and all angles 
$\leq 90^\circ + \theta$.
\end{thm}

This improves a result of S. Mitchell  \cite{Mitchell93}
with upper angle bound 
$ \frac 78 \pi = 157.5^\circ$ and a result of Tan \cite{Tan96}   with
$\frac {11}{15} \pi = 132^\circ$. 

A triangulation is called {\defit Delaunay} if whenever two triangles
share an edge $e$, the two angles opposite $e$  sum to $180^\circ$
or less.  If all the triangles are  nonobtuse, then this is certainly the
case, so Theorem \ref{NonObtuse} immediate implies

\begin{cor} \label{Delaunay} 
Every PSLG with $n$ vertices 
 has a $O(n^{2.5})$  conforming  Delaunay triangulation. 
\end{cor}

This improves a  1993 result of Edelsbrunner and Tan \cite{ET93}
that any PSLG has a conforming Delaunay
 triangulation of size $O(n^3)$.
Conforming Delaunay triangulations for $\Gamma$
  are also called Delaunay
refinements of $\Gamma$. There are numerous papers discussing
Delaunay refinements including
\cite{Edelsbrunner-2000},
\cite{Erten-Ungor-2009},
\cite{NS-94},
\cite{Ruppert-1993},
  \cite{Saalfeld1991}
and
 \cite{Shewchuk-2002}.
The argument in this paper does not seem to 
give a better estimate for Delaunay triangulations
than for nonobtuse triangulations, nor does the proof 
appear  to simplify in the Delaunay case. 
Finding an improvement (either for the estimate 
or the argument) in the Delaunay case would be 
extremely interesting.

An alternative formulation of the Delaunay condition is that 
every edge in the triangulation is the chord of an open disk 
that contains no vertices of the triangulation. We say the 
triangulation is {\defit Gabriel} if every edge is the diameter
of such disk. It is easy to check that a nonobtuse triangulation 
must be Gabriel, so we also obtain a stronger version of the 
previous corollary: 

\begin{cor} \label{Gabriel} 
Every PSLG with $n$ vertices 
 has a $O(n^{2.5})$  conforming  Gabriel triangulation. 
\end{cor} 

Given a finite planar point set $V$,  and a point 
$ v\in V$, the {\defit Voronoi cell} corresponding to $v$ is
the open set of points that are strictly closer to 
$v$ than to any other point of $V$. The union of the
boundaries of the all the Voronoi cells is called
the {\defit Voronoi diagram} of $V$.
In \cite{Salzberg}, it is shown that given a nonobtuse 
triangulation with $N$ elements, one can find a set of $O(N)$
points so that the Voronoi diagram of the point set covers
all the edges of the triangulation. Thus we obtain 

\begin{cor} \label{learning} 
For every PSLG $\Gamma$ with $n$ vertices, there is a point 
set $S$ of size $O(n^{2.5})$ whose Voronoi diagram covers $\Gamma$.
\end{cor} 

The authors of \cite{Salzberg} were interested in a type of 
machine learning called ``nearest neighbor learning''.  
Given $\Gamma$ and $S$ as in the corollary, and any point 
$z$ in the plane,  
we can decide which complementary component
  of $\Gamma$  $z$ belongs to by finding the 
element $w$  of $S$ that is closest to $z$;  $z$  and $w$ must
 belong to the same complementary  component of $\Gamma$.
 Thus the corollary 
says that  a partition of the 
plane by a PSLG of size $n$ can be ``learned'' from a point 
set of size $O(n^{2.5})$.  This answers a question from 
 \cite{Salzberg}  asking if a polynomial number of points 
always suffices. 


Acute and nonobtuse   triangulations arise in a variety 
of other contexts.  In recreational 
mathematics one asks for the smallest 
triangulation of a given object into acute
or nonobtuse pieces. For example, a square can obviously be 
meshed with two right triangles, but less obvious is the 
fact that it  can be acutely  triangulated 
with  eight elements but not seven; see \cite{Cassidy-Lord}.
 For further results of this type see
 \cite{Gardner60a}, 
\cite{Gardner60b}, 
\cite{GM60}, 
\cite{HIZ}, 
\cite{Itoh-2001},
 \cite{Itoh-Yuan}, 
 \cite{Itoh-Z-2},
\cite{Itoh-Z},
 \cite{Saraf}, 
\cite{Yuan-Z}, 
\cite{Yuan-Z-2}, 
  \cite{Z-2004}, 
the 2002  survey \cite{Zamfirescu} and 
the 2010 survey \cite{CZsurvey}.
There is less known in higher dimensions, but 
recent work has shown 
there is an acute triangulation of $\reals^3$, 
 but no acute triangulation
of $\reals^n$, $n \geq 4$ 
\cite{BSKV}, 
\cite{KPP}, 
\cite{Krizek2006}, 
 \cite{VaHiGuRa2007}, 
 \cite{VaHiGuRa2010},
\cite{VaHiZhGu2009}.
Finding polynomial bounds for conforming   Delaunay tetrahedral meshes 
in higher dimensions remains open.

In various numerical methods involving meshes,
a nonobtuse  triangulation
frequently gives  simpler and better behaved  algorithms.
For example, in \cite{Vavasis-96} 
Vavasis bounds   various matrix norms  arising from
the finite element method in terms of the number $n$
 of triangulation elements; for general
triangulations his estimate is exponential in $n$, but for nonobtuse
triangulations it is only linear  in $n$.
Other examples where nonobtuse or Delaunay  triangulations give
simpler or faster methods include: 
 discrete maximum principles 
\cite{Bobenko-Springborn-2007}, 
\cite{Ciarlet-Raviart},
\cite{Vanselow};
 Stieltjes matrices in finite element methods 
\cite{BHV}, 
\cite{Spielman-Teng-2004}; 
 convenient description of the dual graph \cite{Bern-Gilbert-92};
 the Hamilton-Jacobi equation  \cite{Barth-Sethian}; 
 the fast marching method  \cite{Sethian}; 
 the tent pitcher algorithm for meshing space-time 
\cite{Spacetime}, 
\cite{Thite},
\cite{Ungor-Sheffer}.

The ideas  in this paper are used in a companion paper
 \cite{Bishop-quadPSLG}
to obtain conforming quadrilateral meshes for PSLGs 
that have  optimal angle
bounds  and optimal worst case complexity.
The precise statement from this paper used in 
 \cite{Bishop-quadPSLG} is Lemma \ref{quad mesh lemma}; 
this follows from a slight modification of 
 the proof of Theorem \ref{Triangles}. 
The result obtained in \cite{Bishop-quadPSLG} says
that  every PSLG has an $O(n^2) $  conforming  quadrilateral mesh   with
all angles $\leq 120^\circ$ and
all new  angles $ \geq 60^\circ$.
The angle bounds and
quadratic complexity bound are both sharp.




Many thanks to Joe Mitchell and Estie Arkin for numerous 
conversations about computational geometry in general and 
the results of this paper in particular. Also thanks  to two
anonymous referees for many helpful comments and suggestions
on two earlier versions of the paper; their efforts 
greatly improved the presentation  in this version.

In Section \ref{Gabriel edges} we recall a theorem of 
Bern, Mitchell and Ruppert \cite{BMR95} that connects
nonobtuse triangulation to finding Gabriel edges, and 
we sketch the proof of this result in Section \ref{ProofBMR}.
In Section \ref{Refine tri} we use their theorem to give a simple 
proof of Theorem \ref{Refine Triangulation}. In 
Sections \ref{standard}-\ref{return regions} 
we discuss propagation paths, 
dissections and return regions; these are used in 
the proofs of both Theorems \ref{NonObtuse} and  \ref{Triangles}.  
Sections \ref{ideal}--\ref{proof spirals} 
give the proof of Theorem 
\ref{Triangles} and Section \ref{quad mesh lemma sec} 
summarizes facts 
from the proof that are used in the sequel paper 
\cite{Bishop-quadPSLG}. Section \ref{intro NonObtuse} gives an 
overview of the proof of Theorem \ref{NonObtuse}
and Sections \ref{Perturb sec}--\ref{empty} provide the details. 


\section{ The theorem of Bern, Mitchell and Ruppert } \label{Gabriel edges}

Given a point set $V$ and two points $v,w \in V$, the segment
$vw$ is called a {\defit Delaunay edge} if it is the chord of some
open disk that contains no points of $V$ and is called a
{\defit Gabriel edge} if it is the diameter of such a disk
(see \cite{GS}). We will call a PSLG  $\Gamma $ with vertex
set $V$ and edge set $ E$  Gabriel  if 
every edge in $E$  is  Gabriel for $V$.
Given a PSLG which is not Gabriel, can we always add
extra vertices to the edges, making a new PSLG that is 
Gabriel? 
We are particularly interested in the case when $P=T$
is a triangle.  See Figure \ref{GabrielTriangle}. 

\begin{figure}[htb]
\centerline{
\includegraphics[height=1.5in]{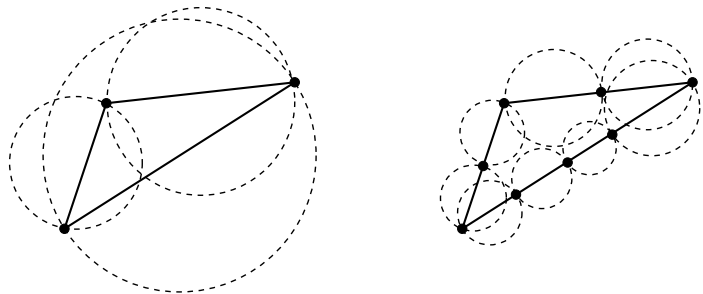}
 }
\caption{\label{GabrielTriangle}
A triangle that is not Gabriel and one way to add points 
so that it becomes Gabriel.
}
\end{figure}

The connection between Gabriel triangles and
 nonobtuse triangulation 
is given by the following result of Bern, Mitchell and 
Ruppert \cite{BMR95}. Suppose $T$ is a triangle with 
three vertices $V_0= \{a,b,c\}$  and suppose $V$ is a
finite  subset of the edges of $T$. Then $T\setminus 
(V_0 \cup V)$ is a finite union of segments and 
we let $V'$ denote the midpoints of these segments. 

\begin{thm} \label{BMR}
Suppose we add $N$ points $V$ to the edges of a triangle $T$,
 so that the triangle becomes Gabriel. 
Assume further that no point of the interior of
$T$ is in more that two of the Gabriel disks.
Then the interior  
of $T$  has a nonobtuse triangulation consisting of $O(N)$
right triangles and the triangulation vertices
on the boundary of   $T$ 
are exactly the vertices of $T$ and  the points in 
$V$ and $V'$.
\end{thm} 

This follows from the proof in \cite{BMR95}, but 
this precise statement does not appear there,
so in the next section we will briefly describe 
now to deduce Theorem \ref{BMR} from the arguments 
in \cite{BMR95}. 

We say that  one triangulation ${\cal T}_1$
 is a refinement of another
triangulation ${\cal T}_2$, 
if each triangle in  ${\cal T}_2$  is a union of triangles in 
${\cal T}_1$. 
If $\Gamma$ is a PSLG that is a triangulation and we 
add enough points to the edges of $\Gamma$ to make 
every triangle Gabriel, then the resulting nonobtuse 
refinements of each triangle agree along any common
edges  (the set of boundary vertices is $e \cap (V \cup V')$
for both triangles with edge $e$). Thus we get:

\begin{cor}
Suppose that   $\Gamma$ is a   planar triangulation with 
$n$ elements, and that we can add $N$ vertices to the 
edges of $\Gamma$ so  that every triangle becomes 
Gabriel. Then $\Gamma$ has  a refinement 
consisting of  $O(n+N)$   right triangles.
\end{cor} 

It is fairly easy to see that we can always add a 
finite number of vertices and make each triangle 
Gabriel. Thus nonobtuse  refinement of 
a triangulation is always 
possible, but the difficulty is to bound $N$ in terms 
of $n$.
Since any PSLG with $n$ vertices can be triangulated 
using $O(n)$ triangles (and the same vertex set),
Theorem \ref{NonObtuse} is reduced to

\begin{thm} \label{ETS Gabriel}
Given any triangulation with $n$ elements we can add $O(n^{2.5})$
vertices  to the edges  so that every  triangle becomes Gabriel.
\end{thm} 

We  will prove this for general planar triangulations
 later in the paper. We start with the 
simpler case of triangulations of simple polygons
(Theorem \ref{Refine Triangulation}) in Section \ref{Refine tri}.
The key feature of this special case is that the 
elements  of such a triangulation form a tree under 
adjacency (sharing an edge); this fails for general 
triangulations, and dealing with this failure is 
the main goal of this paper.

In Theorem \ref{ETS Gabriel}  we only need to 
check that each triangle becomes Gabriel, not that  
the point set is Gabriel for the whole triangulation. 
The difference is that in the first case if we take 
a disk  $D$ with an edge $e$ as its diameter, we only have to 
make sure that $D$ does not contain any vertices belonging 
to the two triangles that have $e$ on their boundary. 
In the second case, we would have to check that $D$ does 
not contain any vertices at all. This 
apparently stronger condition  follows from the 
weaker  ``triangles-only'' condition, but we don't need 
this fact for the proof of Theorem \ref{NonObtuse}.

\section{Sketch of  the proof of Theorem \ref{BMR} }
 \label{ProofBMR}

As noted above, Theorem \ref{BMR} is  due to
Bern, Mitchell and Ruppert in \cite{BMR95},
but the precise result is not given in that paper. 
For the convenience 
of the reader, we briefly sketch how our 
statement is deduced using the argument in 
\cite{BMR95}.  

The basic idea in \cite{BMR95} is to pack the 
interior of a polygon $P$ (which in our case 
is just a triangle $T$)   with disks until the
remaining region is a union of pieces that are
each bounded by three or four circular arcs,
or a segment lying on the polygon's  boundary. 
These are called  the remainder regions.

Each  remainder 
region $R$ is associated to a simple polygon 
$R^+$, called the augmented region of $R$  as 
follows. Each circular arc in the boundary 
of $R$ lies on some circle,and we add to $R$ the 
sector of the circle subtended by this arc. 
Doing this for each boundary arc of $R$ gives 
the polygon $R^+$. 
See Figure \ref{ProofBMR0}.

\begin{figure}[htb]
\centerline{
\includegraphics[height=2.0in]{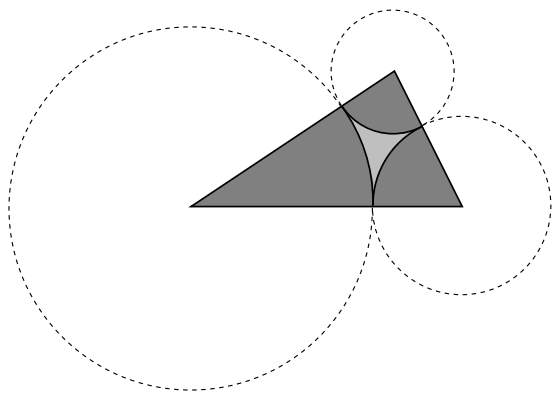}
 }
\caption{\label{ProofBMR0}
The light gray is a 3-sided remainder region  $R$ 
and the dark gray is the corresponding augmented
region $R^+$.
}
\end{figure}

The augmented regions 
decompose the original polygon into simple
polygons  and the authors of \cite{BMR95} show 
how each augmented region can be   meshed 
with right triangles.
Moreover, the mesh  of $R^+$ only has vertices 
at the vertices of $R^+$ (the centers of the 
circles), the endpoints of the boundary arcs
(the tangent points between disks)
 or on the straight line segments 
that lie on the boundary of $P=T$. 

We modify their construction first  placing the Gabriel 
disks along the edge of the triangle, as shown 
in Figure \ref{ProofBMR1}.  The disk packing 
construction of \cite{BMR95} will only be applied 
to the part of the triangle outside the Gabriel 
disks, hence none of the remainder regions 
that are formed will have straight line 
boundary segments on the boundary of $T$.

\begin{figure}[htb]
\centerline{
\includegraphics[height=2.0in]{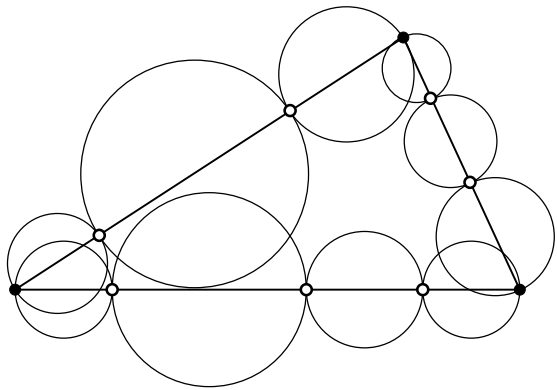}
 }
\caption{\label{ProofBMR1}
The white points are the set $V$ that makes the 
triangle $T$ Gabriel.  The Gabriel disks may 
intersect in pairs, but we assume that  no 
three of them intersect.
}
\end{figure}

Whenever two Gabriel disks overlap, we place a small
disk  $D$ near each of the intersection points of their 
boundaries. The disk $D$   is tangent
to both overlapping disks and its interior is 
disjoint from all the Gabriel disks. 
See Figures \ref{ProofBMR2} and \ref{ProofBMR3} 
for two situations where this can occur: the 
boundaries of the Gabriel disks (which we call 
the Gabriel circles) either  have two intersections 
inside $T$ or they have just one intersection inside 
$T$. 

Figure \ref{ProofBMR2} shows  what happens 
when the boundaries intersect at two points. 
We place two disks near the intersections points 
(the dashed disks in the figure), and we form 
a quadrilateral by connecting the centers of the 
four circles. This quadrilateral is then meshed 
with 16 right triangles as shown in the figure.
The overall structure is shown on the left, and
an enlargement is shown on the right. The black 
point on the right where six triangles meet is
the common intersection point of three lines: 
the line  $L_0$ through the two intersection points 
of the Gabriel circles, and the two tangent 
lines $l_1, l_2$  between the dashed disk and the the 
Gabriel disks. It is proven in \cite{BMR95} that
these three lines meet at a single point, 
as shown.  The center of the dashed circle in 
Figure \ref{ProofBMR2} is not necessarily on $L_0$ 
(although it appears this way in the figure).
Figure \ref{ProofBMR3} shows the case when two 
Gabriel circles have one intersection inside $T$ 
(and the other is at vertex of $T$).
The proof that all the triangles
are right is fairly evident (see \cite{BMR95}). 

\begin{figure}[htb]
\centerline{
\includegraphics[height=2.3in]{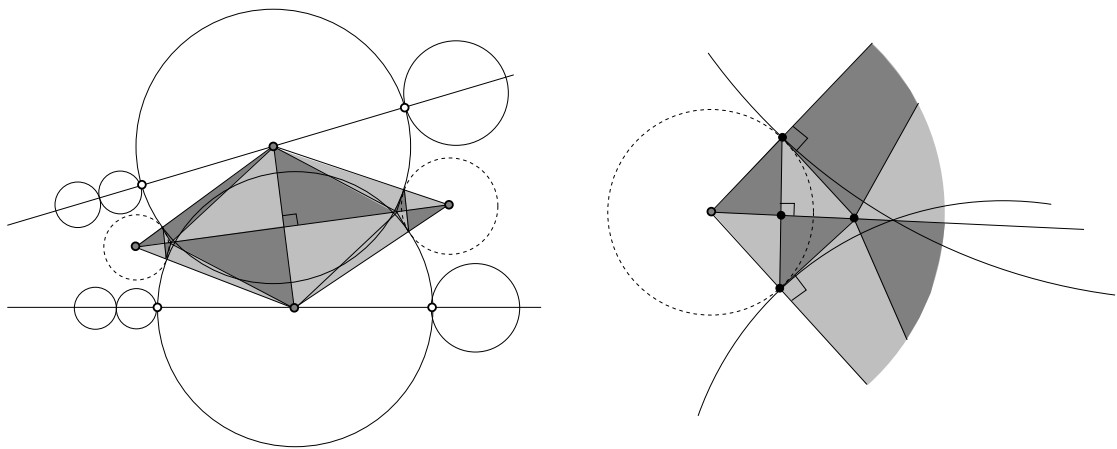}
 }
\caption{\label{ProofBMR2}
Gabriel disks whose boundaries intersect twice
inside $T$. We add tangent disks as shown (dashed)
and triangulate as shown. The picture on the right 
shows more detail near one of the added disks. 
The point where six triangles meet is the 
intersection of the line through the intersection 
points of the Gabriel circles and the two tangent 
lines between the added disk and the Gabriel disks.
}
\end{figure}

\begin{figure}[htb]
\centerline{
\includegraphics[height=3.0in]{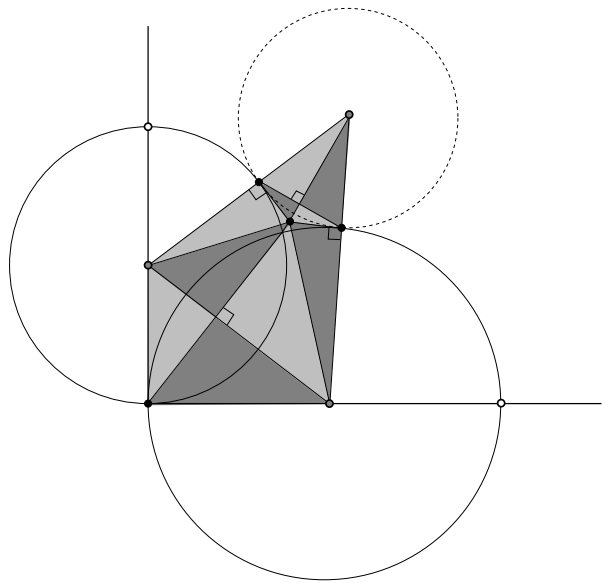}
 }
\caption{\label{ProofBMR3}
Figure \ref{ProofBMR3} is analogous to 
Figure \ref{ProofBMR2}, but shows that case 
when two Gabriel circles have a single 
intersection inside $T$ (the other is 
at a vertex of $T$). 
In this case, 
the mesh is a sub-picture of the previous 
case and only uses 10 right triangles. 
}
\end{figure}

From this point the proof follows \cite{BMR95}. 
Lemma 1 of the paper shows that  we can 
add disjoint disks until all the remainder 
regions have three or four sides and that
the number of disks added is comparable to 
the number of Gabriel disks we started with.
See Figure \ref{ProofBMR4}. 

\begin{figure}[htb]
\centerline{
\includegraphics[height=2.5in]{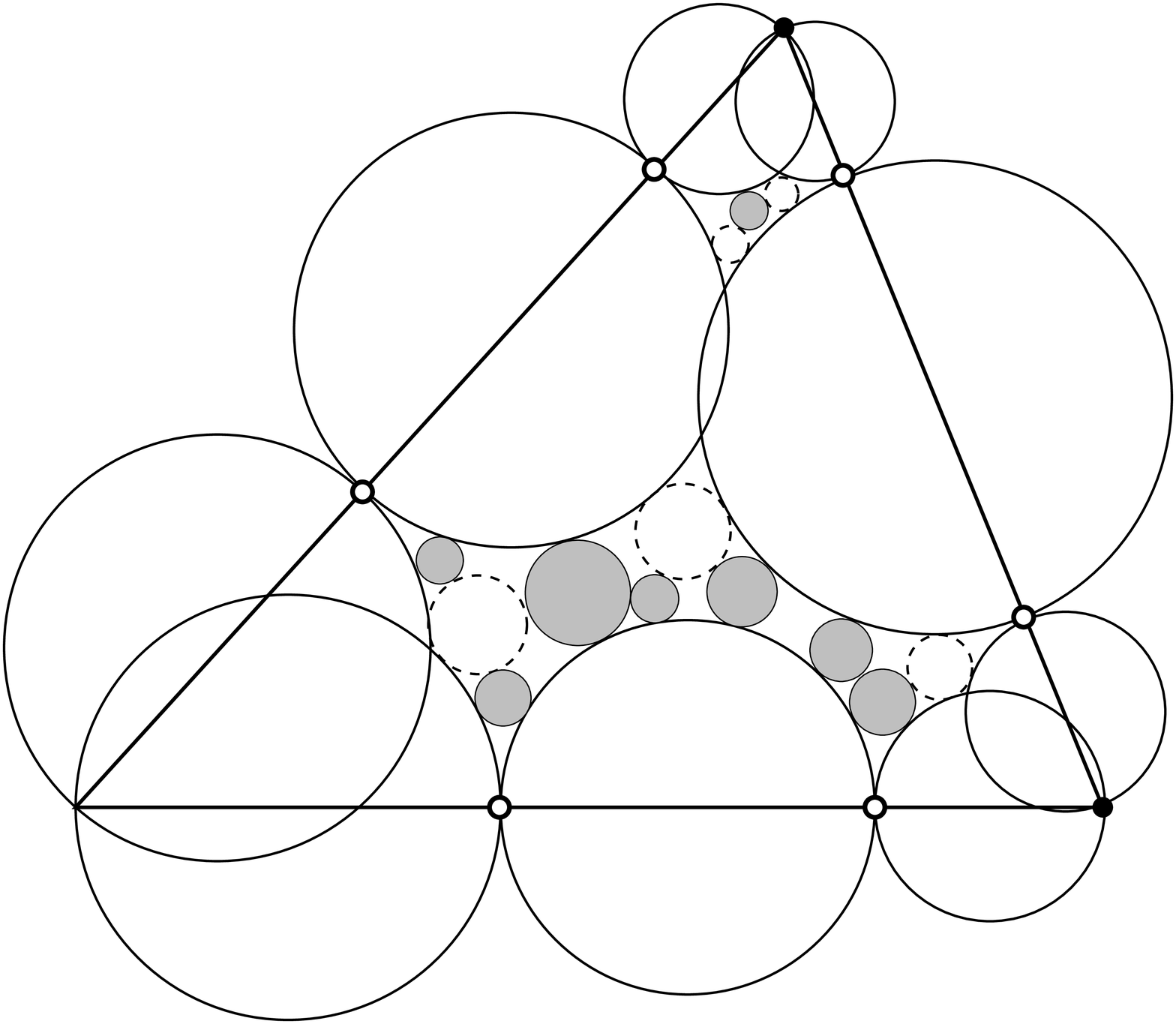}
 }
\caption{\label{ProofBMR4}
After adding the disks  (dashed) tangent to the 
intersecting Gabriel disks, we pack the remaining
region with disjoint disks (gray) until only 
circular arc regions with three or four sides 
remain. This step and the rest of the proof 
follow the proof in \cite{BMR95} exactly. 
}
\end{figure}

Because of the Gabriel disks, none of the 
remainder regions have  straight line 
boundary arcs, so all the augmented regions 
are meshed by right triangles whose vertices 
are either interior to $T$ or lie in the set 
$V$  or in the set $V'$ (the centers of the 
Gabriel disks) and every such point is used. 
This gives Theorem \ref{BMR}.

\section{ Proof of Theorem \ref{Refine Triangulation} } \label{Refine tri}

A PSLG $\gamma $ is a {\defit simple polygon}  if 
it is  a simple  closed Jordan curve.
 Suppose  $\gamma$ is a simply polygon
and $\{ T_k\}_1^N$ is a 
triangulation of $\gamma$  with no Steiner points. 
The union of the edges and vertices of the triangulation 
is the PSLG $\Gamma$ in Theorem \ref{Refine Triangulation}.
 For each triangle 
$T_k$, let $C_k$ be the inscribed circle and let $T_k' \subset 
T_k$ be the triangle with vertices at the three points where 
$C_k$ is tangent to $T_k$. Note that the arcs between 
points of $C_k$ all have angle measure $< \pi$, so 
 $T_k'$ must be acute.
  These  vertices of $T_k'$  will sometimes  be called 
the {\defit cusp points} of $T_k$.
See Figures \ref{Inscribed} and \ref{Inscribed0}. 

\begin{figure}[htb]
\centerline{
\includegraphics[height=1.25in]{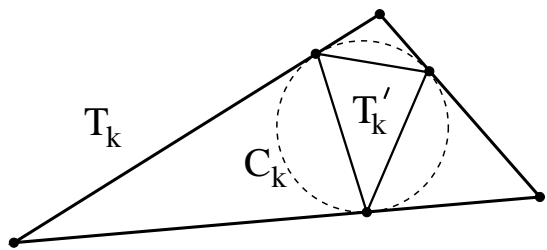}
 }
\caption{\label{Inscribed}
The definition of $C_k$ and $T_k'$. 
}
\end{figure}

\begin{figure}[htb]
\centerline{
\includegraphics[height=2in]{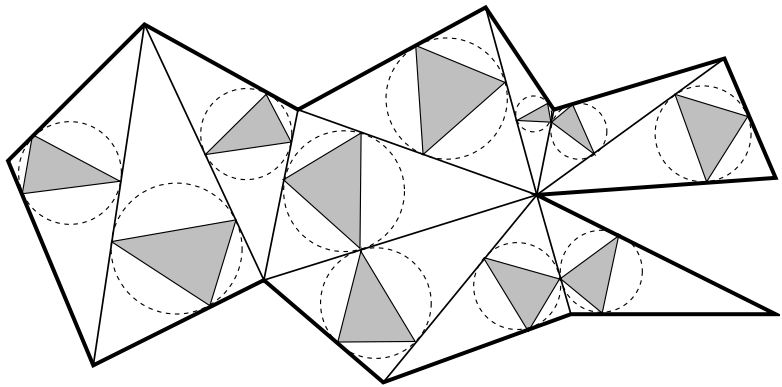}
 }
\caption{\label{Inscribed0}
A triangulation $\{T_k\}$ of a simple polygon,  the inscribed circles 
(dashed) and  triangles $\{ T_k'\}$ (shaded).
}
\end{figure}

Also note that $T_k \setminus T_k'$ consists 
of three isosceles triangles, each with its base
as one side of $T_k'$ and its opposite vertex 
a vertex of $T_k$. Foliate each isosceles triangle
with segments that are parallel to its base
(foliate simply means to write a region as a disjoint union 
of curves).
We call these {\defit P-segments} (since they are 
``parallel''  to the base and they will   also  allow us 
to ``propagate'' certain points  through the triangulation). 

Given a vertex  $v$ of some $T_k'$, this point is either
on $\gamma$ (the boundary of the triangulation) or it is 
on the side of some other triangulation element 
$T_j$, $j \ne k$.  In the first case do nothing. 
In the second case, either $v$ is also a vertex 
of $T_j'$ or it is not. In the first case, again 
do nothing. In the second case, build a polygonal 
path whose first edge is the $P$-segment in $T_j$ 
that has $v$ as one endpoint. The other endpoint 
is a point $w$ on a different side of $T_j$. 
If $w$ is on  $\gamma$, or is a vertex of some $T_i'$, 
$i \ne j$, then end the polygonal arc at $w$. 
Otherwise, $w$ is on the side of some third 
triangle $T_i$ and we can add the $P$-segment 
in $T_i$ that has one endpoint at $w$.
See Figures \ref{Inscribed17} and \ref{Inscribed3}.

\begin{figure}[htb]
\centerline{
\includegraphics[height=2in]{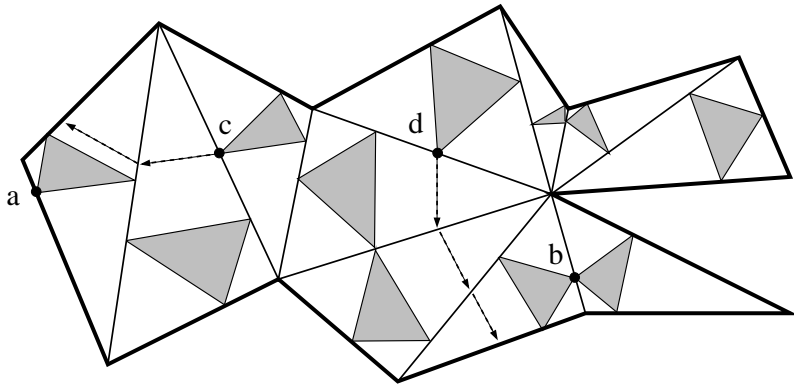}
 }
\caption{\label{Inscribed17}
A cusp point (vertex of a shaded triangle)
 can either be 
on $\gamma$ (the boundary of the triangulation), 
can be a vertex of another shaded triangle,  or 
neither. In the last case we can ``propagate'' the 
vertex using a $P$-segment and continue the process
until one the first two conditions holds. 
Vertices $a$ and $b$ represent the first two 
possibilities and vertices $c$ and $d$ both represent 
the third case; the dashed lines show how these points
propagate until they hit the boundary.
}
\end{figure}

\begin{figure}[htb]
\centerline{
\includegraphics[height=2in]{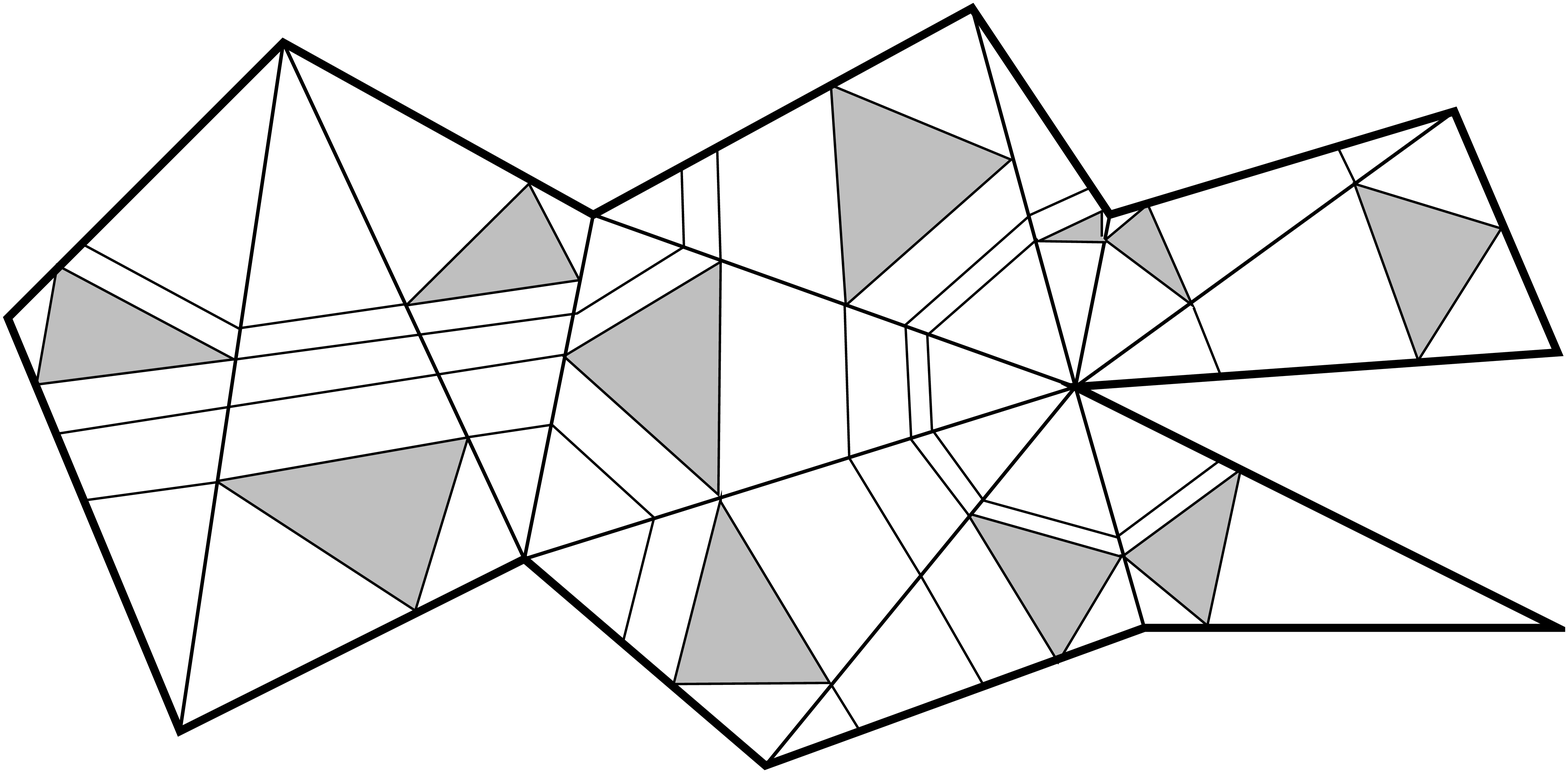}
 }
\caption{\label{Inscribed3}
Here is the same triangulation as in Figure 
\ref{Inscribed17} with all cusp points propagated
until they terminate.
This refines the triangulation
into acute  triangles (shaded) and isosceles
triangles  and trapezoids (white). We claim 
the new vertices make all the original triangles 
Gabriel.
}
\end{figure}

We continue in this way, adding $P$-segments to 
our polygonal path until we either reach a point 
on the boundary of the triangulation or hit a 
point that is the vertex of some triangle $T_m'$.
We call the path formed by adjoining $P$-segments
a {\defit $P$-path}. 
Since our triangles come from a triangulation 
of a simple polygon, they form a tree under 
edge-adjacency and so the  $P$-paths 
starting at the three vertices of 
$T_k'$ must  cross  distinct triangles and hence 
can  use at most $n-1$  segments altogether.
Thus every $P$-path must terminate and 
 all the $P$-paths formed by starting 
at all vertices of all the $\{T_k'\}$  
can create at most $n(n-1)$ new vertices altogether.

\begin{lemma} 
The set $U$ of vertices created
by these $P$-paths crossing edges of $\{T_k\}$
makes every  triangle Gabriel. 
\end{lemma} 

\begin{proof} 
Note that $U$ 
contains every vertex of every $T_k'$ and we also include 
in $U$ the all vertices of  all the $T_k$'s.
To prove the lemma, 
consider a segment $e$ that is a connected 
component of $T_k \setminus U$ (so $e$ is 
a diameter of one of our Gabriel disks).
Since $U$ contains the vertices of $T_k'$, 
$e$  lies on a  non-base side of  one of
 the three  isosceles triangles inside $T_k$.
If we reflect $e$ over the  symmetry axis
of this isosceles triangle we get an edge $e'$ 
on the other non-base side. Moreover, $e'$ 
is also one of the edges created by adding $U$ 
to the triangulation, and the disks $D$, $D'$ 
 with diameters 
$e$ and $e'$ respectively, 
 are reflections of each other through 
the symmetry line of the isosceles triangle.
Thus the boundaries of $D$ and $D'$ intersect, 
if at all, on the line of symmetry, so $D\cap e' 
\subset e'$ and $D'\cap e \subset e$. 
In particular, the open disk $D$ does not 
contain the endpoints of $e'$ and vice versa.
This implies $D$ can't contain any of the 
points of $U$ that lie on the same side of 
$T_k$ as $e'$. Clearly $D$ does not contain 
any points of $U$ on the side of $T_k$
containing $e'$. Finally, $D$ does not 
contain any points of $U $ that are on 
the third side of $T_k$ because $D$ is contained 
in the disk  $D''$,  centered at the vertex of 
$T_k$  where the sides containing $e$ 
and $e'$ intersect and passing through two
vertices of $T_k'$, and this disk does not hit
the third side of $T_k$. See Figure  \ref{Reflect}.   
\end{proof}

\begin{figure}[htb]
\centerline{
\includegraphics[height=1.75in]{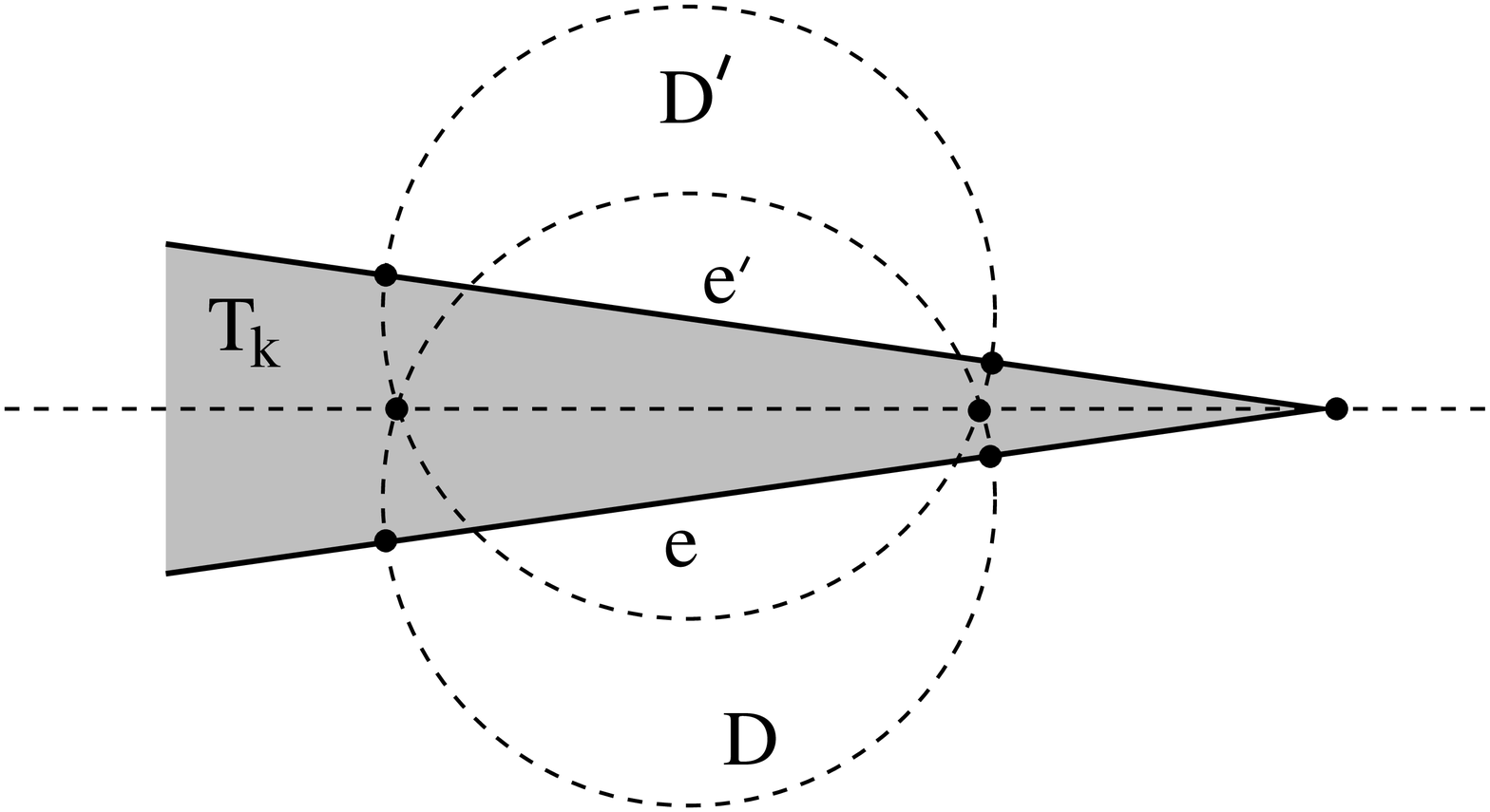}
$\hphantom{}$
\includegraphics[height=1.75in]{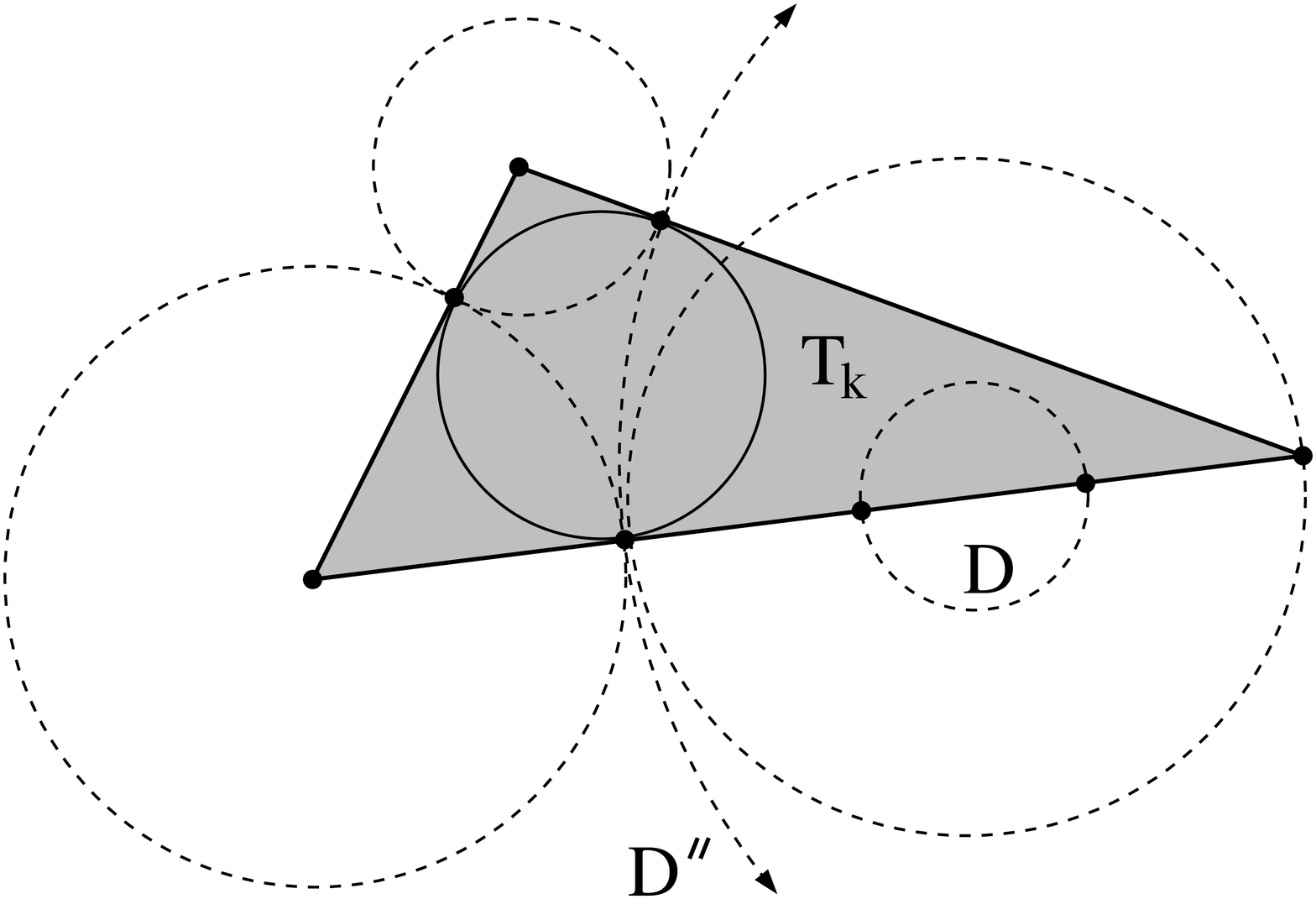}
 }
\caption{\label{Reflect}
The points where propagation paths cross 
triangle edges define Gabriel edges in each 
triangle. Note that at most two Gabriel disks 
can intersect when the cusp points are included
in  set of endpoints (the set $V$ in Theorem 
\ref{BMR}).
}
\end{figure}

This argument also shows that two Gabriel  disks 
can only intersect if they lie on different non-base 
sides of one of the three isosceles sub-triangles. 
Thus no three disks can intersect and the condition 
in Theorem \ref{BMR} is automatically satisfied whenever
the set $V$ contains the vertices of $T_k'$ for 
every $k$.

Thus for each triangle $T_k$, the points 
$U \cap T_k$  makes $T_k$  Gabriel. 
By Theorem \ref{BMR},
$T_K$ has a nonobtuse triangulation with only these 
boundary vertices and using $O(\#(U \cap T_k))$
triangles. These triangulations fit together to 
form a nonobtuse refinement of the original 
triangulation of size $O(\#(U)+n) = O(n^2)$, 
which  proves Theorem \ref{Refine Triangulation}. 

We can make a slight improvement to the algorithm 
above. As we propagated each vertex, 
we could have stopped
whenever the path encountered any isosceles triangle 
with angle $\geq 90^\circ$. In this case, the 
Gabriel condition will be satisfied no matter how
we add points to a non-base sides of the isosceles triangle, 
since the corresponding disks don't intersect the 
other non-base side of the triangle. See Figure 
\ref{Bigger90}. In some cases this might lead to 
a smaller nonobtuse triangulation.  This observation 
will also be used later in the proof of Theorem 
\ref{NonObtuse}, when it will be convenient to 
assume we are dealing only with isosceles triangles
that are acute.

\begin{figure}[htb]
\centerline{
\includegraphics[height=1.75in]{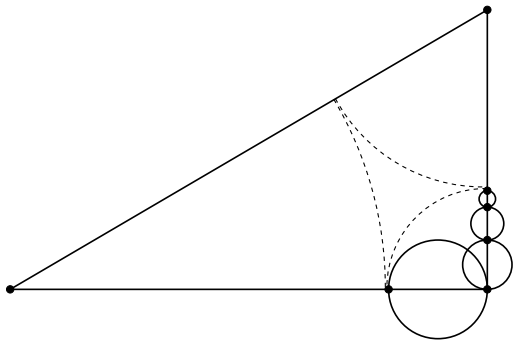}
 }
\caption{\label{Bigger90}
$P$-paths can be stopped when they hit an 
isosceles triangle with angle $\geq 90^\circ$
since the corresponding Gabriel disks can't 
hit the other sides of the triangle.
}
\end{figure}

\section{Dissections and  quadrilateral propagation} \label{standard} 

We now start to prepare for the proofs of Theorems 
\ref{NonObtuse} and \ref{Triangles}. The definitions 
and results in this and the next three sections 
will be used in both proofs.

Suppose $\Omega$ is a domain in the plane (an open connected set). 
We say $\Omega$ has a polygonal  {\defit dissection} if there 
are a finite number of simple polygons (called the {\defit pieces} of 
the dissection) whose interiors 
are  disjoint   and contained in $\Omega$ and 
so that the union of their closures covers the closure
of $\Omega$.  A {\defit mesh} is a dissection where any 
two of the closed polygonal pieces are either  (1) disjoint or 
(2) intersect in a point that is a vertex for both pieces or 
(3) intersect in a line segment that is an edge for both pieces.
See Figure 
\ref{Diss2Mesh2} for an example. A dissection is also 
called a {\defit non-conforming mesh}. A vertex of 
one dissection  piece that lies on the interior of
 an edge for another 
piece is called a {\defit non-conforming vertex}. If there 
are  no such vertices, then the dissection is actually a 
mesh. 

\begin{figure}[htb]
\centerline{
\includegraphics[height=1.7in]{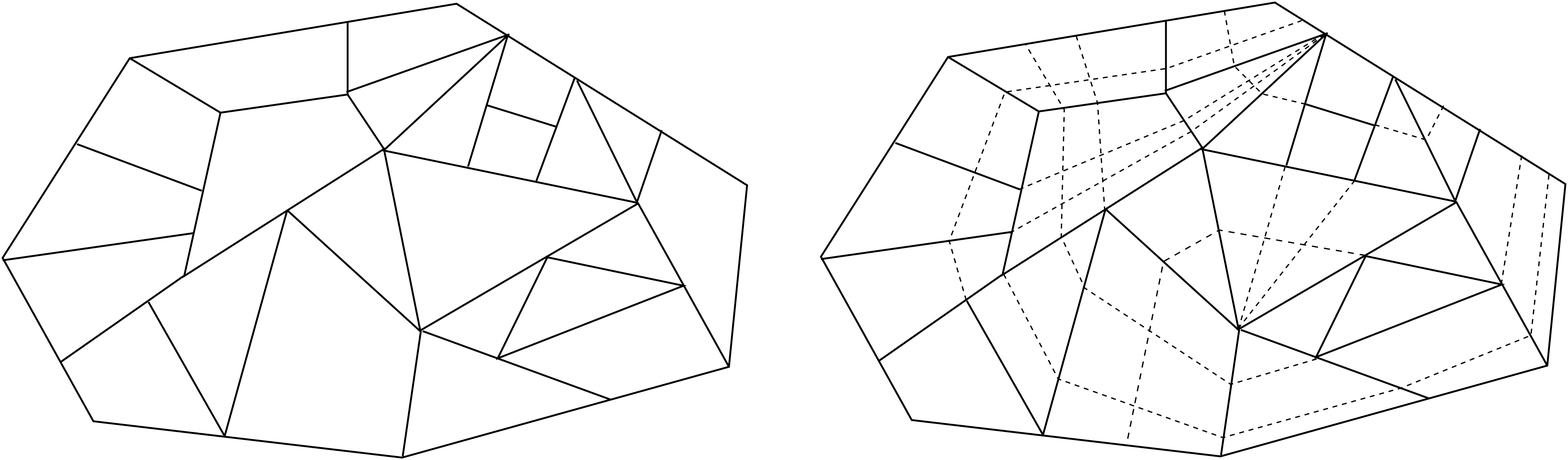}
 }
\caption{\label{Diss2Mesh2}
 On the left is a 
polygon dissected into quadrilaterals and triangles and 
on the right is the standard propagation of the non-conforming 
vertices until the propagation paths leave the polygon. 
If all the paths terminate, this gives a mesh, as described 
in the text.
}
\end{figure}

Given any convex quadrilateral with
vertices $a,b,c,d$  (say in counterclockwise  order), there is a unique affine
map from $[a,b]$ to $[c,d]$ that takes $a$ to $d$ and $b$ to
$c$. A propagation segment in the quadrilateral is a
segment connecting a point in $[a,b]$ to its affine image point
in $[c,d]$ (or connecting a point in $[b,c]$ to its 
affine  image
in $[d,a]$ under the analogous map for that pair of sides).
See Figure \ref{StandardProp}.  In a triangle $A,B,C$ with 
marked vertex, say $A$, propagation paths either connect 
points on $[A,B]$ to their linear images on $[A,C]$ or 
they connect any point on $[B, C]$ to the single point $A$
(this is what we would get if we think of the triangle 
as a degenerate quadrilateral $A,B,C,D$ with $A=D$, 
i.e., one side of length zero).

Given $\theta >0$, we say a quadrilateral is
{\defit  $\theta$-nice} if all the angles
are within $\theta$ of $90^\circ$. 
In this paper we will always assume $\theta < 90^\circ$
so the quadrilateral is convex.
We say a triangle with a marked vertex  
is $\theta$-nice if  the two unmarked vertices have 
angles that are 
within $ \theta$ of $90^\circ$.

\begin{figure}[htb]
\centerline{
\includegraphics[height=1.5in]{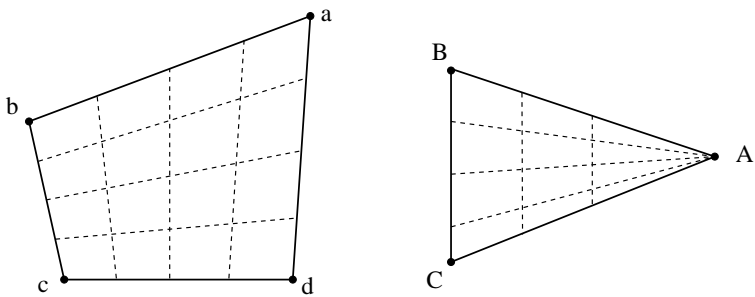}
 }
\caption{\label{StandardProp}
Standard propagation segments for a quadrilateral 
and a triangle with a marked vertex (A).
  In both cases, ``$\theta$-niceness''
is preserved by cutting a piece into sub-pieces by 
such segments.
}
\end{figure}

\begin{lemma} \label{split quad}
Suppose $\theta < 90^\circ$ and that 
$Q$ is a $\theta$-nice  quadrilateral.
If $Q$ is  sub-divided by a  propagation line, then each
of the resulting sub-quadrilaterals is also $\theta$-nice.
\end{lemma}

\begin{proof}
Set $a_t = (1-t)a+tb$ and $c_t = 
(1-t)d + t c$. Let $I_t=[a_t,c_t]$ be the
segment connecting these points and let  $\theta(t)$ be the angle
formed by the segments $[a,b]$ and $I_t$.
It suffices to show this function is monotone in
$t$.  If it were not monotone, then there would be
two distinct values of $s,t \in [0,1]$ where
$I_s$ and $I_t$ were parallel. Because both
endpoints move linearly in $t$,  this implies 
$I_r$  is parallel
to $I_s$ for all $s \leq r \leq t$. Because $\theta(t)$
is analytic in $t$, this means it is constant on
$[0,1]$. Thus $\theta$ is either strictly monotone
or is constant; in either case it is monotone,
as desired.
\end{proof}

Similarly (but more obviously),
 when a $\theta$-nice triangle is cut
 by such a propagation  segment (of either 
type) the resulting pieces are  $\theta$-nice quadrilaterals 
or $\theta$-nice  triangles.

\begin{lemma} \label{make mesh}
Suppose $\Omega$ has a dissection into $\theta$-nice 
pieces (triangles and quadrilaterals). Suppose that 
every non-conforming vertex  can be  propagated  
so that it reaches 
the boundary of $\Omega$ or hits another vertex after 
a finite number of steps. 
Then the  resulting  paths cut the $\theta$-nice
dissection pieces  into $\theta$-nice triangles and 
quadrilaterals that mesh $\Omega$.
\end{lemma} 

The proof is evident since when we are
finished, there are no vertices that are in the interior of 
any edge of any  piece.
See Figure \ref{Diss2Mesh2}. 
 What is not so clear is whether, in general,  
the propagation paths have to end; in the proof 
of Theorem \ref{Refine Triangulation} every propagation 
path did end within a fixed number of steps, but 
in general, the paths may never terminate (see Figure \ref{Irrational})
or may only terminate only after a huge number of steps. 
Later in this paper we discuss two ways of ``bending'' the standard
propagation paths so that they terminate within a 
certain  number of steps, and so that the pieces formed satisfy 
certain  geometric conditions.

\begin{figure}[htb]
\centerline{
\includegraphics[height=2.0in]{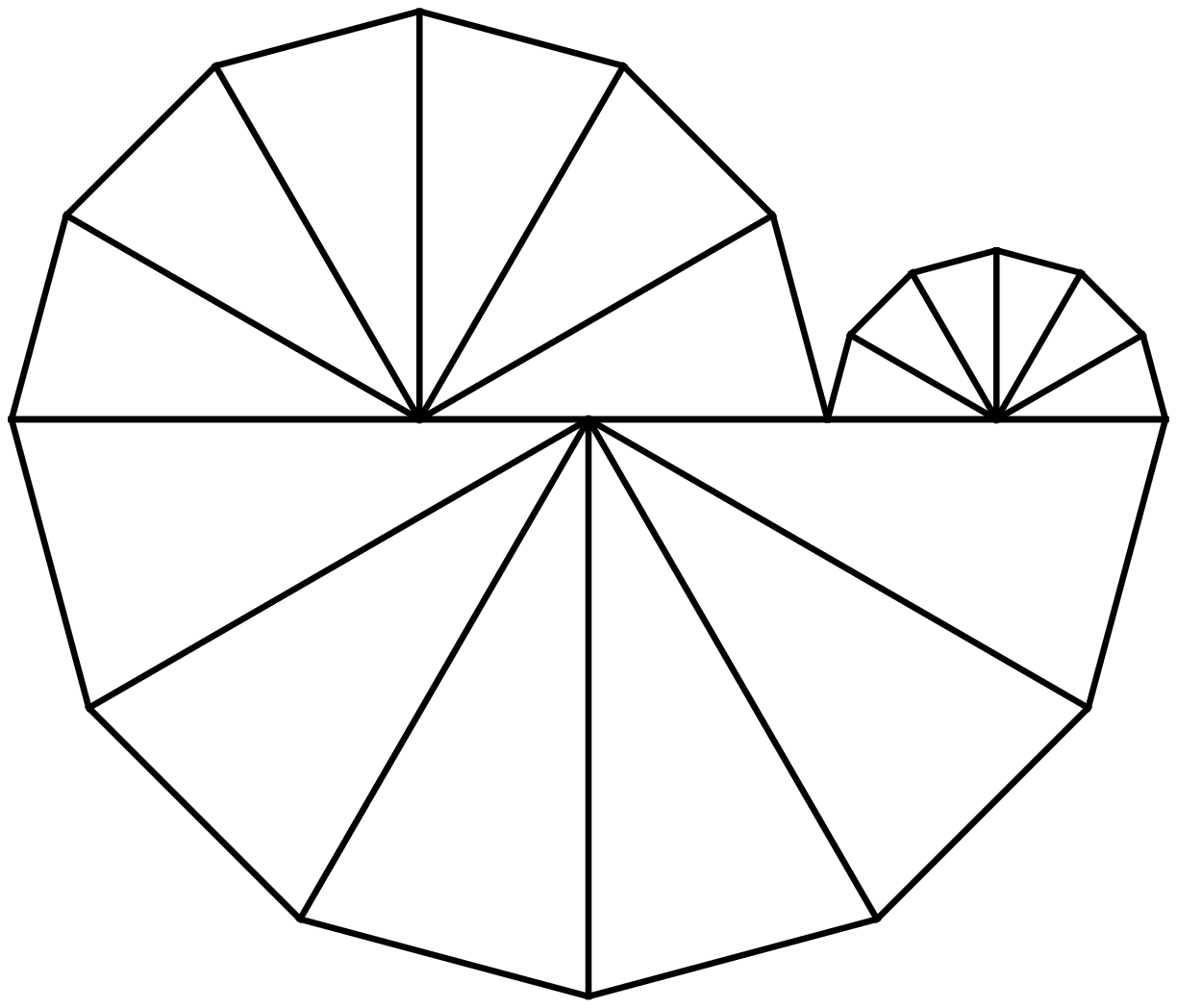}
$\hphantom{xxx}$
\includegraphics[height=2.0in]{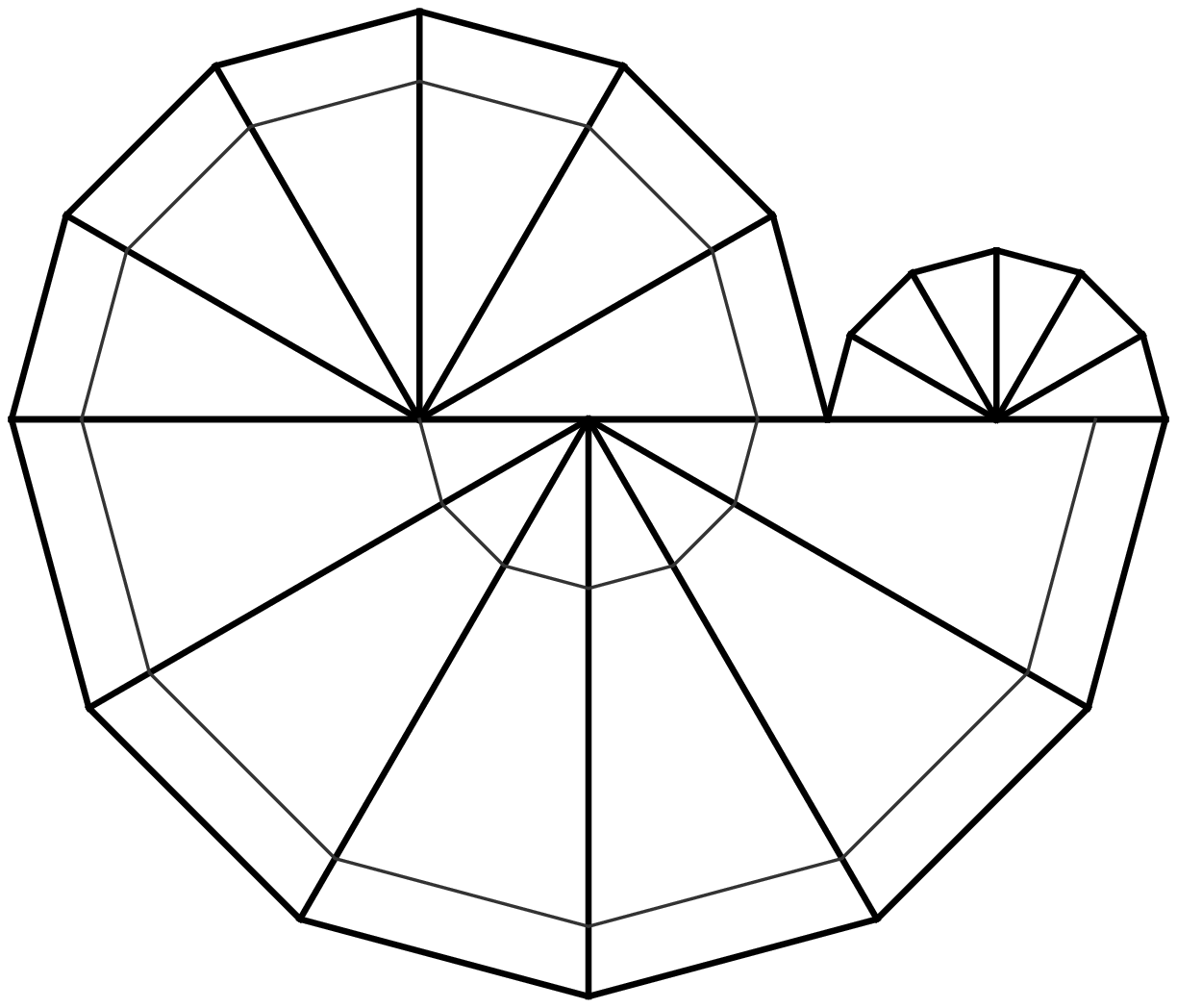}
 }
\vskip.1in
\centerline{
\includegraphics[height=2.0in]{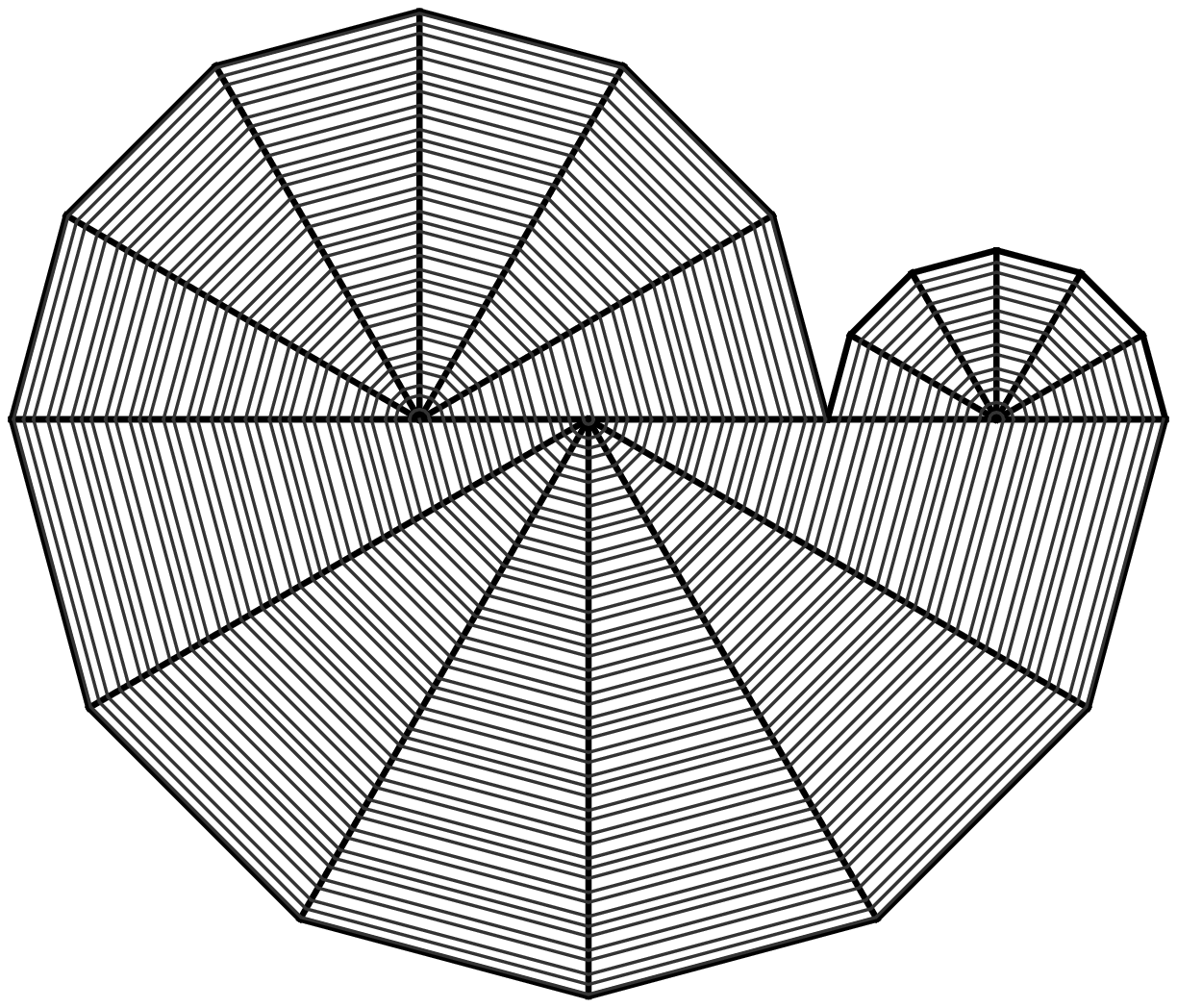}
$\hphantom{xxx}$
\includegraphics[height=2.0in]{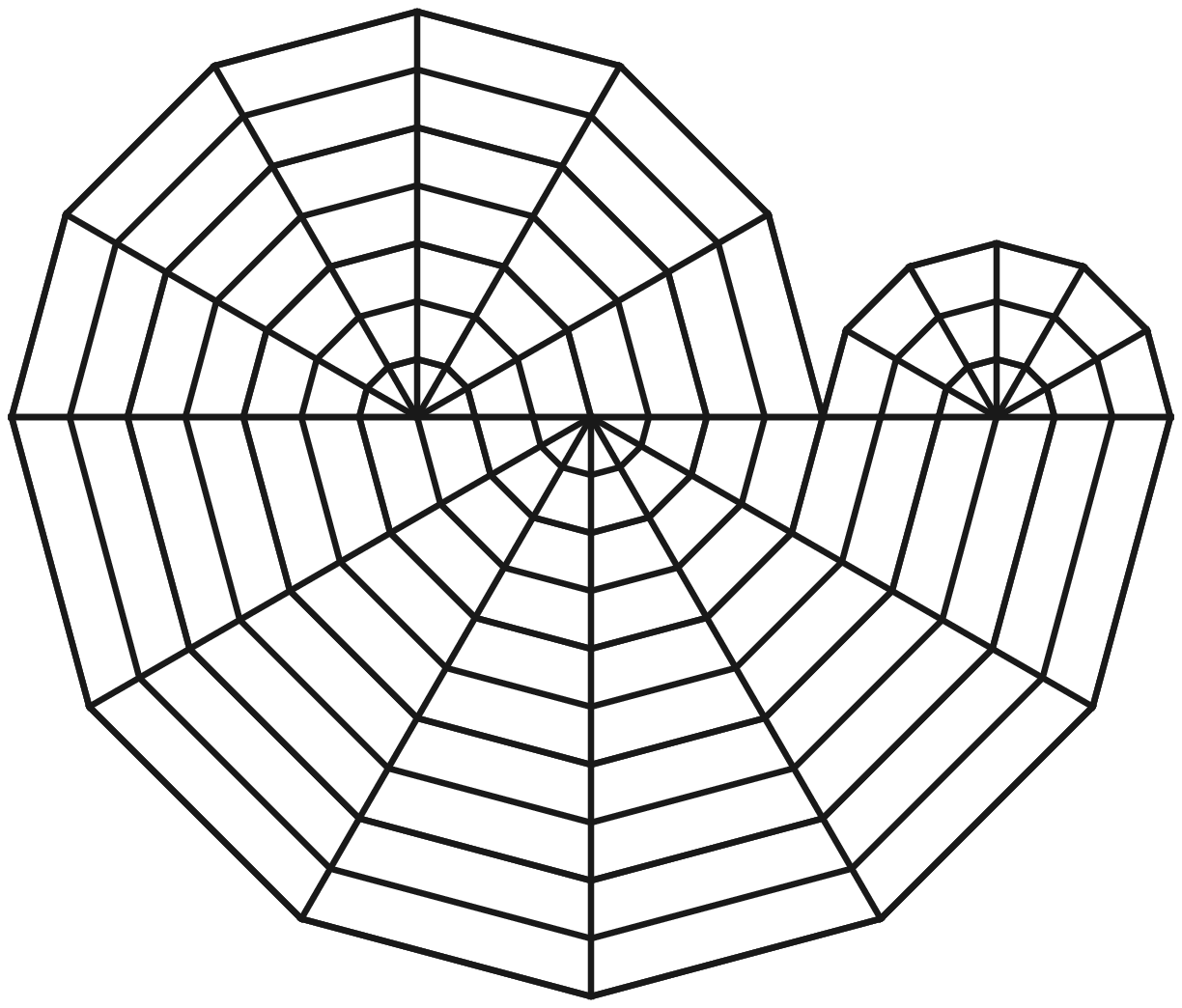}
 }
\caption{\label{Irrational}
The upper lefts shows 
an  isosceles dissection of a region. The horizontal 
segment is $S=[0,1]$ and the six vertices shown on this 
segment are (left to right) $0, \frac \alpha 2, \frac 12, 
\alpha, \frac{1+\alpha}2, 1$.
A path started at a point $x \in S$ will visit $x-\alpha \mod 1$ 
after propagating once through the upper half-plane and 
once through the lower half-plane.  The  upper right 
and lower left pictures show 
a path after 2 visits to the lower half-plane  
and 50 visits for $\alpha = 1/\sqrt{2}$.
If $\alpha$ is irrational, then the  propagation path 
becomes dense in the dissected region.
When $\alpha$ is rational, the propagation paths
either connect non-conforming vertices or are
loops;  the connecting paths  give  a mesh but there is no uniform
 bound on the number of elements. 
The connecting paths for  $\alpha= .7$ are shown at lower right;
the resulting tubes are filled with $P$-loops.
}
\end{figure}

\section{Isosceles dissections} \label{iscoceles dissections}

Next we  discuss a  special  type of polygonal dissection.
An {\defit isosceles triangle} is a triangle $T$ with a marked 
vertex $v$  so that the  two sides  adjacent to $v$  have 
equal length.
An equilateral triangle can be considered as isosceles in 
three ways, but we assume that if such triangles occur,  a 
vertex is specified.

 The side opposite $v$ is called the {\defit base} of 
$T$ and the other two sides are the {\defit non-base} sides 
of $T$.  The {\defit angle} of an isosceles triangle 
will always refer to the interior angle at the vertex 
opposite the base edge. 
 A {\defit $P$-segment} is a segment in $T$ with 
endpoints on the non-base sides that is parallel to the base.
This is a special case of the propagation segments 
for marked triangles in the  previous section.
We require the interior of the segment to 
be in the interior of $T$, so the base itself is not 
a $P$-segment.
We say a triangle is {\defit $\theta$-nice} if all its angles are 
bounded above by $90^\circ + \theta$.

An {\defit isosceles trapezoid}  is a quadrilateral that has 
a line of symmetry that bisects opposite sides. This is 
equivalent to saying that  there is at least one pair
of parallel sides (called the base sides) that have the 
same perpendicular bisector and the other pair 
of sides (the non-base sides) 
have the same length as each other.
We allow rectangles, but in this case we specify a pair of
opposite sides as the base sides.
The {\defit angle} of an isosceles trapezoid is the angle made
by the lines that contain the two non-base sides; we 
take this to be zero if these sides are parallel (when 
the trapezoid is a rectangle). The {\defit vertex} of the 
trapezoid is the point  where these same lines intersect 
(in the case when they are not parallel; otherwise we say 
the vertex is at $\infty$). See Figure \ref{DefnTrap}.
We say a quadrilateral  is {\defit $\theta$-nice}
 if all its  interior angles are between 
 $90^\circ - \theta$ and $90^\circ + \theta$ (inclusive).
This is the same as saying the angle of the trapezoid
is $\leq \theta$.

\begin{figure}[htbp]
\centerline{
 \includegraphics[height=1.5in]{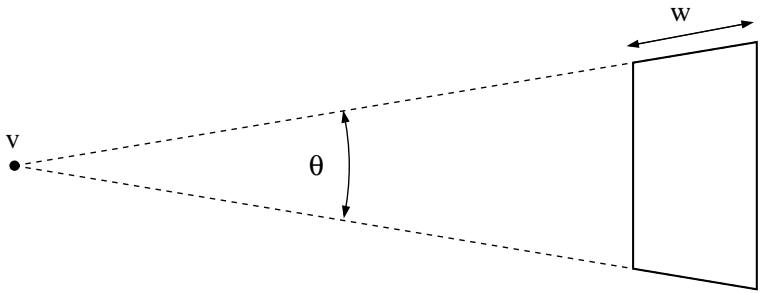}
 }
\caption{ \label{DefnTrap}
An isosceles trapezoid. The base sides are vertical 
in this picture. The vertex is the point $v$ and 
the width $w$ is the length of the non-base sides.
The angle of the trapezoid is $\theta$.
}
\end{figure}

As with isosceles triangles, we can define $P$-segments in 
an isosceles trapezoid   as segments  in the trapezoid 
that are  parallel to the base sides (again, these correspond
to propagation segments for quadrilaterals).
A {\defit $P$-path} is a simple polygonal arc formed by adjoining 
$P$-segments end-to-end.

An {\defit isosceles piece} is either an isosceles triangle 
or an isosceles trapezoid. We will use this term when 
it is unimportant which type of shape it is. When referring 
to a base side of an isosceles piece we mean either
base side for a trapezoid and the base side or
the opposite vertex for a triangle; the vertex is 
considered as a segment of length $0$, so when we refer 
to  the length of the smaller base side of an 
isosceles piece, we mean zero if the piece is a triangle.

 A {\defit $Q$-segment} in an isosceles 
triangle is a segment joining a point of the base to the
vertex opposite the base. The two non-base sides do 
count as $Q$-segments, and we shall also call these 
the {\defit $Q$-sides} of the isosceles triangle.
A $Q$-segment for an isosceles trapezoid is a
propagation segment that connects the base sides 
of the trapezoids. As with triangles, we count 
the non-base sides as $Q$-segments and call
these the $Q$-sides of the trapezoid.
The {\defit width} of the piece is the length 
of a $Q$-side (both $Q$ sides have the same length).
It might be more natural to define the width as the distance
between the base sides, but the definition as given
will simplify matters when we later join 
isosceles pieces to form tubes.

Suppose that $\Omega$ is a domain in the plane (an open 
connected set).  As might be expected, 
an {\defit isosceles dissection} of $\Omega$
is a  finite collection  of disjoint, open  isosceles triangles and 
trapezoids contained in $\Omega$, so that the union of
their closures covers all of $\Omega$. 
However, we also require that when two pieces have 
sides with non-trivial intersection, these sides are 
both $Q$-sides. We do this so that in an isosceles 
dissection, a $P$-path can always be continued 
unless the path hits a vertex of the dissection, or 
hits the boundary of the dissected region.
A {\defit $\theta$-isosceles dissection} is an isosceles 
dissection where every piece is $\theta$-nice.

For example,  Figure \ref{Inscribed0}   in 
Section \ref{Refine tri}  shows a triangulated polygon.
  Let $\Omega$ be the 
part of interior  of the polygon 
with the  (closed) shaded triangles removed; 
the remaining white region 
it is a union of isosceles triangles  that only meet along 
non-base sides. Thus $\Omega$ has an isosceles dissection;  note
that it is not a mesh since the isosceles triangles do not 
always meet along full edges. When we remove the $P$-paths 
generated by propagating the vertices of the central triangles, 
we obtain a mesh into isosceles triangles and trapezoids 
(as required by Lemma \ref{make mesh}).
See Figure \ref{Inscribed3}.
We call this an {\defit isosceles mesh}.

Figure \ref{Irrational}, upper left,  shows an 
isosceles dissection of a region using 18 triangles.
In that example, if $\alpha $ is irrational, then 
the  $P$-paths  never hit the boundary
of the region and can continue forever without terminating.
For $\alpha$ rational the $P$-paths  starting at 
non-conforming vertices terminate at other non-conforming 
vertices and these paths create an isosceles mesh. However, 
the number of mesh elements depends on the choice 
of $\alpha$ and may be arbitrarily large.

In Figure \ref{PropPath4}, we show an isosceles 
dissection using only trapezoids, and an isosceles
 mesh generated
by propagating non-conforming vertices along   $P$-paths.

\begin{figure}[htbp]
\centerline{
 \includegraphics[height=2.0in]{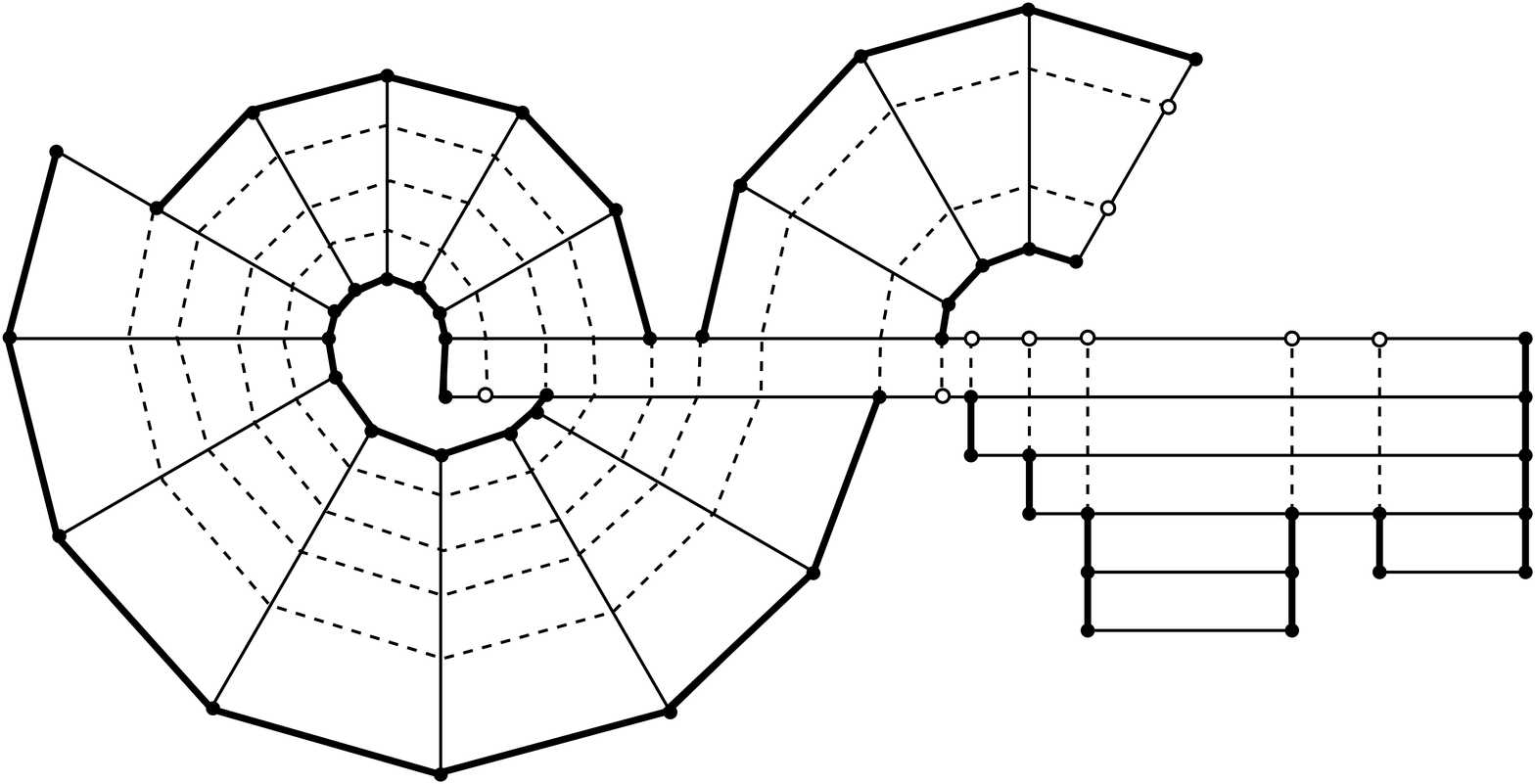}
 }
\caption{ \label{PropPath4}
A domain $W$ dissected by isosceles trapezoids and a
mesh generated  by propagating non-conforming 
vertices along   $P$-paths.
The $P$-sides of the trapezoids are drawn thicker.
}
\end{figure}

A {\defit chain}  in a dissection is a maximal collection
of distinct  pieces $T_1, \dots, T_k$ so that for
 $j=1, \dots k-1$,
$T_j$ and $T_{j+1}$ share a $Q$-side (the sides
are identical, not just overlapping).
If a  piece in the dissection does not share 
$Q$-side with any other piece, we consider it as a chain
of length one.
 For example, the dissection in
Figure \ref{PropPath4} has chains of length 2,4, 5 and 7
and four chains of length 1.
The {\defit $Q$-ends} of a chain are the 
$Q$-side of $T_1$ not shared with $T_2$, and the
$Q$-side of $T_k$ not shared with $T_{k-1}$.
When $T_1$ and $T_k$ also share
a  $Q$-side, then  the chain forms a closed loop. 
We will call this a {\defit closed chain} (this case
is not of much interest to us since no propagation
paths will ever occur inside such a closed chain.)

Suppose $T$ is an isosceles piece. 
Given $\theta >0$,  a {\defit $\theta$-segment} is a
segment in $T$ with one endpoint on each non-base side, 
and so that the segment is 
within angle $\theta$ of being parallel to the base. 
We allow one endpoint of a $\theta$-segment   to be 
a vertex of $T$, but the interior of the segment must 
be contained in the interior of $T$, so we don't consider 
base sides of $T$ to be $\theta$-segments.  A {\defit $\theta$-path}
is a polygonal arc made up of $\theta$-segments joined 
end-to-end. We shall sometimes refer to this as a 
{\defit $\theta$-bent path}. 
Note that if $T$ is an $\theta$-nice isosceles piece that is cut 
by a $P$-segment, then it is cut into two $\theta$-nice 
isosceles pieces. If it is cut by a $\theta$-segment, then 
we get two pieces (triangles or quadrilaterals) that are 
$2\theta$-nice (but not  isosceles unless $\theta =0$). 
See Figure \ref{BentPath1}.

\begin{figure}[htb]
\centerline{
\includegraphics[height=1.75in]{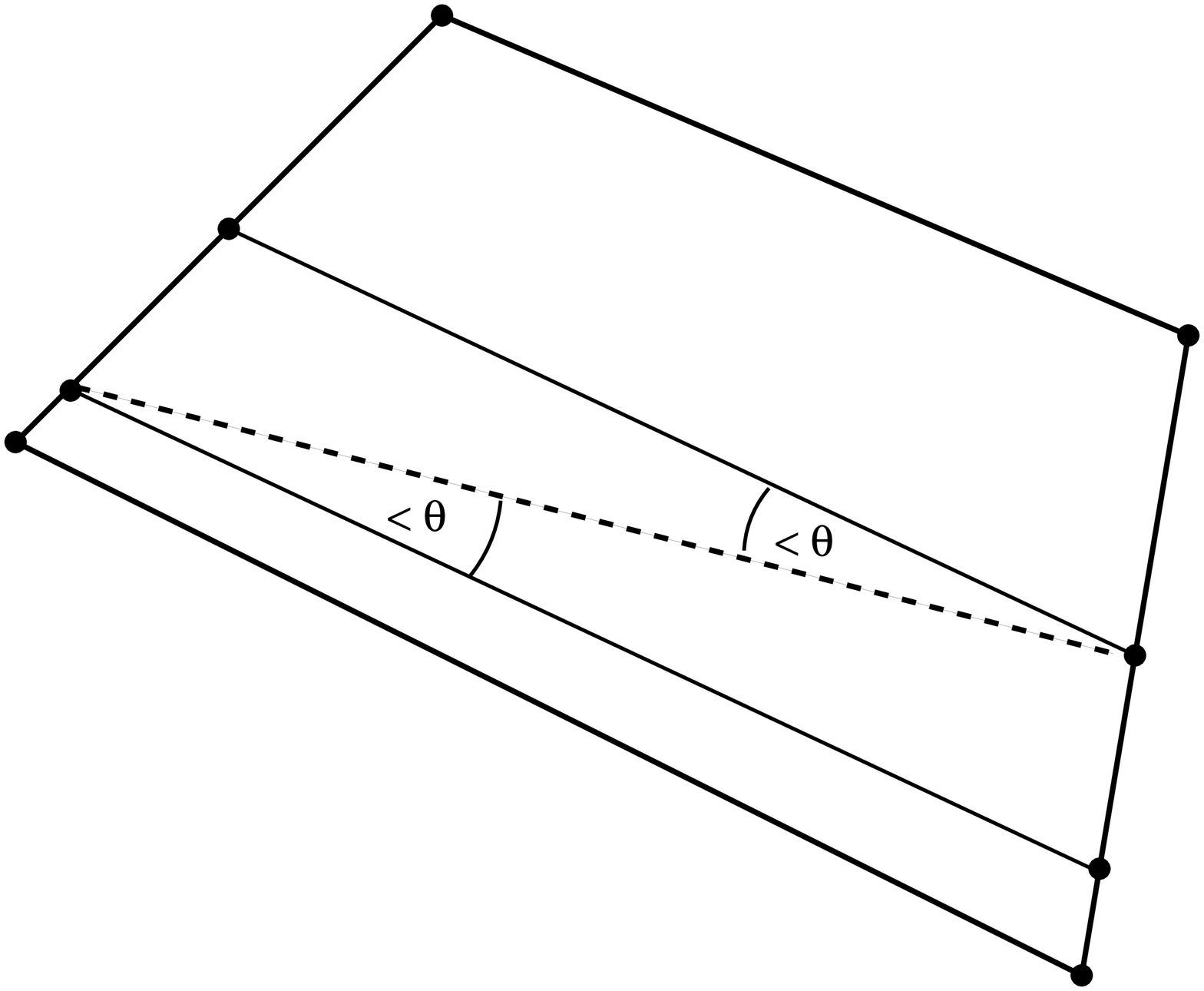}
$\hphantom{xxxxxx}$
\includegraphics[height=1.75in]{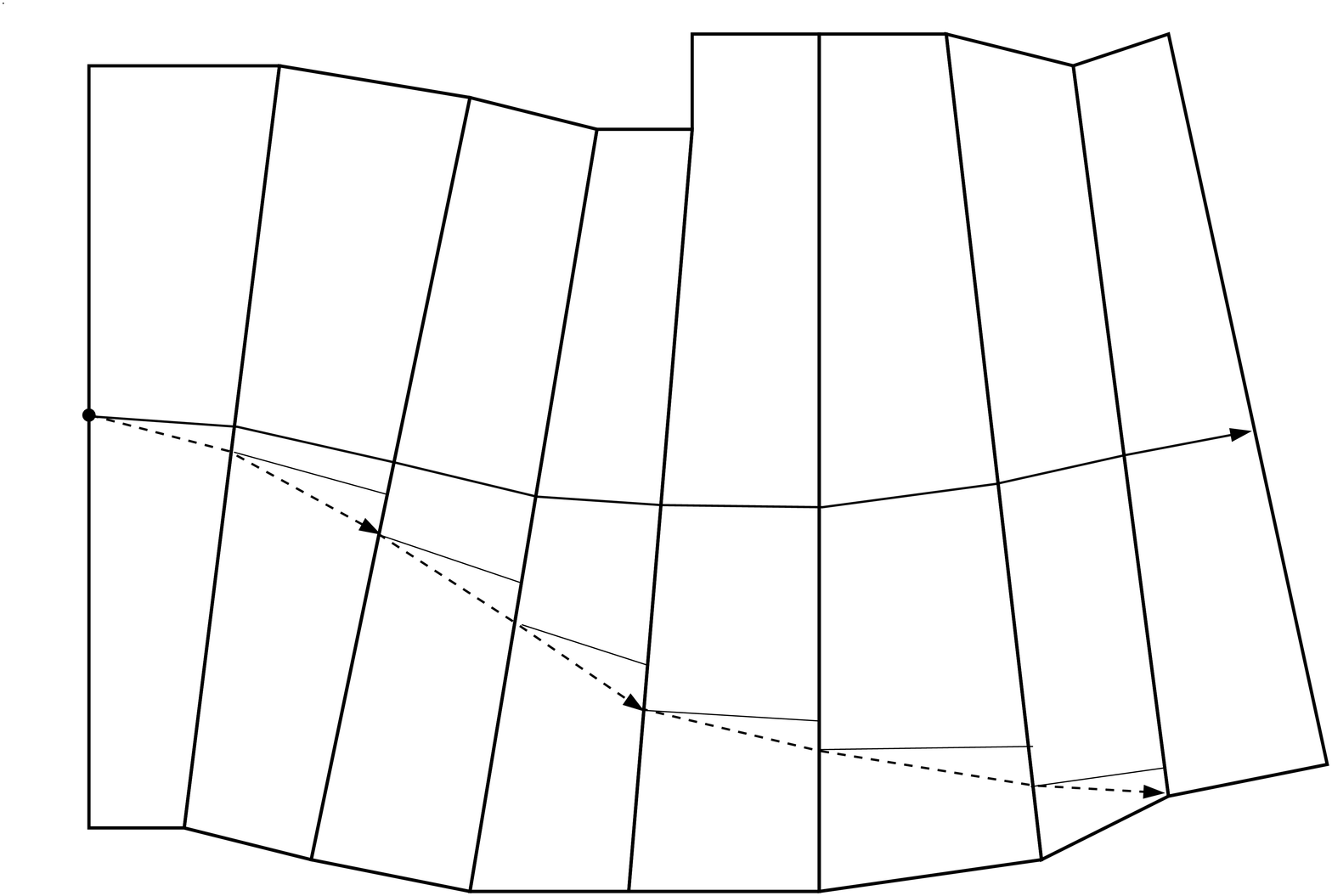}
 }
\caption{\label{BentPath1}
A $\theta$-segment (dashed)  makes angle at most $\theta$ 
with $P$-segments (solid). A $\theta$-path is made up out 
of $\theta$-segments.
}
\end{figure}

We say that a finite set of points on the $Q$-sides of an isosceles
piece make that piece Gabriel if the following holds. 
Each $Q$ side is split into several segments by these
points and we require that the open disks with these segments 
as diameters  do not contain any of the added points or 
corners of the piece.
See Figure \ref{GabrielPieces}.

\begin{figure}[htb]
\centerline{
\includegraphics[height=1.5in]{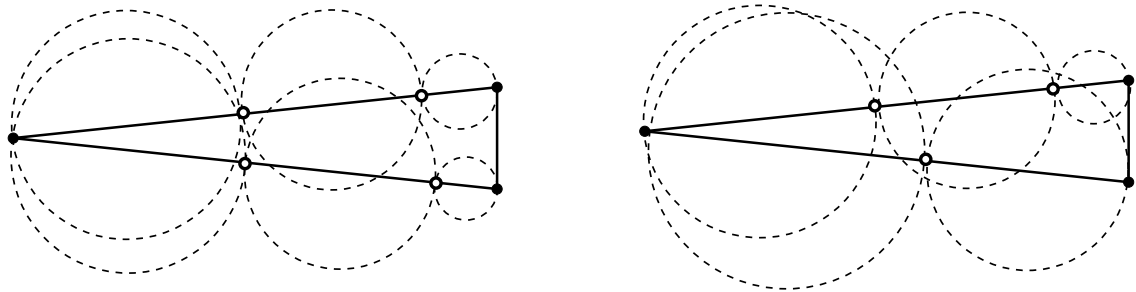}
 }
\caption{\label{GabrielPieces}
On the left the points make this piece Gabriel, on the 
right they do not.
}
\end{figure}

\section{Tubes} \label{tubes} 

Two $P$-paths $\gamma_0, \gamma_1$ 
 are {\defit parallel} if   each point of $\gamma_0$ can 
be connected to a point of $\gamma_1$ by a $Q$-segment 
(equivalently, the paths cross the same sequence of 
isosceles pieces, in the same order).
A {\defit tube } in $\Omega$ is the union 
of two parallel $P$-paths $\gamma_0, \gamma_1$
 and all the $Q$ segments that 
connect the first to the second.  The $P$-paths
$\gamma_0$, $\gamma_1$ are called the {\defit $P$-sides}
or {\defit $P$-boundaries} of the tube.
See Figure \ref{ConstructReturn6}.

 Suppose the endpoints of 
$ \gamma_0$ and $\gamma_1$  are  $\{ x_0, y_0\}$
and $\{ x_1, y_1\}$ respectively and that 
$[x_0, x_1]$ and $[y_0, y_1]$ are  $Q$-segments.
These segments are called 
the {\defit ends} or {\defit $Q$-ends} of the tube.  These may or may not 
be disjoint segments.
The two ends of a tube have the same length, and 
this  common length of each of the two $Q$-ends is called
the {\defit width} of the tube (a tube can be thought of
as a union of isosceles pieces, all of the same width, joined
end-to-end along their $Q$-sides). 
  The points $\{x_0,y_0,x_1,y_1\}$ are the  {\defit corners} 
of the tube 
 (although in some cases, 
these need not be four distinct points in the plane, e.g.
 pure spirals that we will discuss later).
{\defit Opposite corners} of a tube  mean either 
the pair
$\{x_0, y_1\}$  or the pair $\{x_1, y_0\}$.
 A {\defit maximal 
width  tube} is the union of all $P$-paths parallel to 
a given one.  If a tube is maximal width, then each 
of the $P$-path boundaries contains segments that 
are bases for at least one piece that the tube 
crosses (otherwise we could widen the tube).
See Figure \ref{ConstructReturn6}.

\begin{figure}[htb]
\centerline{
\includegraphics[height=3.0in]{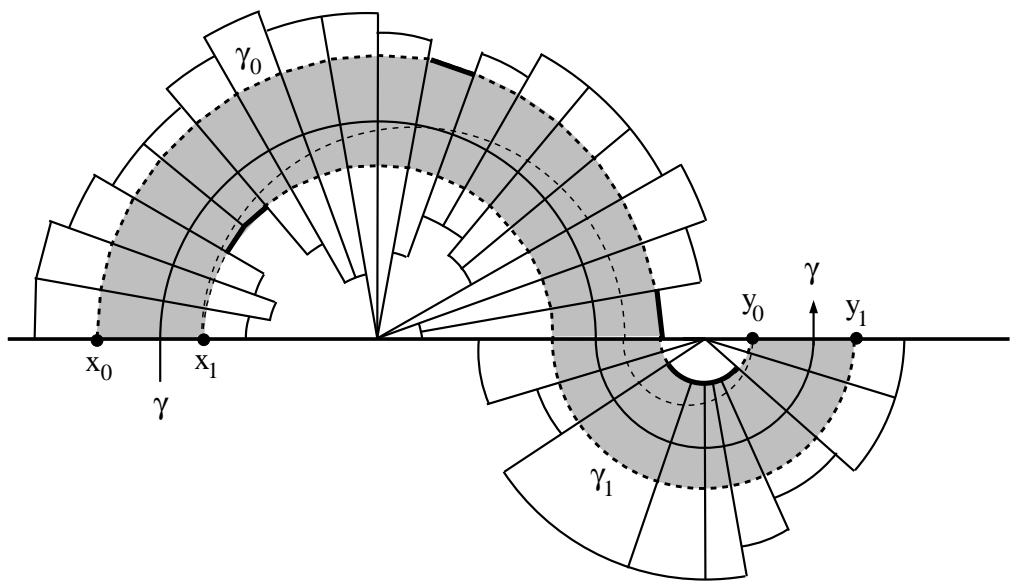}
 }
\caption{\label{ConstructReturn6}
Given a $P$-path $\gamma$ that returns to the same $Q$-segment,
there is an associated  widest return region consisting
of all parallel paths. The $P$-boundary of this tube   consists
of the two curves $\gamma_0, \gamma_1$;  by maximality, each 
must contain a base edge of a dissection piece (highlighted
with darker edges in the figure).
The points $x_0, y_1$ form one pair of opposite
 corners; $x_1, y_0$ 
the other pair.  We are interested in joining opposite 
corners by a $\theta$-path crossing the tube (dashed curve
connecting $x_1$ to $y_0$).
}
\end{figure}

We say a path {\defit strictly crosses} a tube if it 
is  contained in the 
tube and has one  endpoint on each $Q$-end.
We say a path {\defit crosses} a tube if it contains 
a sub-path that strictly crosses the tube.
The  $P$-paths that strictly cross a tube  can be parameterized 
as $\gamma_t$ with $t \in [0,1]$ where $\gamma_0$ and 
$\gamma_1$ are the $P$-sides of the tube as discussed 
above and $\gamma_t$ has endpoints  $ x_t = (1-t)x_0+ t x_1$
 and  $ y_t = (1-t)y_0+ t y_1$. Moreover 
$$ \ell(\gamma_t) = (1-t)\ell(\gamma_0) + t \ell(\gamma_1),$$
where $\ell(\gamma)$ denotes the length of a path $\gamma$. 
This formula is obvious for tubes that have a single isosceles
piece, and it follows in general since a sum of affine functions 
is affine. The path with  $t=1/2$ is called the {\defit center path}
of the tube.
Note that 
\begin{eqnarray} \label{center path est}
 \ell(\gamma_{1/2})  = \frac 12 (\ell(\gamma_0) + \ell(\gamma_1)).
\end{eqnarray}


The {\defit length}  $\ell(T)$ of a tube  $T$ 
is the minimum length 
of the  two $P$-sides, i.e., 
$$ \ell(T) = 
\min(\ell(\gamma_0), \ell(\gamma_1)).$$ 
It is possible for a tube to have length zero, 
e.g.,  when all the pieces are triangles with 
a common vertex.
 See the left side of Figure \ref{LengthZero}. 

The length of an isosceles piece is the length of 
its shorter base edge (zero for triangles). 
If a tube $T$ is made up of isosceles pieces 
$\{ T_k\}$ then it is possible  to have both $\ell(T) > 0$
and and $\ell(T_k) =0 $ for all $k$.
See the right side of Figure \ref{LengthZero}.
We define the {\defit minimal-length } of a tube to be
$$ \tilde \ell(T) = \sum \ell(T_k),$$
i.e., we sum over the minimal base length for each piece
of the tube, 
whereas $\ell(T)$ is defined by summing over all segments 
in one $P$-boundary of the tube or all segments in the other.
Clearly $\tilde \ell(T) \leq \ell(T)$.

\begin{figure}[htb]
\centerline{
\includegraphics[height=1.5in]{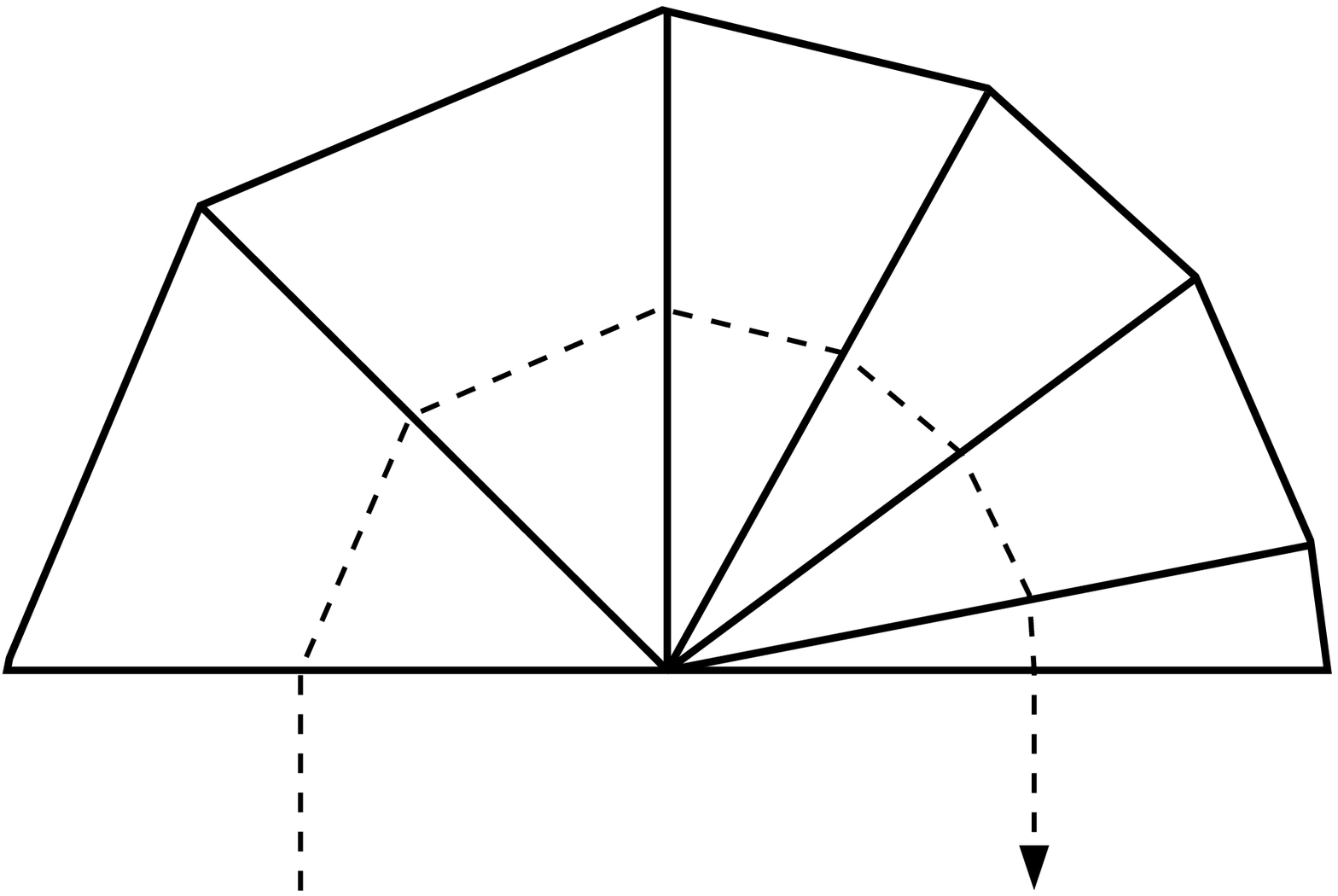}
$\hphantom{xxxx}$
\includegraphics[height=1.5in]{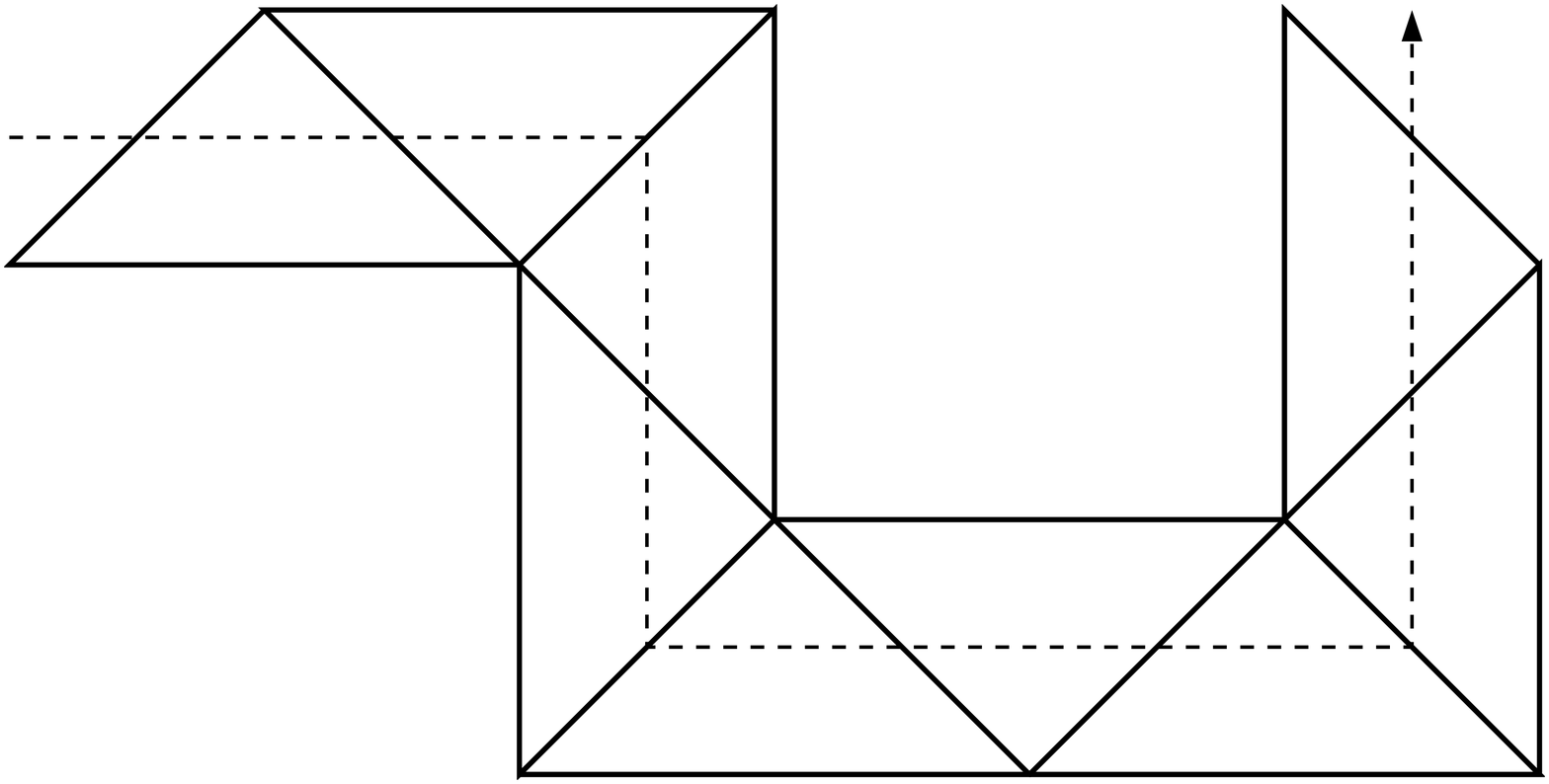}
 }
\caption{\label{LengthZero}
On the left is a tube of length zero. On the right 
is a tube with positive length, but zero minimal-length.
}
\end{figure}

As noted above, it 
is possible to have both $\tilde \ell (T) =0$
 and $\ell(T) >0$. 
  However, if we split a tube 
into two parallel tubes using the center path then 
this cannot happen for either sub-tube:

\begin{lemma} \label{sum short}
Let $Q$ be a tube and $Q_1$, $Q_2$ the parallel sub-tubes 
obtained by splitting $Q$ by its  center path $\gamma$. 
Then  $ \ell(Q) \leq  \ell(Q_1)$ and $\ell(Q) \leq 4 \tilde 
\ell (Q_1)$. 
\end{lemma}

\begin{proof}
Let $\gamma_1$ be the common $P$-boundary of $Q$ and $Q_1$, 
and $\gamma_2$ the common $P$-boundary of $Q$ and $Q_2$.
See Figure \ref{SplitTube}. 
By  Equation (\ref{center path est}), 
the  center path $\gamma$ of $Q$  has length between 
the lengths of $\gamma_1$ and $\gamma_2$. Thus 
$$\ell(Q) = \min(\ell(\gamma_1), \ell(\gamma_2))
  \leq \min(\ell(\gamma_1), \ell(\gamma)) = \ell(Q_1).$$
This is the first  inequality in the lemma.

\begin{figure}[htb]
\centerline{
\includegraphics[height=2.0in]{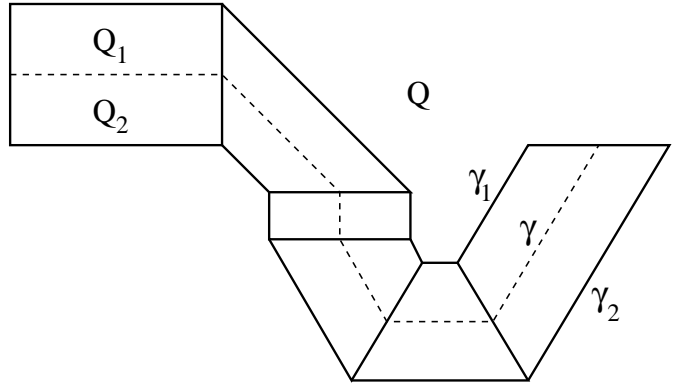}
 }
\caption{\label{SplitTube}
The tube $Q$ is split into two parallel tubes 
$Q_1, Q_2$, by its mid-path $\gamma$. 
}
\end{figure}

To prove the second inequality, suppose 
$L = \ell(Q) $ is the length 
of the tube  $Q$ and $\gamma_1$ is  
the $P$-boundary shared by $Q$ and $Q_1$. 
The length of $\gamma_1$ is the sum of the lengths 
of its segments and we group this  sum into two 
parts, depending on whether or not the segments are  the longer 
or shorter base sides of   the corresponding  isosceles 
pieces in $Q$ (if the 
piece has equal length bases, its makes no difference
in which sub-sum we place the segment). 
Call the two sums $L_0$ and $L_1$ where these 
give the sums over the shorter edges and longer edges respectively.
 By definition $L = L_0+L_1$. 
If $L_0 \geq \frac 12 L$, then the short sides of 
$Q$ (and hence the short sides of $Q_1$) add up to at
least $ \frac 12 L_1$. So in this case
$\tilde \ell(Q_1) \geq \frac 12 L \geq \frac 14 L$, 
as desired.   If $L_0 < \frac 12 L$,
 then $L_1 > \frac 12 L$ and hence
the corresponding mid-segments of the pieces add 
up to $\frac 14 L$ (since the mid-segment has length 
as least half the longer base side). But these mid-segments
are the shorter  base sides of  pieces in $Q_1$, so 
again $\tilde \ell (Q_1) \geq \frac 14 L$, as desired. 
This proves the lemma.
\end{proof}

The $P$-segments give an identification between the
non-base sides of an isosceles piece that preserves
length.  The {\defit displacement} of a  $\theta$-segment
$[a,b]$ is $|p-b|$ where $[a,p]$ is a $P$-segment. It is 
easy to check that this is unchanged if we reverse the  
roles of $a$ and $b$.  
Similarly, the two ends of a tube are identified by an 
isometry  induced by the parallel $P$-paths defining the 
tube. 
The displacement of a path   that strictly crosses the  tube is 
$|b-p|$ where $a, b$ are the endpoints of the path and 
the $P$-path starting at $a$ hits the other end at $p$.
The most important estimates in the remainder of the paper
involve  how much displacement 
a path can have, given that it satisfies certain limitations 
on its ``bending'' across each isosceles piece.
For Theorem \ref{Triangles}, the bending is bounded 
by a fixed angle $\theta$, and for Theorem \ref{NonObtuse}
the amount of bending depends on the piece and is 
determined by the Gabriel condition.

\section{Return regions}  \label{return regions}

In this section we introduce a collection 
of  regions, one of which must
be hit by any $P$-path that is sufficiently long  (in 
terms of the number of pieces it crosses). We will 
classify the regions into four types, and bound the 
total number of regions needed to form such an unavoidable 
collection. 

Suppose, as above, that $\Omega$ has an isosceles 
dissection. A return path is a $P$-path that begins 
and ends on the same $Q$-side of some piece of 
the dissection, and that  intersects  any $Q$ segment at 
most  three times. 
Figure \ref{ReturnPaths} shows four ways that this 
can happen:
\begin{description}
\item [C-curve] both ends of  $\gamma$ hit the same side of $S$ and 
       $\gamma\cup S$ separates the endpoints of $S$ from 
       $\infty$,
\item  [S-curve] $\gamma$ starts and ends on different sides of $S$,
       crossing $S$ exactly once in between, and this crossing 
       point separates the endpoints of $\gamma$ on $S$, 
\item  [G1-curve]$\gamma$ starts and ends on different sides of $S$,
       crossing $S$ exactly once,  and this crossing 
       point does not separate the endpoints  of $\gamma$ on $S$, 
\item [G2-curve] 
      $\gamma$ starts and ends on different sides of $S$  with no
        crossings $S$.
\end{description} 
When we refer to a G-curve, we can mean either a G1 or a G2-curve.

\begin{figure}[htb]
\centerline{
\includegraphics[height=2.0in]{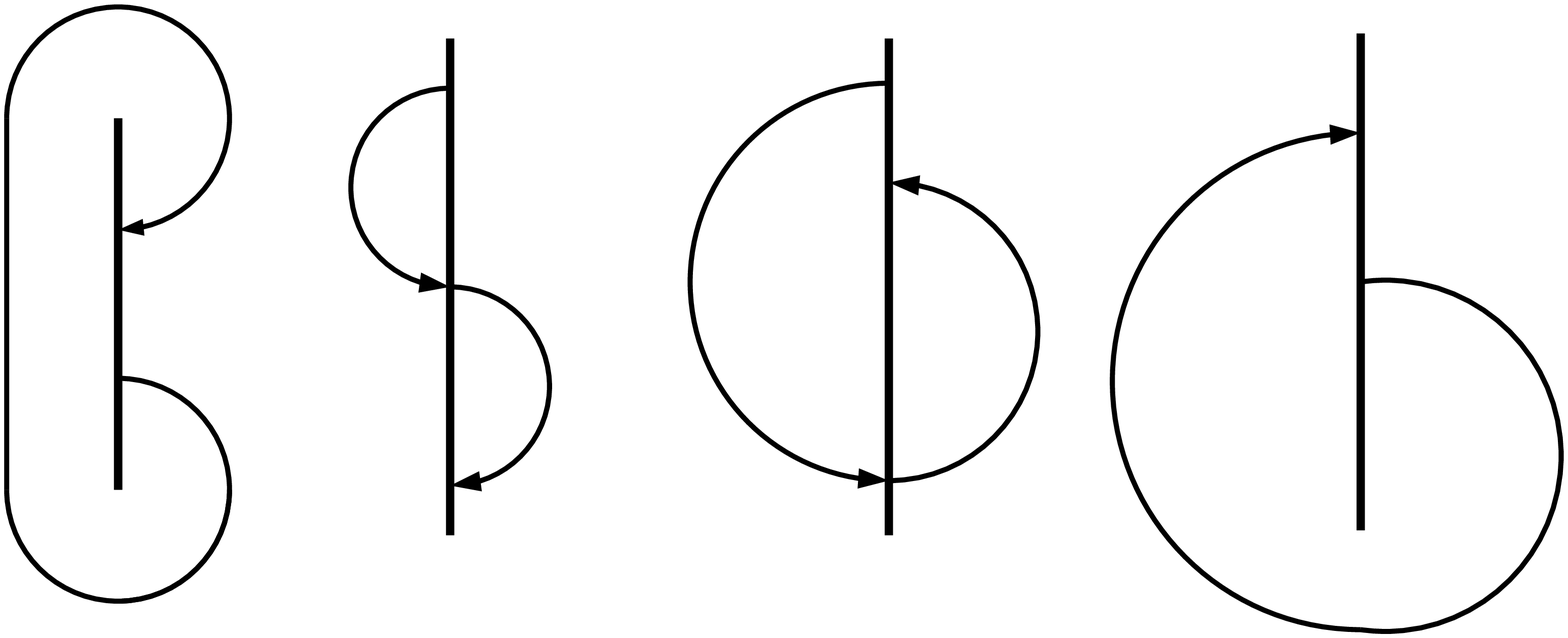}
 }
\caption{\label{ReturnPaths}
A C-curve, S-curve and the two types of G-curve.
Each is named for the letter it 
vaguely resembles.
}
\end{figure}

\begin{lemma} \label{must cross}
Suppose $n$ is the number of isosceles pieces in 
an isosceles dissection.
Every $P$-path $\gamma$  with $2n+1$ segments contains a  sub-path that
is a return path of one of the four types described above.
\end{lemma} 

\begin{proof}
A $P$-path
$\gamma $ with $2n+1$ segments must cross some 
dissection piece  three times by the pigeon hole principle. 
Therefore $\gamma$  crosses a non-base side $S$ of such a piece at
least three times. By passing to a sub-path, if necessary, 
we can also assume  $\gamma$ does not hit any  $Q$-side 
more than three times.
 Suppose that $\gamma$ does not contain a 
C-curve or  a G2-curve as a subpath.  Then  the sub-path
between  its
first and second visit to $S$  starts and ends on the same side 
of $S$ and the same for the sub-path between its second and third 
visit (but now it starts and ends on the other side of $S$). 
Thus the subpath formed between the first and third visits is 
either a S-curve or a G1-curve. Thus one of the four types of
curve must occur as a subpath.
\end{proof}

A tube consisting of parallel return paths will be 
called a {\defit return tube} and be called a C-tube, 
S-tube, G1-tube or G2-tube depending on the type
of curves it contains (clearly all parallel curves must be 
of the same type). We want to show that the length of a 
return tube cannot be   too small compared to its  width.
We do this by considering the different types of tubes 
one at a time.  We call a return tube a {\defit simple tube} if the 
two ends have disjoint interiors (they may share a corner); 
otherwise the region is called a {\defit spiral}. 
A C-tube or S-tube must be a simple tube; a G-tube can be 
either be a simple tube  or a spiral. As the name suggests,
simple tubes are easier to understand and we start  with this 
case.  See Figure \ref{Gregion3}.

\begin{figure}[htb]
\centerline{
\includegraphics[height=1.8in]{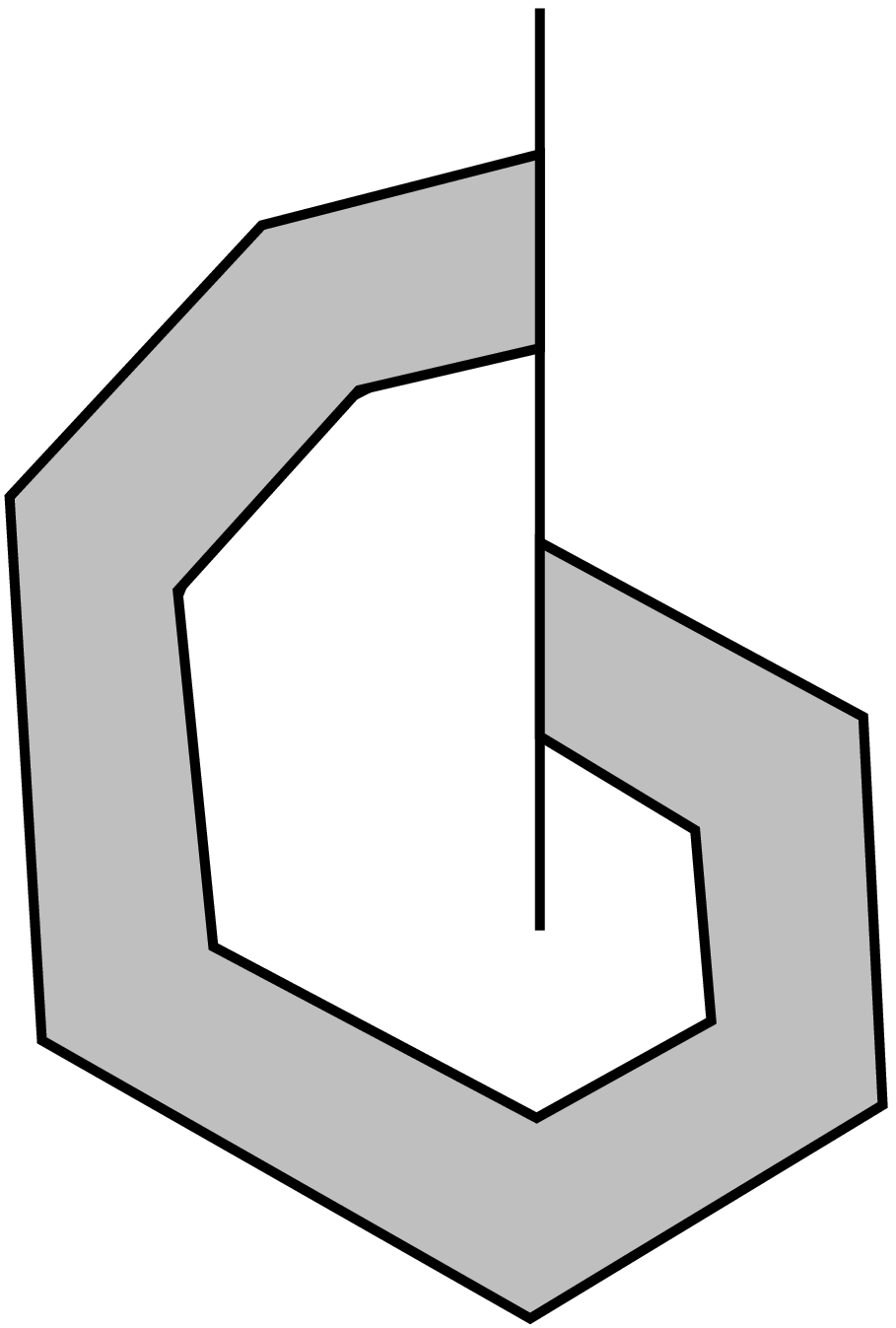}
$\hphantom{xxxxx}$ 
\includegraphics[height=1.8in]{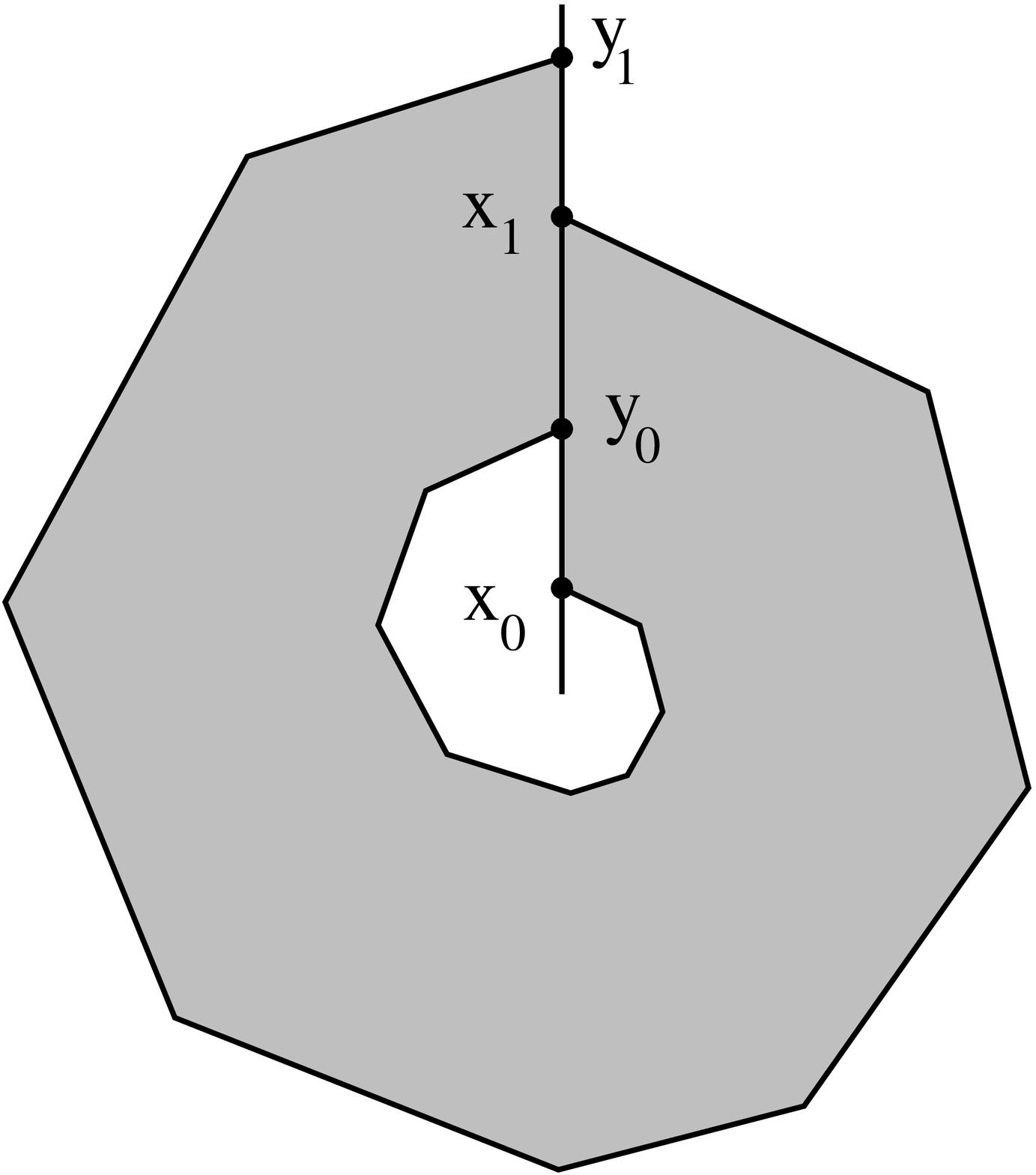}
 }
\caption{\label{Gregion3}
A G-region can form a single tube with disjoint ends 
(a simple G-tube)   or 
the two ends can overlap (a spiral).
} 
\end{figure}

\begin{lemma} \label{C-tube lemma} 
The length  $L$ of a C-tube is at least twice its width $w$. 
\end{lemma}

\begin{proof} 
The ends  of a C-tube are disjoint intervals on the 
same $Q$-segment $S$, so the length of $S$ is at least $2w$.
But both $P$-sides of the tube cover $S$ when projected orthogonally 
onto the line containing $S$, so  both $P$-sides have 
length  $ \geq |S| \geq 2w $.  The length of the tube is the 
minimum of these two path lengths, so is also $\geq 2 w$.
\end{proof} 

\begin{lemma} \label{S-tube lemma} 
The length  $L$ of a S-tube  is at least twice its width $w$. 
\end{lemma} 

\begin{proof}
Split the S-tube 
into four sub-tubes as follows.
Each $P$-path in the tube is cut by a point
where it crosses $S$ and this cuts the 
tube into two sub-tubes that also have 
width $w$, say $U_1$, $U_2$, which meet 
end-to-end.  Each of these  are split into 
two thinner tubes by the central $P$-path $\gamma_{1/2}$, 
giving four sub-tubes called 
$U_1^i, U_1^o, U_2^i, U_2^o$, e.g., $U_1^i$ is the {\defit 
inner part} of $U_1$ and  its endpoints on $S$ separate 
the endpoints of the {\defit outer part} $U_1^o$ on $S$.
See Figure \ref{InnerOuter}.
Note that the length of the original S-tube is at 
least the minimum of the lengths of the two outer tubes
(since they each have a $P$-boundary contained in the 
$P$-boundary of the S-tube).

\begin{figure}[htb]
\centerline{
\includegraphics[height=2.0in]{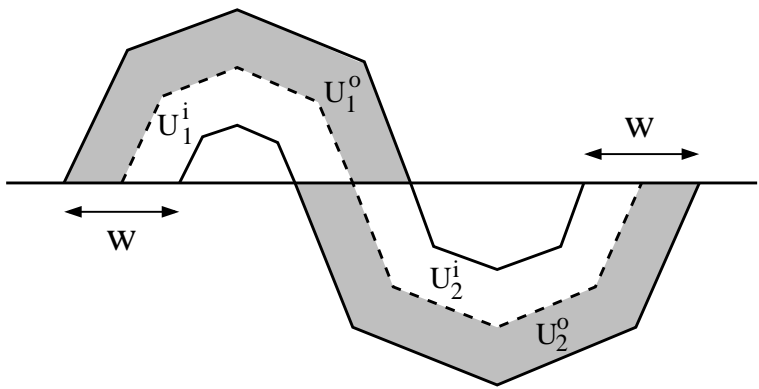}
 }
\caption{\label{InnerOuter}
The outer tubes of a S-tube are shaded.
} 
\end{figure}

However, the endpoints of  the outer  $P$-boundary of an 
outer part are separated by at least  distance $2w$, 
and each $P$-boundary of $S$ contains the outer 
$P$-boundary of one of its outer parts. Thus both 
$P$-boundaries have length at least $2w$.
\end{proof}

\begin{lemma} \label{G-tube lemma} 
The length $L$  of a simple G-tube is at least its width.
\end{lemma} 

\begin{proof}
Suppose $I$ and $J$ are the ends of the 
spiral. If we project either $P$-side of the 
tube orthogonally onto $S$, then it covers 
either $I$ or $J$, so both sides of the tube 
are longer than the tube is wide.
\end{proof}

Putting together the last three lemmas we get 

\begin{cor} \label{long tubes}
The length $L$ of a simple return tube is at least 
its width $w$. If the tube is not a $G$ tube, then
$L \geq 2 w$.
\end{cor} 

The more interesting and difficult return regions 
are the spirals: G-tubes  where 
the $Q$-ends  $[x_0, x_1]$ and $[y_0, y_1]$ 
overlap but are not identical.

 Suppose  $S$ is a spiral return tube
and  the corners  are ordered on $S$ as 
$x_0 < y_0 < x_1 < y_1 $ and define $Q$-segments
 $I_0= [x_0, y_0]$ and $I_1 = [x_1, y_1]$.
See  Figure \ref{ManyReturns}.
There is a $P$-path in that starts at $y_0$ and ends 
at a point $z$ in $I_1$ and is composed of $P$-paths in the tube 
joined end-to-end where they cross $S$. 

If $z < y_1$, then  we can  remove  the  simple G-sub-tube 
with $Q$-end  $[z, y_1]$ as shown in Figure 
\ref{ManyReturns} (the dark gray tube). What is left
after removing this simple tube a {\defit pure 
spiral}, a union of parallel $P$-paths that each  
consist of $N$ $P$-paths from the original tube.
We call  $N$   the {\defit winding number}
 of the  pure spiral (in this notation a simple tube has
winding number one).

\begin{figure}[htb]
\centerline{
\includegraphics[height=2.5in]{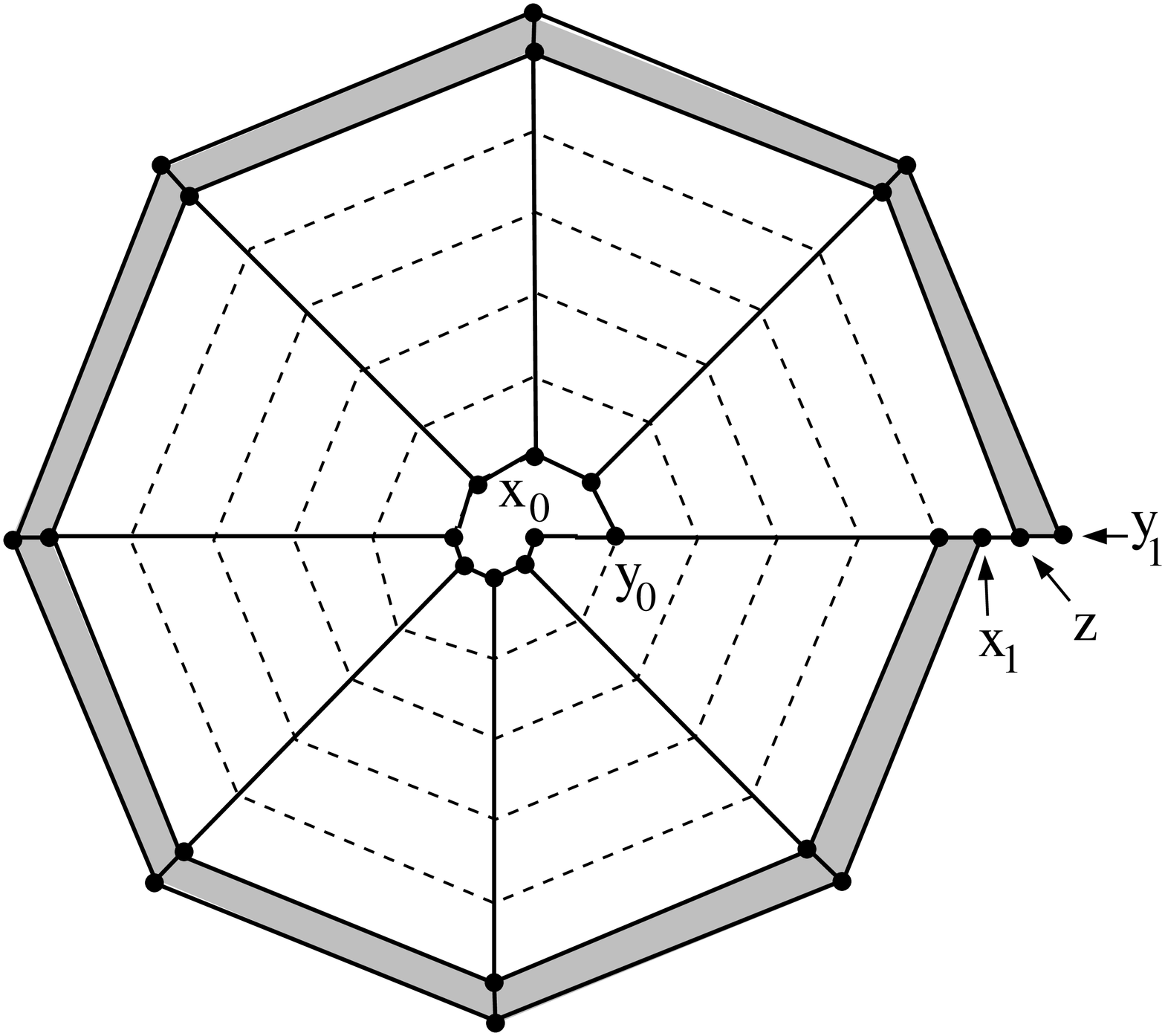}
 }
\caption{\label{ManyReturns}
A spiral can always be divided into a simple tube 
(shaded)
and a pure spiral. The pure spiral 
 can be thought of as many parallel simple 
G-tubes that don't cross $S$, or as one very long tube
that crosses $S$ multiple times. 
The pure spiral here has five windings.
}
\end{figure}

Since a pure spiral is made up of $N$ simple G-tubes
all with the same width joined end-to-end,  it  is clear
that the length of a spiral  with $N$ windings
should be at least $N$ 
times its width. However, we can  do better than this, 
since each turn of the spiral  is longer than the previous 
one. 

\begin{lemma} \label{pure spiral length}
A pure spiral with $N$ turns and width $w$ has length at least 
$ N^2 w$.
\end{lemma} 

\begin{proof}
  Without loss of generality we may scale the spiral so the
  width $w=1$.
  Use  the segment
   $[x_0, y_1]$ to cut the entire spiral into $N$ simple G-tubes.
  The first has length at least $w$, because both $P$-boundaries 
  project orthogonally onto one of the $Q$-ends.  In general, 
  both the parts of the  $P$-boundary paths of the $j$th sub-tube 
  have length at least $2j-1$. To see this, consider 
  the curves in each half-plane defined by 
   the  line $L$  through $x_0, y_1$, and project orthogonally onto
   $L$.  The part in one half-plane projects to a segment of length
    at least $j-1$ and the other to a segment of length 
    at least  $j$. See Figure \ref{LengthEst}.  Hence
  the $j$th tube has length at least $2j-1$. 
  Summing $1+3 +5+ \dots + (2N-1) = N^2 $  gives the result.
\end{proof} 

\begin{figure}[htb]
\centerline{
\includegraphics[height=1.0in]{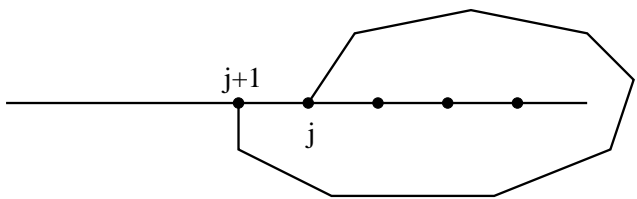}
 }
\caption{\label{LengthEst}
Estimating the length of a G-curve.
} 
\end{figure}

The argument proving  Lemma \ref{pure spiral length} will be used again 
in Section \ref{dyadic merging}.  
Next we slightly refine Lemma \ref{must cross} 

\begin{lemma} \label{must cross 2}
Suppose $n$ is the number of isosceles pieces in 
an isosceles dissection.
Every $P$-path $\gamma$  with $3n+1$ segments
 contains a  sub-path that begins and ends on 
the $Q$-end of a chain and 
is a return path of one of the four types described above
(S, C, G1, G2).
\end{lemma} 

\begin{proof}
Apply Lemma \ref{must cross} to the initial path 
of $2n+1$ steps, to get a return path of one 
of the four types. If the path already begins and
ends on the $Q$-end of a chain there is nothing to 
do. If it begins and ends at an interior 
$Q$-segment of the chain then by deleting or 
extending the paths as shown in Figure \ref{EndChains}
we can obtain a path that begins and ends 
on the $Q$-end of the chain. The number of steps 
added is less than  $n$, so the new path 
must be a sub-path of $\gamma$.
\end{proof} 

\begin{figure}[htb]
\centerline{
\includegraphics[height=2.5in]{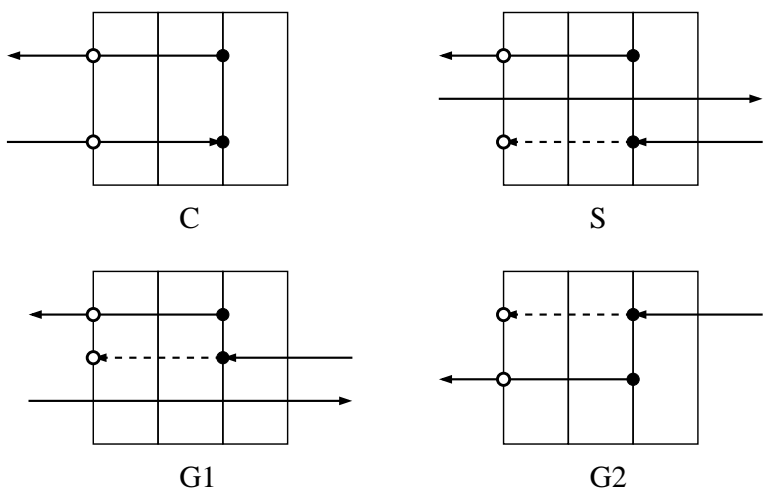}
 }
\caption{\label{EndChains}
A return path ending inside a chain can 
easily be modified to start and finish  on 
the $Q$-end of the chain (the extensions 
are shown as dashed lines and follow 
$P$-paths). Black dots are the original 
endpoints of the path an the white dots
are the modified endpoints. 
} 
\end{figure}

We say that a return region is {\defit standard} if 
is of one of the types (C, S, G1, G2) discussed 
above and if it begins and ends on segments that 
are  the $Q$-ends of some chains (possibly the same 
chain or  two different chains). 
The following is one of the key estimates of this paper.

\begin{lemma} \label{return regions lemma} 
If $\Omega$ has an isosceles dissection into $n$ pieces with $M$
chains, then 
there are $O(M)$ standard return regions 
with disjoint interiors  so that
any $P$-path with more than $5n+1$ segments  must 
hit one of the regions.
\end{lemma}

\begin{proof}
Each chain in the dissection has two $P$-boundaries. 
Each of these $P$-boundaries may or may not be part of 
the $P$-boundary of a standard return region.  If it
is, then associate to the $P$-boundary of the chain 
the maximal width standard return region that contains 
the $P$-boundary of the chain in its own boundary. 
Note that at most $2M$ return regions can be 
selected in this way, since there are $M$ chains 
and each has two $P$-boundaries.

We claim that any $P$-path  $\gamma$ in the dissection with 
$5n+1$ steps contains a sub-path that crosses one 
of the selected  return regions.  
Let $\gamma'$ be the path obtained by deleting $n$ 
steps from each end of $\gamma$.
By Lemma \ref{must cross 2}, 
 $\gamma'$  contains a sub-path 
$\gamma''$ that is a return path of one of the four standard 
types and which begins and terminates on the $Q$-end of a 
chain. Thus the set of paths parallel to $\gamma''$ forms 
a  standard return tube $T$ of maximal width.
If $T$ is one of the chosen return regions, then we 
are done since $\gamma'' \subset \gamma$ crosses $T$.

On the other hand, suppose $T$ is not one of the 
chosen regions.  Since $T$ is maximal, it contains 
the $P$-boundary of some chain  $C$ within its own 
$P$-boundary. Since $T$ was not chosen, 
there  must be another return 
region $T'$, at least as wide as $T$, that was chosen and 
$T\cap C \subset T'\cap C$. Thus every 
path crossing $T$ hits $T'$. Thus $\gamma''$ hits $T'$.
 If we add $n$ steps to both ends 
of $\gamma''$, the new, longer path must now cross 
$T'$, but it is still a sub-path of $\gamma$. 
Thus $\gamma$ crosses some return region in the
chosen  collection.  This proves the claim.

The collection of  maximal width 
return regions defined above may overlap. 
To get disjointness, we order the chosen regions 
$R_1, \dots , R_m$, $m=O(M)$,  from widest to narrowest 
and  label the first region ``protected'' and 
label the remainder ``unprotected''. At each stage we 
look at the first unprotected region $R_k$ in 
the current list and 
see if there are any $P$-paths  strictly crossing it that 
intersect a protected region (anything earlier in 
the list). If there is no such path, then label 
$R_k$  protected, and move to the next region. 

If there is a  $P$-path  $\gamma$ strictly crossing 
$R_k$  that hits a protected region $R_j$, then
remove from $R_k$  any  $P$-paths  that hit 
$R_j$.
 Since   $R_j$ 
is at least as wide as $R_k$, removing these paths 
gives a connected return region $R_k' \subset R_k$
(possibly empty).  Now re-sort 
the list by width.  $R_k'$ either stays where it is or 
moves later in the list; the protected regions 
all stay where they are. Since two regions will never 
overlap after the first time they were compared, this 
process stops  after at most $m^2$ steps, and gives
a collection of return regions with disjoint 
interiors.  Moreover, once a region is protected, 
it is never modified again and is part of the 
final collection.

Finally, we have check that every long enough $P$-path 
hits one of the disjoint regions.
If $\gamma$ is any path with $5n+1$ segments then it 
crossed some region in the original list. Suppose 
 $R_\gamma$ 
is the first region on the  original sorted 
list that is crossed by $\gamma$. 
If a part of $R_\gamma$ containing
$\gamma$  is deleted in the construction, 
then $\gamma$ must hit a  protected return region.
Thus $\gamma$ hits a 
return region in the final collection, as desired.
\end{proof} 

This lemma is used in the proofs of both Theorems 
\ref{NonObtuse} and \ref{Triangles}. 
In both cases we will reduce to meshing a region 
$\Omega$ that has an isosceles dissection. We will 
propagate the non-conforming vertices until they 
either hit the boundary of $\Omega$ or hit the 
$Q$-end of a return region. By Lemma \ref{return regions lemma}, 
one of these two options must occur with $O(n)$
steps. 
The  proofs of 
the two theorems differ mostly in how we construct
$\Omega$ and  how  we mesh inside the return regions.

\section{Proof of Theorem \ref{Triangles}:   reduction 
to  a meshing lemma  } \label{ideal} 

 This is the first of four sections that construct the 
 almost nonobtuse triangulation in Theorem 
 \ref{Triangles}.
 Unlike the proofs of Theorems \ref{NonObtuse} and 
 \ref{Refine Triangulation}, the proof of Theorem
\ref{Triangles} does not make use of  
Theorem \ref{BMR}
to create the triangulation; we shall construct 
the triangulation directly. However,
 we will  use the result of Bern, Mitchell and Ruppert 
\cite{BMR95} that
any simply $n$-gon has a nonobtuse triangulation
with $O(n)$ triangles.  
We will also  make use of return regions and bending paths, 
both ideas we shall use again in the proof of 
Theorem \ref{NonObtuse}.

  As in the proof of Theorem \ref{Refine Triangulation},
  we start with a PSLG that is a triangulation.
   In  that earlier proof we divided each triangle $T$ into 
   a central triangle and three isosceles triangles. Here 
  we will replace the single central triangle by a 
  simple polygon that approximates a triangle with 
  circular edges.

Given a triangle $T$ with 
vertices $v_1, v_2, v_3$, let $C$ be
the inscribed circle and let  $z_1, z_2, z_3$ be  the three points 
where this circle is tangent to the triangle (numbered
so that $z_k$ lies on the side of $T$ opposite $v_k$).
Any pair of the $z$'s 
 are equidistant from one of the vertices of $T$ and 
hence are connected by a circular arc  centered at this vertex.
This defines a  central region bounded by three 
circular arcs, each 
pair of arcs tangent where they meet (see the shaded area on the left
of Figure \ref{AlmostCenter8}). We will replace each of these
circular arcs by a polygonal path inscribed in the arc.
For example,
 let $\gamma_1$ be a polygonal arc inscribed in the 
circular arc connecting $z_2$ and $ z_3$.
If $\gamma_1$ consists of $m$ equal length segments, 
then the angle  subtended from $v_1$  by these segments
 is less than $ \pi/m$,
 (since $\gamma_1$  subtends at most angle $\pi$).
If we then connect the vertices of $\gamma_1$
to  $v_1$ we obtain a chain of isosceles triangles 
all with angle $\leq \theta =  \pi/m$.
See the right side of Figure \ref{AlmostCenter8}. 
  Taking the union of these isosceles triangles over all 
  $T$ in the original triangulation gives 
    the region $\Omega$, that clearly has 
   a $\theta$-nice isosceles dissection with $O(n)$ chains 
   and $O(n/\theta)$ pieces. 

\begin{figure}[htb]
\centerline{
\includegraphics[height=2in]{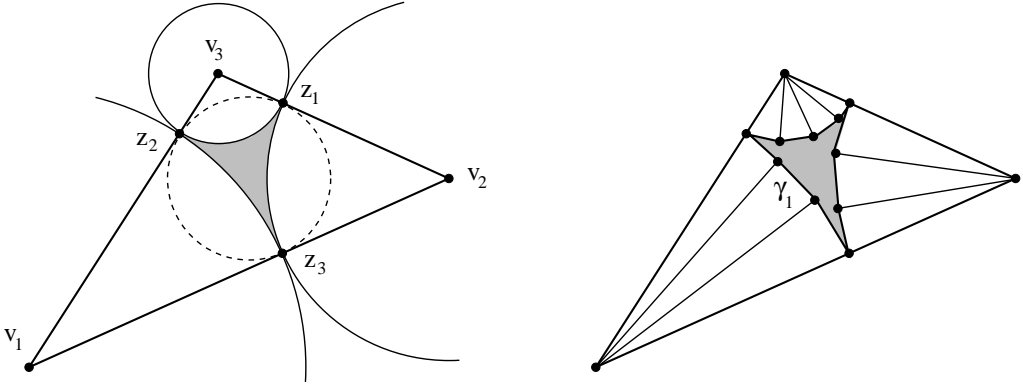}
 }
\caption{\label{AlmostCenter8}
We   define 
a simple polygon by inscribing polygon arcs on 
circular arcs as shown above. If we use $m$ 
evenly points on each arc then remaining region 
clearly has a $\theta$-isosceles dissection 
for $\theta = \pi/m$. 
}
\end{figure}

\begin{lemma} \label{mesh exists}
Suppose $\Omega$ is a region that has 
a  $\theta$-nice isosceles dissection
(both triangles and quadrilaterals are 
allowed).
Assume that the dissection has $O(n)$ chains and 
$O(n/\theta)$ pieces.
Then there is a mesh of $\Omega$  using $O(n^2/\theta^2)$
$\theta$-nice quadrilaterals and triangles. 
Moreover, each dissection piece  $T$
contains  at most $O(n/\theta)$  mesh elements 
and every mesh element $Q$ 
 contained in $T$
is bounded by at most two sub-segments
of the $Q$-sides of $T$ (possibly points) and 
at most two  
$\theta$-paths in $T$ (possibly the vertex of $T$, if $T$ 
is a triangle).
\end{lemma}

We will prove this in the next three sections.
We can deduce Theorem \ref{Triangles} from 
Lemma \ref{mesh exists} as follows. 
By a result of Bern, Mitchell and Ruppert, each 
central polygon has a nonobtuse triangulation with 
at most $O(1/\theta)$ triangles. This triangulation 
may place extra vertices on the edges of the
central polygon, but not more than $O(1/\theta)$
vertices in total. 
Each such edge  $e$ is the base of one of the 
isosceles triangles $T$ in the dissection, and
we connect the extra vertex on $e$ to the opposite
vertex of $T$ by a $Q$-segment $S$.
See Figure \ref{ExtraVertices}. 
 
\begin{figure}[htb]
\centerline{
\includegraphics[height=2in]{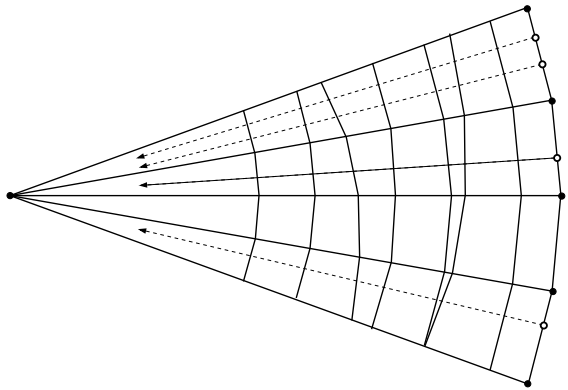}
 }
\caption{\label{ExtraVertices}
Each isosceles triangles $T$ in the
dissection of $\Omega$ is meshed by at most 
$O(n/\theta)$ quadrilaterals and triangles 
(here a chain of four triangles is shown).
Connecting ``extra'' vertices (white dots) on the  base
of $T$  to the opposite vertex thus creates 
at most $O(n/\theta)$ extra mesh pieces per 
dissection triangle. 
}
\end{figure}

 This creates a 
new $\theta$-nice quadrilateral or triangle for 
each  $\theta$-path that $S$ crosses. 
Since there are at most $O(n/\theta)$ such 
paths per piece, each extra vertex on $e$ creates 
at most $O(n/\theta)$ new elements of the mesh. 
Since there are $O(n)$ central polygons and 
each has at most $O(1/\theta)$ extra boundary points 
coming from its nonobtuse triangulation, at most 
$O( n^2/\theta^2)$ extra pieces are  created overall.  

As the final step, we add diagonals to the 
quadrilateral pieces  of the mesh, getting 
a triangulation. Since all the quadrilaterals 
are $\theta$-nice, the resulting triangles have 
maximum angle $90^\circ + \theta$, which proves 
Theorem \ref{Triangles}.

Our application of Lemma \ref{mesh exists} to 
Theorem \ref{Triangles} only needs to apply to 
dissections consisting entirely of triangles, 
but has been stated for more general isosceles dissections 
which may use both triangles and  
quadrilaterals. The extra generality does not 
lengthen the proof at all, but it is useful for 
the  application to optimal quad-meshing given 
in \cite{Bishop-quadPSLG}. That application involves
an isosceles dissection that uses only 
trapezoid pieces; the precise variant of Lemma 
\ref{mesh exists} that is needed in that paper 
will be stated and proved in 
Section \ref{quad mesh lemma sec}.

\section{Proof of Lemma \ref{mesh exists}: outside the 
return regions}
 \label{ proof outside}

We continue with  the proof of Theorem \ref{Triangles}, 
by starting the proof of Lemma \ref{mesh exists}.
 In this section we will 
mesh the  part of  $\Omega$  that is outside 
the return regions.

Let $\{ R_k\}_1^N$ be the   disjoint return 
regions for $\Omega$ given by Lemma \ref{return regions lemma}.
Since there are $O(n)$ chains 
there are  $ O(n)$ return regions.  
For each triangle  $T_k$, 
and each of the three   vertices of $T_k'$ on its boundary,
 construct the $P$-path starting at this 
point and continued until it hits another cusp point, leaves
$\Omega$ or enters a return region. 
Lemma \ref{make mesh} says these paths cut the isosceles
pieces of the dissection into isosceles pieces that form a mesh. 
   By Lemma \ref{return regions lemma},
 each $P$-path we generate 
terminates within $O(n/\theta)$ steps and there are less than 
$3n$ of these paths (at most three per triangle), 
so a total of 
$O(n^2/\theta)$  mesh pieces  are  created 
outside the return regions.  
Moreover, each such path crosses a single dissection 
piece at most $O(1)$ times. Thus each  dissection piece can 
be crossed at most $O(n)$ times but such paths.

Next we place $O(1/\theta)$ evenly spaced points on 
both $Q$-sides of each return region (the reason for this 
will be explained in the next section). Each of these
points is propagated by $P$-paths 
 outside the return region it 
belongs to, until it runs into the boundary of $\Omega$
or hits the $Q$-side of some return region (possibly the 
same one they started from). As above, this generates 
a $\theta$-nice  mesh outside the return regions. 
There are $O(n)$ return regions and   $O(1/\theta)$
points per region to be propagated. Each path continues 
for at most $O(n/\theta)$ steps, so  at most 
 $O(n^2/\theta^2)$ mesh elements are created in total.
Moreover, each dissection piece is crossed at most 
$O(1)$ times by each path, so is crossed $O(n/\theta)$
times in total by such paths.

\section{Proof of Lemma \ref{mesh exists}: the simple tubes}
 \label{proof triangles} 

Next, we have  mesh  inside the return regions.
In this section, we deal with the return regions that 
are simple tubes and in the next section we deal 
with spirals.

For the first time we will use 
$\theta$-paths rather 
than  $P$-paths (recall that a $\theta$-path is made up of 
segments that are within $\theta$ of parallel to the base 
of the isosceles piece; when we cut a $\theta$-nice piece by 
a $\theta$-path we get two $2\theta$-nice pieces).
We need the following lemma.

\begin{lemma} \label{conn corners}
Suppose $Q$ is a tube whose width $w$ is at most $\sin \theta$ times
its  minimal-length $\tilde \ell$.
Then opposite 
corners of $Q$ can be connected by a $\theta$-path inside 
the tube. 
\end{lemma}

\begin{proof}
Suppose $T$ is a $\theta$-nice isosceles piece 
 and $x,y$ are the  endpoints of 
a $P$-segment  $S =[x,y]$  crossing $T$. 
Then any point $z$ on the same  side as $y$ 
and within distance $\sin(\theta)|x-y|$ can joined to $x$ by a 
$\theta$-path.  See Figure \ref{ThetaPath}.
Thus if $\{T_k\}$ is an enumeration of the pieces making up the
tube and $\ell_k$ is the minimal base length of the $k$th 
piece, then we can create a $\theta$-path that crosses the 
tube and    whose endpoints 
are displaced by $\sum_k \ell_k \sin \theta = 
\tilde \ell \sin \theta $
with respect to a $P$-path. This proves the lemma.
\end{proof}

\begin{figure}[htb]
\centerline{
\includegraphics[height=1.75in]{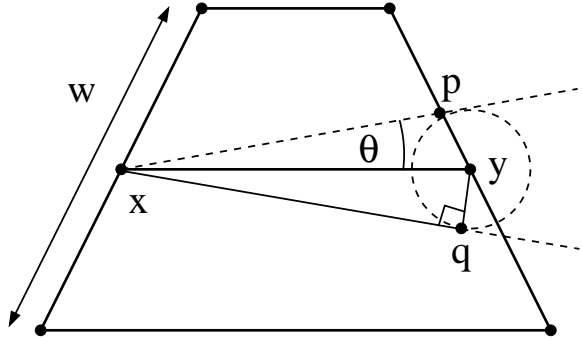}
 }
\caption{\label{ThetaPath}
Clearly  $|p-y|\geq |q-y| = |x-y|\sin \theta 
\geq \tilde \ell \sin \theta $ where $q$ is 
the closest  point to $y $ on the line making angle 
$\theta$ with the segment $xy$.
}
\end{figure}

\begin{cor}
If $R$ is return region that is a  C-tube,    S-tube or  simple 
G-tube, then we can 
cut $R$ into $O(1/\theta)$  parallel sub-tubes and connect 
opposite corners of the each tube by a $\theta$-path 
contained in that tube. 
\end{cor} 

\begin{proof}
Suppose the return region $R$ has length $L$ and width $w$. 
Choose an even integer $M \geq  {4}/{ \sin \theta}$
and split the return region $R$ into the disjoint 
 union of $M$ thinner tubes $\{ T_j\}$  of width $w/M$.
By Lemmas \ref{C-tube lemma}, \ref{S-tube lemma} and 
\ref{G-tube lemma} each of these new  tubes has length that is at least 
$w$  and width equal to $w/M$. Thus each has length 
that is at least $ M $ times as long as its width. 

Since  $M \geq 2$, each of our thin tubes is half of 
a thicker tube  that is still inside the given tube $R$.
By Lemma \ref{sum short}
$$ \tilde \ell(T_j) \geq \frac 14 \ell(T_j) \geq
\frac 14 \ell(R) \geq \frac 14 w.$$
On the other hand, the width of $T_j$ is $w/M$.
 Hence the minimal-length of each tube $T_j$  is 
more than 
$ M/4 = 1/\sin\theta$ times its width, 
so the previous Lemma \ref{conn corners} applies 
to $T_j$, as desired.
\end{proof} 

We can now continue with the proof of  Lemma 
\ref{mesh exists}. 
We then cut  the  return  
region into  $O(1/\theta)$ parallel tubes as described above, 
and divide each tube by a $\theta$-path connecting opposite 
corners.  This meshes each tube using $2\theta$-nice 
pieces.
Any $P$-path hitting a $Q$-end of one of these tubes 
is then propagated to a corner on the opposite end of 
the tube by standard quadrilateral propagation paths.
 This gives a $2\theta$-nice
mesh of the tube that is consistent with all the meshes 
created outside the tube.

Since there were at most $O(n/\theta)$ $P$-paths that 
might terminate  and each return region has at most 
$O(n/\theta)$ isosceles pieces, at most $O(n^2/\theta^2)$ 
$2\theta$-nice 
triangles and quadrilaterals are created inside all the return 
regions.
Moreover, each dissection piece is crossed by at most
$O(n/\theta)$ paths (there are $O(1/\theta)$ paths per return 
region and $O(n)$ return regions), so it contains at most this many 
mesh pieces. 


\section{Proof of Lemma \ref{mesh exists}: the spirals}
\label{proof spirals} 

This is the final section in the 
proof of Theorem \ref{Triangles}.
Here we   prove  Lemma \ref{mesh exists} inside 
 the spiral return regions. 

Since any spiral can be divided into a simple G-tube and 
a pure spiral, and we can treat the simple tube as 
above,  it suffices to deal with the pure spirals. 
Let $N$ be the winding number of the spiral; we may 
assume $N \geq 2$, since otherwise the spiral  can be 
treated as a 
tube and can be triangulated as in the previous section.
Let $p$ be the number of isosceles pieces that are hit
by the spiral (this is the number of steps it take to complete 
one winding of the spiral). 

The spiral can be divided into $N$ tubes joined end-to-end, 
each starting and ending on the same $Q$-edge of some 
isosceles pieces. We divide the first and last of these 
tubes into $O(1/\theta)$ parallel thin tubes. Then any 
$P$-path that enters the tube from either end can 
be $\theta$-bent so that it terminates at the corner 
of one of the thin tubes after winding once around 
the spiral. 

If $N = O(1/\theta)$, then we simply propagate all the 
interior corners of the thin tubes at one end of 
the spiral  around the spiral 
until they run into the  corners of  the 
thin tubes at the other end. This generates 
$O( N \cdot p \cdot \theta^{-1}) = (p \theta^{-2})$ 
new vertices.
See Figure \ref{SpiralMesh2}.

\begin{figure}[htb]
\centerline{
\includegraphics[height=1.8in]{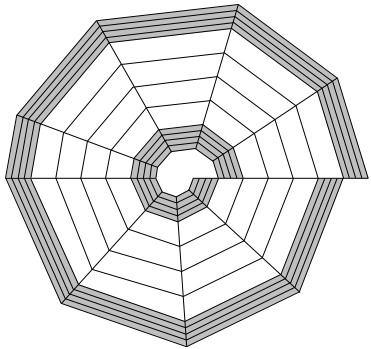}
 }
\caption{\label{SpiralMesh2}
Cut the inner and outer tubes into $O(1/\theta)$ 
parallel narrow tubes and $\theta$-bend all 
entering $P$-paths so they terminate inside these
narrow tubes. The corners of the narrow inner tubes are 
then propagated around the $N$ turns of the spiral 
until they hit the corners of the narrow outer tubes. 
This creates a $2\theta$-nice mesh of the spiral 
using $O(n N /\theta)$ pieces. 
The figure shows $p=9$ and $N=6$.
}
\end{figure}

If $N \gtrsim 2\theta^{-1} \geq 1/\sin(\theta)$
 then  after $O(\theta^{-1})$
spirals, we can create a $\theta$ curve that is a closed loop 
and we let the paths generated by the  interior 
corners of the inner thin part hit this closed 
loop. We create another $\theta$-bent closed loop at radius 
$N-2$ and let the paths generated by the corners of the 
outer thin tubes  hit this.  No propagation paths enter 
the region between the two closed loops and a total of 
$O(p  \theta^{-2})$ vertices are used.
See Figure \ref{SpiralMesh5}. 

\begin{figure}[htb]
\centerline{
\includegraphics[height=1.9in]{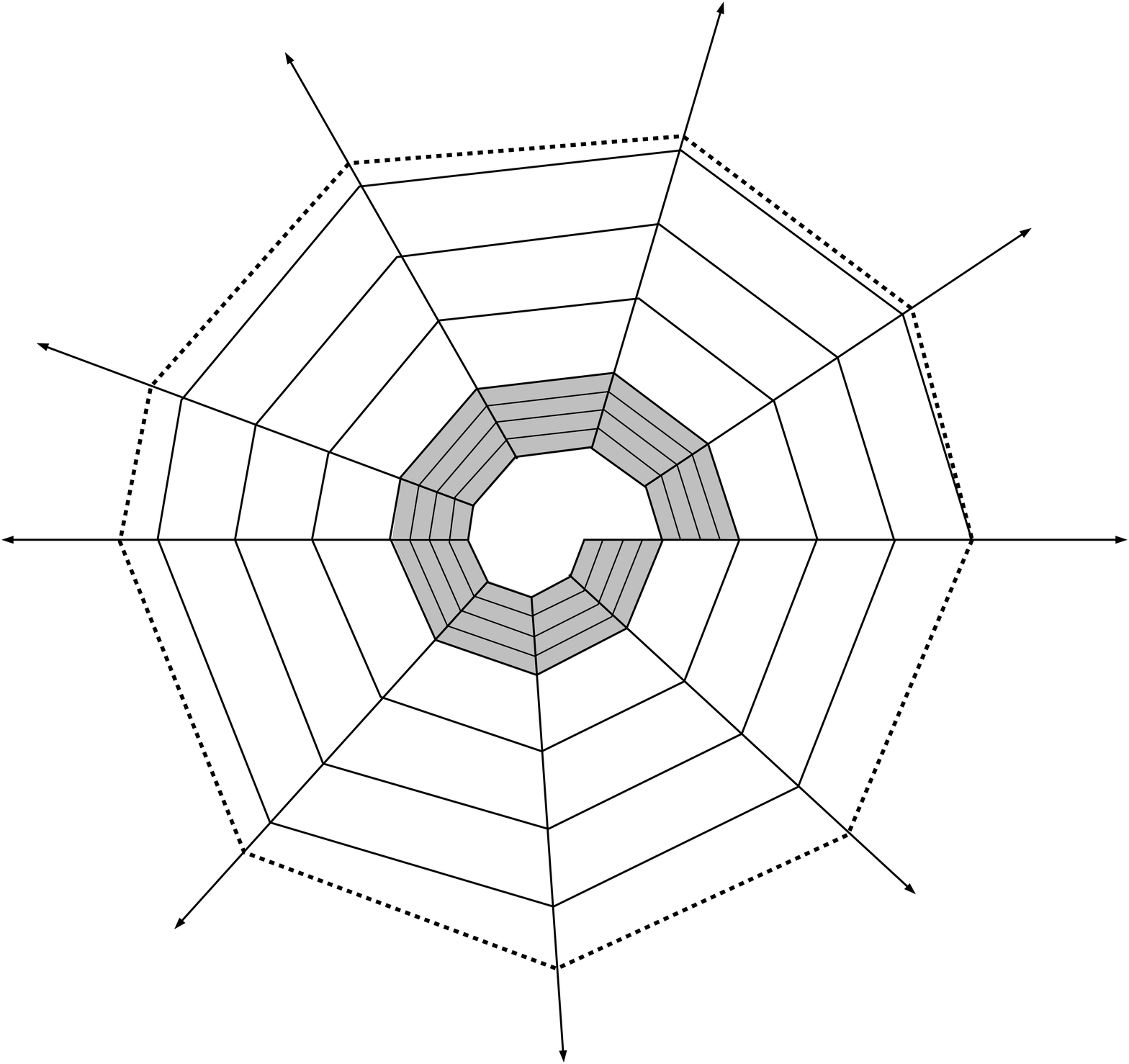}
$\hphantom{xxxx}$
\includegraphics[height=1.9in]{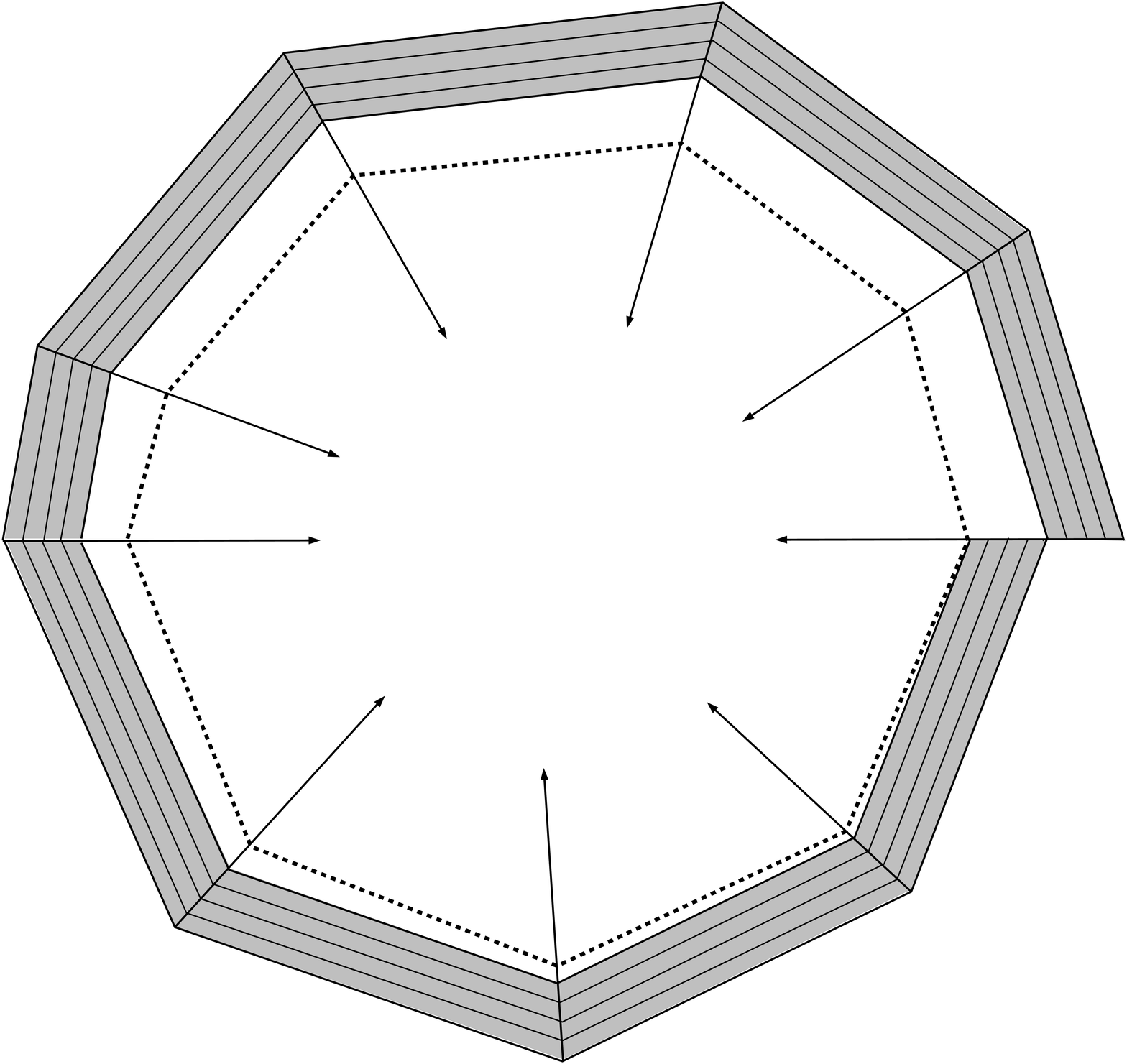}
 }
\caption{\label{SpiralMesh5}
If the winding number is  much larger than 
$  \frac 1 \theta$, then after $O(1/\theta)$
windings we can create a closed $\theta$-bent loop
(the dashed curve).
We then propagate the corners of the narrow inner
tunes until they hit a vertex of this closed loop.
We can also create a $\theta$-bent loop one 
winding in from the outer tube and use it similarly 
to  propagate the 
corners of the narrow outer tubes until they hit 
a vertex of this loop.  
}
\end{figure}

The propagation paths cut the spiral into $2\theta$-nice 
triangles and quadrilaterals. 
Moreover, as in the case of simple tubes, it is easy 
to check that each dissection piece is crossed 
by at most $1/\theta$ paths in the  construction 
of each spiral. Since there are $O(n)$ spirals, 
this means there are at most $O(n/\theta)$ such 
crossings of a dissection piece in total. The 
propagation paths that enter each spiral cross each 
dissection piece at most once, and there are 
$O(n/\theta)$ such paths in total, hence $O(n/\theta)$
such crossings of each piece. 

This completes the proof of Lemma \ref{mesh exists}, 
and hence the proof of Theorem \ref{Triangles}. 

\section{A lemma for quadrilateral meshing} \label{quad mesh lemma sec}

We now restate our conclusions 
in a form that is useful for proving the theorem 
on optimal quad-meshing in \cite{Bishop-quadPSLG}.
Readers interested only in Theorem \ref{NonObtuse}
may skip this section.

\begin{thm} \label{quad mesh lemma}
Suppose
 that $W$ is a polygonal  domain  with an isosceles trapezoid
 dissection with $n$  pieces.
Suppose also that   $0^\circ \leq \theta \leq 15^\circ$ and
 that  every dissection  piece is  $\theta$-nice. Finally, suppose
 the number of chains in the dissection is  $M$.
Then we can remove
$O(M/\theta)$ $\theta$-nice quadrilaterals of uniformly bounded
eccentricity from $W$ so that
the remaining region $W'$  has a $2\theta$-nice quadrilateral mesh
with $O( n M / \theta)$ elements.
At most $O(M/\theta)$ new vertices are created on the
$Q$-boundary of $W'$.
At most $O(M)$ vertices are created on the $P$-boundary
of $W'$, and no more than $O(1)$ vertices are placed in
any single $P$-side of any dissection piece of $W'$.
For this quad-mesh,
any boundary point  on a $Q$-side of $W'$  propagates to another
boundary point after crossing at most  $O(n)$ quadrilaterals.
\end{thm}


\begin{proof}
The proof is exactly the same as the argument of the last 
few sections, except for some slight modifications inside 
the return regions. 

For each return region we place $O(1/\theta)$ equally 
space points along the two $Q$-sides of the region 
and propagate these outside the return regions until 
they hit the boundary of $\Omega$ or hit the $Q$-side 
of some return region. There are $O(M/\theta)$ such 
paths and they generate at most $O(M n /\theta)$ 
quadrilaterals and $O(M/\theta)$ endpoints on 
$Q$-sides of $\partial \Omega$. 

First consider  return regions that are simple tubes. 
As before,  split each such region into $O(1/\theta)$ 
parallel sub-tubes and  so that in each sub-tube we 
can connect 
opposite corners by a $\theta$-path. Now, however, 
we remove a small quadrilaterals at a pair of 
 opposite corners of the 
tube. These quadrilaterals have one edge on a $P$-boundary
of the tube, one edge on a $Q$-end of the tube, one 
vertex in the interior of the tube and the two edges 
adjacent to this vertex are chosen to lie a $P$-segment
and a $Q$-segment. See Figure \ref{QMlemma1}.

\begin{figure}[htb]
\centerline{
\includegraphics[height=2.0in]{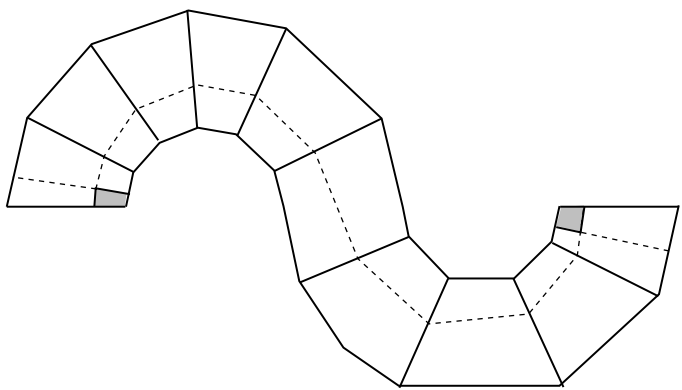}
 }
\caption{\label{QMlemma1}
We place quadrilaterals (shaded) at opposite corners 
of a tube, and connect the internal corners by a $\theta$-path.
Every path entering tube either  immediately 
hits the shaded quadrilateral at that end, or propagates 
to hit the shaded quadrilateral at the other end.  
We also have to propagate the interior corner of the 
shaded quadrilateral along a $Q$-path. This gives 
a mesh of every tube by  $2\theta$-nice quadrilaterals.
}
\end{figure}

We then connect the interior corner of each of the two 
quadrilaterals by a $\theta$-curve. This requires less displacement
than connecting the corners, so it is clearly possible
to do this  (to make 
it easier to see, we could always increase the number of 
tubes and decrease their width by a fixed factor). 
We have freedom in choosing the size of the quadrilaterals, 
and so we can arrange for all the quadrilaterals chosen 
in the same dissection trapezoid to have sides along  
the same $Q$-segment. Thus when we $Q$-propagate the 
corners of the quadrilaterals,  only two extra points 
will be created on the $P$-side of the dissection piece
containing such a quadrilateral.

If we  apply  quadrilateral propagation 
to  each  $P$-path entering the tube from either end, 
it crosses the tube and hits a $Q$-side of the removed 
quadrilateral at the other end of the tube.
 See Figure \ref{QMlemma1}.
This gives a $2\theta$-nice quadrilateral mesh inside
the modified  tubes.

 Inside the spirals we do a similar thing. In the 
previous proof,  paths 
inside spirals were terminated by bending them in a sub-tube 
of the spiral until they hit a corner on the opposite side 
of the tube from where they entered, in order to form a loop.
 so the same construction 
works. Outside the return regions, the $P$-paths convert the 
$\theta$-nice  dissection into a $\theta$-nice quadrilatal 
mesh (previously the only triangles created by the $P$-paths 
were in triangular pieces of the dissection, which we now assume 
don't exist). 
See Figures \ref{QMlemma2} and \ref{QMlemma3}.
\end{proof} 

\begin{figure}[htb]
\centerline{
\includegraphics[height=2.5in]{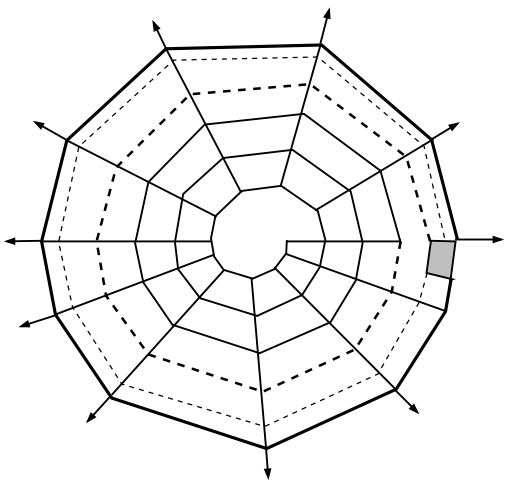}
 }
\caption{\label{QMlemma2}
For spirals with $ \gg \theta^{-1} $ windings, 
we can make a  $\theta$-path loop 
in the $j$th spiral when $j \gtrsim 1/\theta$ 
(solid thick curve). 
Then place a quadrilateral as shown with one 
edge on the loop; one corner is $\theta$-propagated 
around the spiral once to hit the center of a 
side of the same quadrilateral (thin dashed curve).
 The boundary of 
the spiral is $\theta$-bent to hit the other corner
(thick dashed curve).
}
\end{figure}

\begin{figure}[htb]
\centerline{
\includegraphics[height=2.5in]{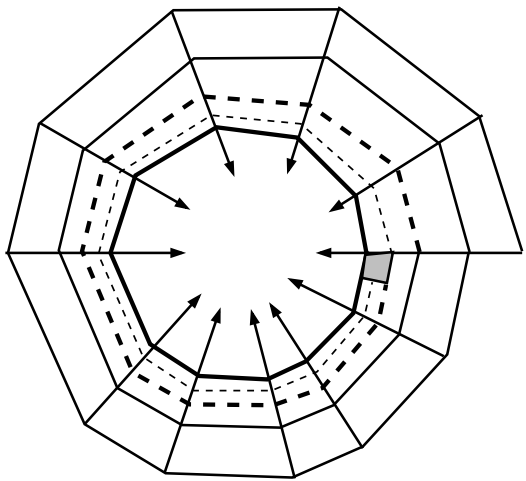}
 }
\caption{\label{QMlemma3}
The quadrilateral construction near the outside of 
a large spiral.
This is  similar to the  construction in 
Figure \ref{QMlemma2}, but we can do it in the sub-tube 
adjacent to the outermost one.
The two constructions give a $2\theta$-nice mesh of the entire 
spiral (minus the two quadrilaterals).
}
\end{figure}


\section{Overview of the proof of Theorem \ref{NonObtuse}} 
\label{intro NonObtuse}

The remainder of the paper gives 
the proof of Theorem \ref{NonObtuse}. 
 In this section we give the  overall 
strategy of the proof  and we
 will provide the details in 
the following sections. 

The proof combines ideas already seen in the proofs of 
Theorems \ref{Refine Triangulation} and \ref{Triangles}
but requires a different displacement estimate 
in tubes and  a more intricate construction in the 
spirals.
As explained in Section \ref{Gabriel edges}, it suffices to prove
Theorem \ref{ETS Gabriel}: assume $\Gamma$ is a triangulation 
and  show we can place $O(n^{2.5})$ points along the edges so 
that each triangle becomes Gabriel. As in the proof of 
Theorem \ref{Refine Triangulation} we start taking
the dissected domain $\Omega$ to be the original 
triangles $\{ T_k\}$  with the central triangles
$\{ T_k'\}$  removed 
(recall  the vertices of $T_k$  are the three points
where the inscribed circle touches the triangle $T_k$). 
We do not use the ``approximate circular-arc triangles'' 
that were used in the proof of Theorem \ref{Triangles}.

 For each triangle 
$T_k$, remove the closed triangle $T_k'$ as in Section 
\ref{Refine tri}.  
As before, $T_k \setminus T_k'$ is a union of 
three isosceles triangles. Keep the isosceles triangles
with angle $< 90^\circ$; as explained at the end of 
Section \ref{Refine tri},
 isosceles triangles with angles $\geq 90^\circ$
can be ignored because adding any set of points to the 
$Q$-edges will make the triangle Gabriel. The remaining region 
$\Omega $ thus has an isosceles dissection by 
$O(n)$  acute triangles.
We construct return regions for $\Omega$ just as before.

Each vertex of each $T_k'$ are propagated by 
$P$-paths until then leave $\Omega$ or hit 
the $Q$-side of a return region. This creates
$O(n^2)$  crossing points on $\Gamma$. 

If  a return region has $k$ isosceles pieces
then we will place $O(\sqrt{k})$ even spaced 
points in each $Q$-end of the region and propagate 
these until they leave $\Omega$ or hit a return 
region. 
Since $k = O(n)$, this  creates at most $O(n^{2.5})$ 
new points. If different return regions 
had to use distinct isosceles pieces this 
estimate  would be $O(n^2)$ instead.  
Improving the exponent in Theorem \ref{NonObtuse} 
seems to be entirely a matter
of understanding the behavior of distinct 
return regions that share isosceles pieces.

Why do we split the $Q$-ends of the return 
regions into $O(\sqrt{k})$ pieces? 
When we bend the $P$-paths inside the 
return regions, 
we must verify that the Gabriel condition is 
satisfied by the points that we generate.
This is a more restrictive condition than the 
$\theta$-bending of the earlier proof, so  
paths can be bent less and hence take a more 
steps to terminate.
The difference is  illustrated 
in  Figure \ref{Bending1}. The left side   
shows the range of options for 
a $\theta$-segment crossing a single rectangle; the allowable 
displacement is roughly $\theta|a-b|$.
The center and right  pictures 
of Figure \ref{Bending1}
 show the restrictions on a Gabriel path. Note that
there are two such restrictions: the exit point $b$ must be 
between the Gabriel disks tangent at $a$ and the entrance point 
$a$ must be between the disks tangent at $b$. This restricts 
$b$ to an interval of length approximately $|a-b|^2/w$, 
where $w$ is the width of the piece. This estimate 
will be made more precise in the next section; the main  point 
is that it  shrinks quadratically with  $|a-b|$ whereas
the estimate for $\theta$-paths  decreased linearly. Thus
the proof of Theorem \ref{NonObtuse} requires longer, 
narrower tubes than the proof of Theorem \ref{Triangles}.

\begin{figure}[htb]
\centerline{
\includegraphics[height=2.5in]{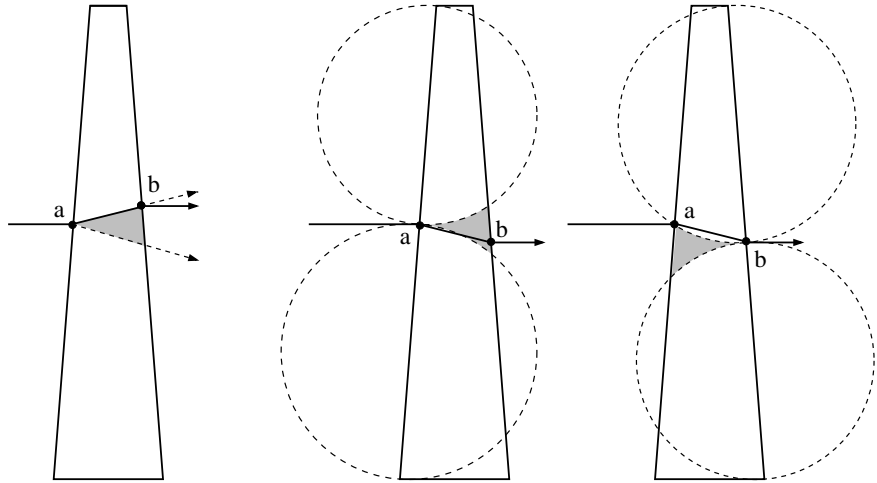}
 }
\caption{\label{Bending1}
A $\theta$-bent path can reach any point defined by 
a cone with angle $2 \theta$, but a Gabriel bent 
path can only  reach points defined by the cusp between 
two tangent disks. Moreover, this is a two part 
condition: the exit point must be the cusp defined 
by the entrance point and the entrance point must
be in the cusp defined by the exit point.
}
\end{figure}

To illustrate the idea, consider a simple case: 
a square divided into $k$ thin parallel 
rectangles. See Figure \ref{Bending3}. 
A Gabriel path crossing the square takes $k$ steps,
each with displacement
$ \simeq 1/k^2$, so  the total displacement is $\simeq 1/k$.
At first glance, this  seems  to 
say  we should cut the square into $O(k)$ parallel 
tubes; then we could get all entering paths to terminate before
hitting the far  side of tube. This works, but  leads to the 
estimate $O(n^3)$ in Theorem \ref{NonObtuse}.

\begin{figure}[htb]
\centerline{
\includegraphics[height=2.0in]{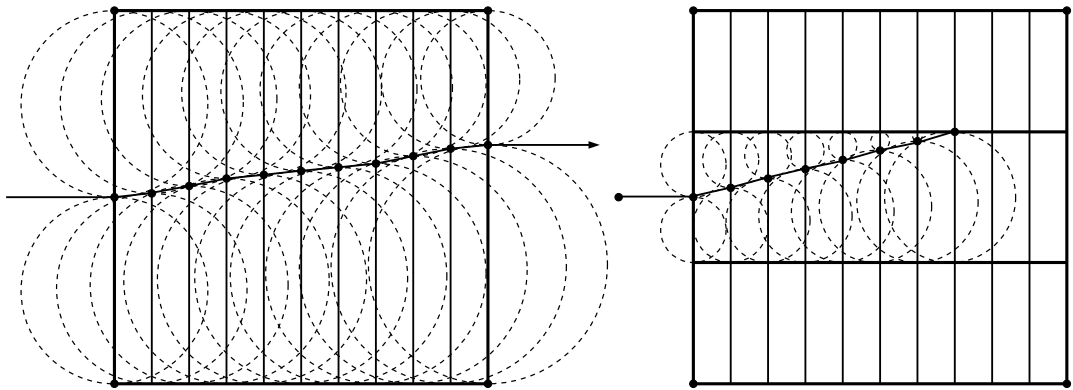}
 }
\caption{\label{Bending3}
A Gabriel-bent path must stay outside certain pairs 
of tangent disks. When we cut a unit square into 
$\frac 1k \times 1$ rectangles  then  a Gabriel  path 
has total displacement at most 
$O(1/k)$. If we cut the square 
into $O(\sqrt{k})$ horizontal tubes,  paths  in each tube 
take $k$ steps with displacement of $k^{-3/2}$ and hence 
total displacement of $1/\sqrt{k}$. Since this 
is also the width of the tube, we can Gabriel-bend a 
path to hit the side of a tube before leaving it.
}
\end{figure}

We can do better. Cut the square into $ \sqrt{k}$ 
tubes instead. Now the tangent disks have diameter $k^{-1/2}$
and the cusp regions where we choose our next point have 
height $\simeq \frac 1{\sqrt{k}} (\frac 1{\sqrt{k}}/\frac 1k)^2
 =  k^{-3/2}$. Thus a Gabriel path 
takes $k$ steps,  each with displacement $\simeq k^{-3/2}$,
for a total displacement  $\simeq k^{-1/2}$, which 
is  the approximate width of the tube. Thus using only 
$O(\sqrt{k})$ tubes, we can bend Gabriel paths enough to  
hit the far corner of the tube (and thus terminate).

We shall prove in the next two sections this holds 
for any  return regions that are simple tubes,
 not just squares with rectangular pieces.
Each return region that is a C-tube, S-tube 
or simple G-tube will be split into  $O(\sqrt{k})$
narrow parallel tubes and the entering propagation paths 
will be  Gabriel bent until they  a far corner of
the tube; here  $k$  
is the number of isosceles pieces forming the tube.

We also place $O(\sqrt{k})$ 
narrow parallel sub-tubes at the two ends of spiral 
return regions, i.e., we subdivide the innermost and 
outermost windings of the spiral. As with   simple 
tubes, all  paths entering the spiral  can be 
bent within these narrow tubes to terminate within $O(k)$ 
steps. But then we 
have to propagate both the external  and internal corners 
of the narrow tubes. The external corners propagate outside 
the spiral until they terminate  just as described for 
the corners for narrow tubes  in the 
previous paragraph. 

The most difficult part  of the proof of Theorem 
\ref{NonObtuse} is dealing with the  $\sqrt{k}$ internal 
corners that  propagate through the spiral; since we have no 
bound for the number  $N$ of windings of the spiral in terms 
of $n$, this could produce arbitrarily many new vertices.
Thus  propagation paths of the internal corners  must be bent 
to  terminate earlier. Consider the case of paths that start 
near the inner end of the spiral (the outer part is handled 
in the same way, but is easier, since the windings of the
spiral are longer). We consider what happens for very large 
spirals (where the number of windings $N$ is bigger than 
the number $k$ of isosceles pieces in the spiral; for smaller 
values of $N$ we truncate the construction at the appropriate stage.)

We first bend the propagation paths  so that adjacent paths merge,
and then merge adjacent merged paths, and continue until all 
the propagation paths generated by the $O(\sqrt{k})$ internal 
corners have merged into a single path. This occurs around 
winding $k^{1/3}$. This path is then propagated as 
a $P$-path out to winding $k^{1/2}$. 
See Figure \ref{BigPicture}.

\begin{figure}[htb]
\centerline{
\includegraphics[height=4in]{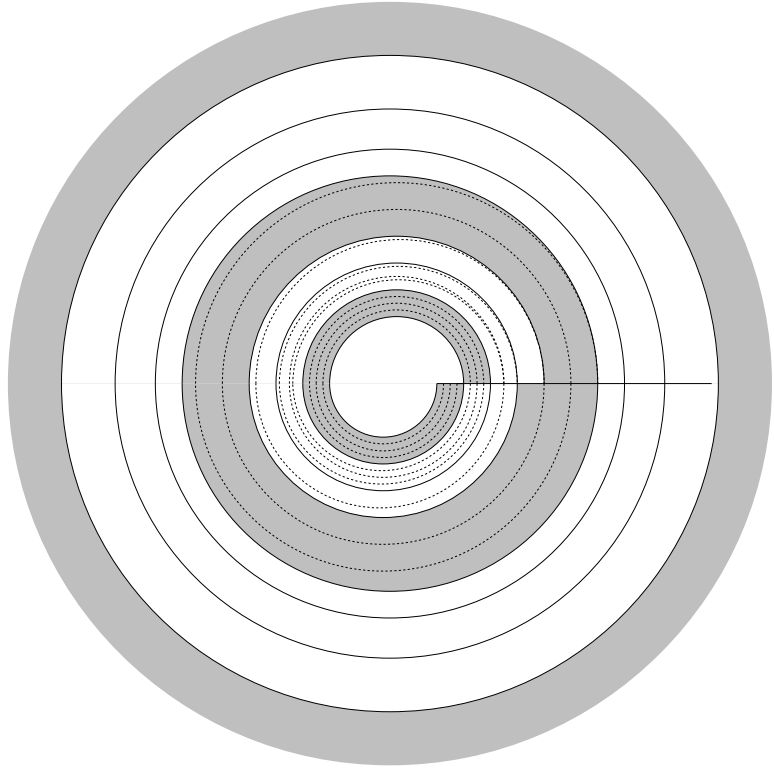}
 }
\caption{\label{BigPicture}
This illustrates the stages in the spiral construction.
 First (light gray), we divide the tube
into thinner tubes and entering paths are bent to hit
the sides of these. Next  (white), 
the thin tubes are bent and collapsed 
in pairs; in this figure four tubes are merged into 
one after two windings. In the third stage (gray),
the single tube is propagated until we can bend it 
to intersect itself. Next (white) comes a 
sequence of closed loops that gradually grow further 
apart. 
Finally we reach the empty region (gray), where no 
paths propagate.
This figure gives a rough idea of the construction,
 but scales have 
been drastically compressed to make all the stages visible 
in the same picture.
}
\end{figure}

 At this stage we have 
enough freedom to bend the curve to hit itself, forming a closed
loop that wraps once around the spiral.  
This is similar to 
what we did in the proof of Theorem  \ref{Triangles}, but 
in this case,  in order for this closed loop to  be Gabriel, 
there must be  another (larger)  closed loop parallel to it.
This did not occur in Theorem \ref{Triangles}.  This 
constraint requires us to   construct a sequence of 
parallel closed  loops in the spiral between windings $k^{1/2}$
and $k$. The closed loops gradually can become farther and 
farther apart; only $O(\sqrt{k})$ loops are used in all.
At winding $k$, there is no need 
for  a ``next'' loop  and the sequence of closed loops ends.
 The part of the spiral beyond winding $k$ is  an ``empty'' region
until we reach a closed loop coming from  the analogous 
construction in the outer half of the spiral.

In the remainder of the paper we give the details 
of the argument sketched above.

\section{Gabriel bending in isosceles pieces} \label{Perturb sec} 

This section contains the main estimate   used in
proof of Theorem \ref{NonObtuse}.

Suppose $a$ and $b$ are endpoints of a $P$-path  in an 
isosceles piece $T$.
If  we keep $a$ fixed,  how far we can move $b$ and still have the
Gabriel condition hold? More precisely, can we 
 find an  $\epsilon >0$ so that all points on  the 
same $Q$-side as $b$  that 
are within distance $\epsilon$ of $b$ can be connected 
to $a$ by a   Gabriel segment?  If this holds 
we say that the {\defit allowable displacement} 
for the piece is at least $\epsilon$. 

\begin{lemma} \label{G tube bending}
Suppose $T$ is an isosceles piece  of width $w$.
 Suppose   $[a,b]$ is a
$P$-segment crossing $T$ and $R$ is the distance of 
$a$ from the vertex of the piece ($R=\infty$ if the 
piece is a rectangle). Then if $c$ is a point on the 
same $Q$-side as $b$ and is within distance 
\begin{eqnarray} \label{displacement est}  
\epsilon = 
\frac {|a-b|^2}{4} \max(\frac 1w, \frac 1R),
\end{eqnarray}
of $b$, then $[a,c]$ is a Gabriel segment crossing $T$. 
In particular, the allowable displacement is 
at least $\epsilon$.
\end{lemma}

\begin{proof} 
First suppose the isosceles piece is a rectangle
($R = \infty$).
Consider disjoint  sub-segments of the $Q$-side containing
 $a$ that have $a$ as a common endpoint and consider 
the disks  $D_1, D_2$ with these segments as diameters. See 
Figure \ref{GabrielBend1}.  Assume these disks have
 radii $r$ and $s$. The diameters of these disks 
are  disjoint segments  that both lie on the same 
non-base side of an isosceles piece, so their
length adds up to be less than the width of the 
piece, i.e., $ 2r+2s \leq w$. Thus  $\max(2r,2s) \leq w$.

\begin{figure}[htb]
\centerline{
\includegraphics[height=1.5in]{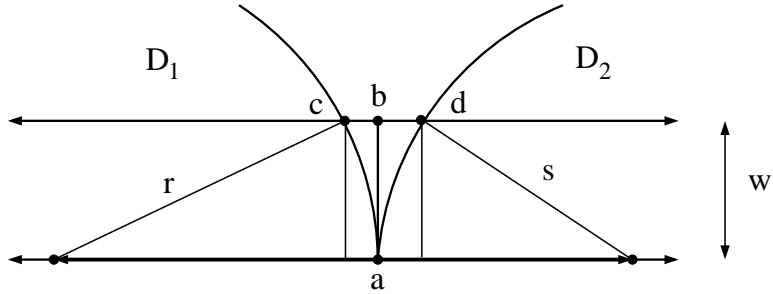}
 }
\caption{\label{GabrielBend1}
The segment $[a,b]$ is a $P$-segment for a rectangular 
piece. A simple estimate shows  
$|c-b| \geq |a-b|^2/r$ and $|d-b|\geq |a-b|^2/s$.
}
\end{figure}

Suppose 
$[c,d]$ is the $Q$-segment containing $b$ that is 
disjoint from these disks (again, see Figure 
\ref{GabrielBend1}). 
We want to estimate  $|c-b|$ and $|d-b|$ from below.  
Such a lower bound gives  the desired  lower bound 
on the  allowable displacement.

If the disk $D_1$ is too small, i.e.,  $r < |a-b|$, 
 then the disk $D_1$ does not hit the 
$Q$-side containing $b$ and the Gabriel condition is 
automatically satisfied.  Thus we may assume 
$r \geq  |a-b|$. 
Then by the Pythagorean theorem 
$$ |c-b| = r- \sqrt{ r^2 -|a-b|^2} ,$$
or  (using $1-\sqrt{1-y} \geq y/2$ on $[0,1]$), 
$$  \frac 1r |c-b| = 1- \sqrt{  1 - \frac{|a-b|^2}{r^2} } 
         \geq  \frac{|a-b|^2}{2 r^2}   $$
so  $  |c-b| \geq |a-b|^2/2r$.
The calculation for the other disk is identical, so the 
two disks omit  all points within distance  
$$  \epsilon = |a-b|^2 \min( \frac 1r, \frac 1s)  \geq 
\frac{|a-b|^2}{2w}.$$
of  $b$.
Since $1/R=0$ in this case, this implies 
(\ref{displacement est}).

Next we consider what happens when the piece is not 
a rectangle. To be concrete, we assume one $Q$-side
lies on the real axis, the  vertex of the piece is 
at $-R$ and the $P$-path connects $a=0$ to $b$ in 
the upper half-plane. 
Suppose the piece has angle $\theta$.
Some elementary trigonometry shows that the disk $D_1$ 
does not hit the $Q$-side containing $b$ if 
(see Figure  \ref{GabrielBend3})
$$ r <   (R-r)\sin \theta. $$

\begin{figure}[htb]
\centerline{
\includegraphics[height=2.0in]{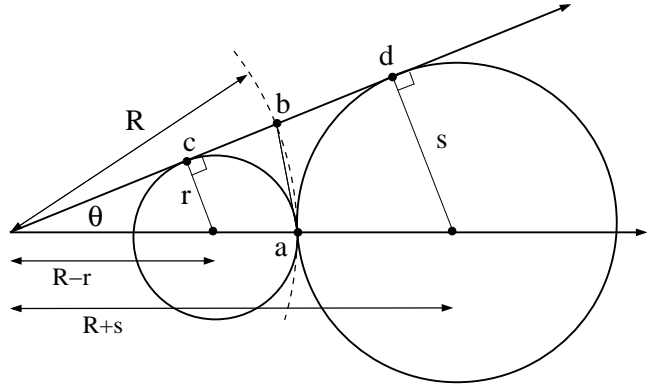}
 }
\caption{\label{GabrielBend3}
If $r,s$ are small enough then the Gabriel disks 
for one $Q$-side don't hit the other $Q$-side.
}
\end{figure}

Since $|a-b| = 2 R \sin \theta/2$, this is equivalent to
\begin{eqnarray} \label{r ineq 1}
 r < R\frac {\sin \theta}{1 + \sin \theta} 
      =    \frac {|a-b|\sin\theta}{2\sin(\theta/2) (1 +\sin \theta)}.
\end{eqnarray} 
By the double angle formula,  for $0 \leq \theta \leq \pi/2$ 
we have  $\sin \theta = 2 \sin \frac \theta 2 \cos  \frac \theta 2$, so
 \begin{eqnarray} \label{sin est}
  \frac  {\sin\theta}{2\sin(\theta/2) (1 +\sin \theta)} 
  =  \frac  {\cos\theta/2}{(1 +\sin \theta)} 
   \geq   \frac  {\cos \pi/4}{(1 +\sin  \pi/2 )} 
            =  \frac 1{2 \sqrt{2}}.
\end{eqnarray}
Hence (\ref{r ineq 1}) holds if $r < |a-b|/2\sqrt{2}$.  
If this condition  holds, then  the point $c$ is a corner of
the piece $T$, so the estimate holds trivially to the
left of $b$. Therefore 
we may assume $r \geq |a-b|/2 \sqrt{2}$ in what follows.
Note also that 
this implies   $w \geq |a-b|/\sqrt{2}$ since $r\leq w/2$.

A similar calculation shows $D_2$ does not hit the 
opposite $Q$-side if 
$$ s <   \frac {|a-b|\sin\theta}{\sin(\theta/2) (1- \sin \theta)}
  = 2 |a-b| \frac {\cos \theta/2}{1-\sin \theta}.$$
See Figure \ref{GabrielBend3}. 
The trigonometric
 function on the far right is increasing for $\theta 
\in [0, \pi/2]$
(compute the derivative), so it takes its minimum 
value $1$ at $\theta =0$.
Hence $D_2$ does not hit the opposite $Q$-side if 
$s < 2 |a-b|$.
In this case $d$ is a
corner of the isosceles piece and the lemma holds 
trivially to the right of $b$.
 Therefore we may assume $s \geq |a-b|$ in what follows.

Now suppose we have an isosceles piece with angle 
$\theta >0$. We normalize the picture as in 
Figure \ref{GabrielBend5} with one $Q$-side along
the real axis, the vertex of the piece at $-R$. 
The other $Q$-side is labeled $L$.
We consider a $P$-path with one endpoint at 
the origin and the other endpoint (labeled $b$ in the 
figure)  on $L$ in the upper half-plane. We also consider 
disks  $D_0$, $D_1$, $D_2$ 
centered at points $-R$,$-r$, $s$ on the real
line that are tangent at the origin. 
We let $[c,d]$ be the segment of $L \setminus (D_1
\cup D_2)$ that contains $b$.
 See Figure \ref{GabrielBend5}.

\begin{figure}[htb]
\centerline{
\includegraphics[height=2.5in]{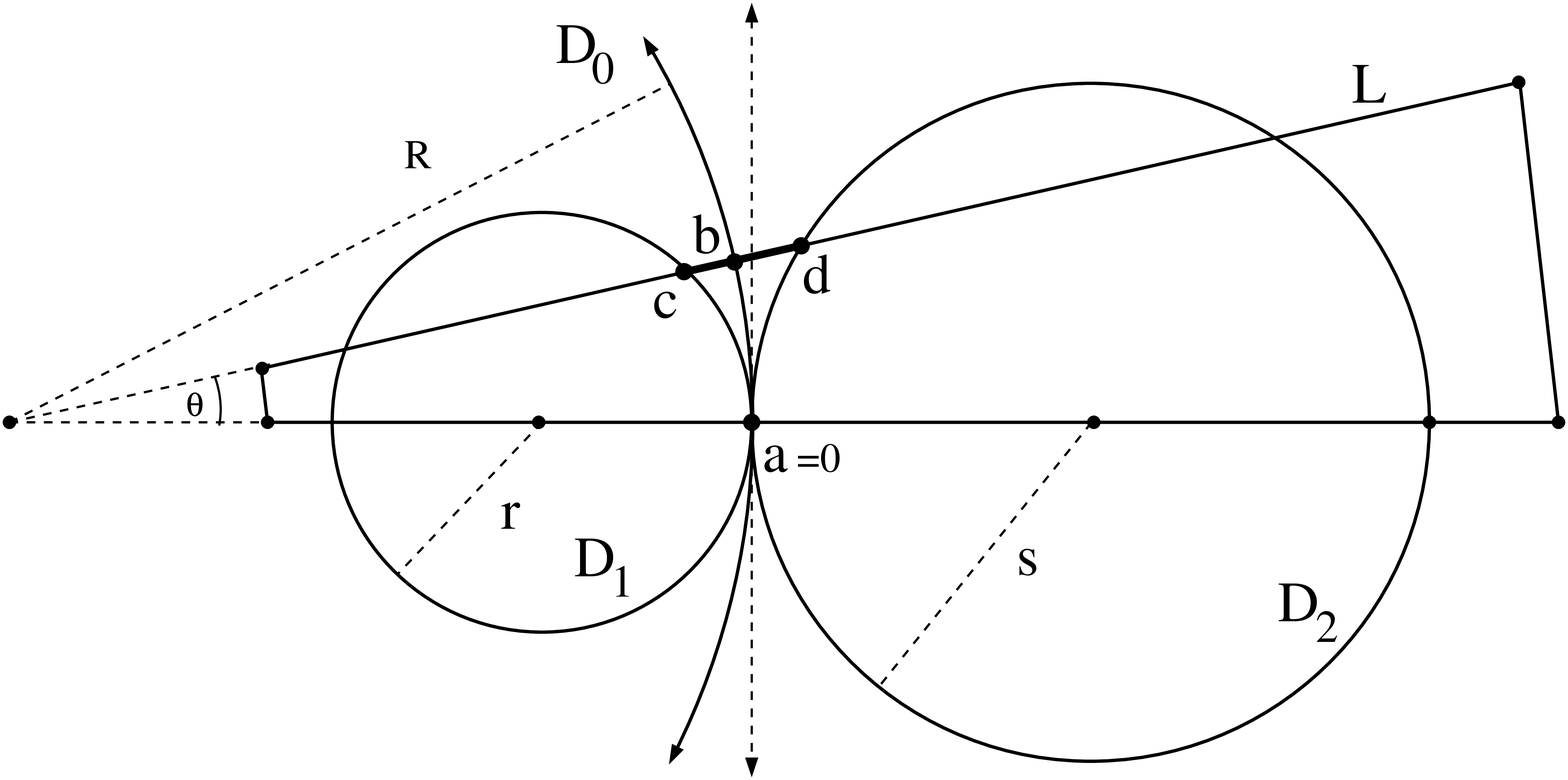}
 }
\caption{\label{GabrielBend5}
We want a lower bound on $|b-c|$ and $|b-d|$ 
in terms of $|a-b|$, $r$, $s$ and $R$. We prove this 
by applying the transformation $z \to 1/z$ to this 
picture, to get the picture in Figure 
\ref{GabrielBend7}.
}
\end{figure}

\begin{figure}[htb]
\centerline{
\includegraphics[height=3.5in]{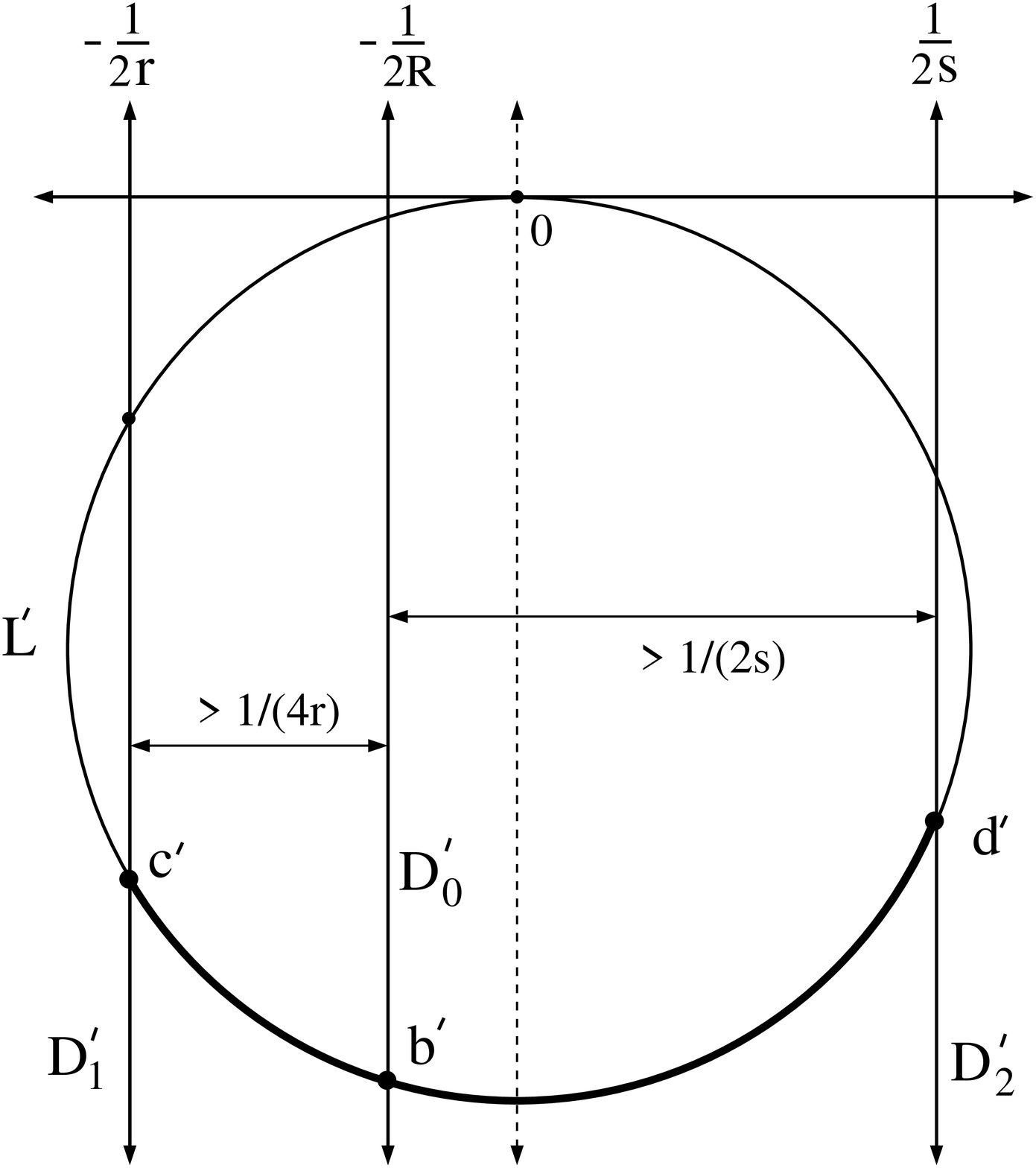}
 }
\caption{\label{GabrielBend7}
This is the inversion of Figure \ref{GabrielBend5}. 
The point $a=0$ maps to $a'=\infty$ and 
the circles through $0$  map to vertical lines, 
whose distance apart is easy to compute. The line 
$L$ maps to a circle $L'$. 
}
\end{figure}

Now apply the map $z \to 1/z$. The origin is mapped to 
infinity and the circles centered on the real line 
passing  through $0$ now map to vertical lines.
(The map $z \to 1/z$ is a linear fractional transformation
that maps $0$ to $\infty$, so circles through 
$0$  map to circles through $\infty$, i.e., lines.) 
See Figure \ref{GabrielBend7}.
The boundary of $D_1$ maps to $L_1=\{ x= -1/2r\}$,
the boundary of $D_2$ maps to $L_2=\{ x= 1/2s\}$
and the boundary of $D_0$ maps to $L_0= \{ x = 1/2R\}$. 
Since  $R \geq 2r$, we have the distance between 
$L_0$ and $L_1$ is at least $1/4r$.

The line $L$ is distance $ \rho= R \sin \theta$ from the origin, 
whereas $|a-b| = 2 R \sin(\theta /2), $
so  by a similar calculation to (\ref{sin est}), we get
$$ \rho=  \frac {|a-b| \sin \theta}{2 \sin (\theta/2)} 
       \geq \frac{|a-b|}{\sqrt{2}},$$
since we assume $ 0 \leq \theta \leq \pi/2$.
This means that on the line $L$, the derivative of 
$1/z$ is bounded above by $2/|a-b|^2$. Therefore the 
image of the segment $[c,b]$ has length at 
most $ 2  |c-d|/|a-b|^2  $. Thus
$$ |c'-b'| \leq 2 \frac {|c-d|}{|a-b|^2}.$$
 However,  the image  of this 
segment is a circular arc that connects the lines $L_1$
and $L_0$ and hence has length  at least 
$$
    |c'-b'| \geq |\frac 1{2r} - \frac 1{2R}| \geq \frac 1{4r}
$$
since $r  \leq R/2$.
Combining these inequalities gives 
$$ 2   \frac{|c-b|}{|a-b|^2}   \geq  \frac 1{4r}.$$
Since $r \leq w/2$ and $r \leq R/2$ this gives 
$$
 |c-b| \geq   \frac{|a-b|^2}{8 r}
\geq \max(   \frac{|a-b|^2}{4w}, \frac{|a-b|^2}{4R})
.$$
Similar calculations  show 
$$  \frac { 2|d-b| }{ |a-b|^2}   \geq |b'-d'|,$$
and 
$$ |b'-d'| \geq \frac 1{2R} + \frac 1{2s} 
\geq \frac 1{2R} + \frac 1{w} 
$$
from which we deduce 
$$ |d-b| \geq    \frac {|a-b|^2}4 ( \frac 1{R} + \frac 1w)
 \geq    \frac {|a-b|^2}4 \max( \frac 1{R} , \frac 1w).$$
\end{proof}

\section{Gabriel bending in tubes} \label{Gabriel tubes} 

Next we apply the Cauchy-Schwarz inequality to our 
displacement estimate  for pieces, to get a
displacement estimate for tubes:

\begin{lemma} \label{bend in tube} 
Suppose $T$ is  a tube of width $w$, minimal-length 
$L = \tilde \ell(T) $ and  consists of $p$ isosceles pieces. 
 Let $a,b$ be points on opposite ends 
of $T$ that are connected by a $P$-path in $T$. Then $a$ can 
be connected  by a Gabriel path to any point on the opposite 
end of $T$ that is within distance $d=L^2/(4pw)$ of $b$.
\end{lemma}

\begin{proof}
For the $j$th piece in the tube, 
 let $\ell_j$ be the length of the piece (the 
shorter of its two base lengths, zero for triangles).
Then  $L = \sum_j \ell_j$ by definition.
By  the Cauchy-Schwarz inequality
$$ L^2  = (\sum_j \ell_j)^2 
       \leq (\sum_j 1^2) \cdot (\sum_j \ell_j^2) =
       p \sum_j \ell_j^2  .$$
The  allowable displacement of each piece is at least $\ell_j^2/4w$, 
so the total  allowable displacement is at least  
$$ \sum_j  \frac {\ell_j^2} {4w} \geq \frac {L^2}{4p w}.$$ 
\end{proof} 

\begin{cor} \label{corner cor}
If $T$ is a tube with minimal-length $L$, width $w$ composed 
of $p$ pieces and $ w  \leq L/2\sqrt{p}$, then opposite corners 
of the tube can be connected by a Gabriel path inside the tube.  
\end{cor} 

\begin{proof}
By assumption we  have $L \geq 2w \sqrt{p}$ so the 
allowable displacement is  at least
$$   \frac {(2w\sqrt{p})^2}{ 4 w p}  \geq  w,$$
where $w$ is the 
width of the tube. Thus opposite corners can be connected. 
\end{proof}

\begin{cor} \label{narrow tubes} 
Suppose $T$ is a C-tube, S-tube or simple G-tube composed
of $p$ isosceles pieces  and we 
cut $T$ into $M$ parallel, equal width sub-tubes with $M 
\geq  8 \sqrt{p}$. Then the opposite corners of each sub-tube 
can be connected  by a Gabriel path in that sub-tube.
\end{cor}

\begin{proof}
Since $M \geq 2$, each sub-tube is half of a wider tube inside $T$
and hence has length that is at   most four times its 
minimal length by Lemma \ref{sum short}. 
Moreover, the length of each sub-tube is 
bounded below by the length of $T$. 
  By Corollary \ref{long tubes}
 every simple return tube has length 
that is at least its  width $w$ and hence 
its  minimal-length is at least $w/4$.
Thus each sub-tube has minimal-length 
at least $M/4 \geq 2\sqrt{p} $ times longer than its width. 
The conclusion then follows from  Corollary \ref{corner cor}.
\end{proof}

\section{Gabriel bending in spirals} \label{G bend spirals} 

Finally we have to consider bending in pure spirals. 
This is the final, but most complicated, step in the proof
of Theorem \ref{NonObtuse}, and we break the construction 
into several steps described in different sections. 

\begin{lemma}
Suppose $S$ is a pure spiral made up of at most 
$k$ isosceles pieces.  Then we can mesh the interior of 
the spiral using at most $O(k^{1.5})$ quadrilaterals
and triangles so that the added vertices make every isosceles
piece in the spiral Gabriel. 
Also,  every path entering the spiral can 
be Gabriel bent to terminate within one winding.
The bound  on the number of points added 
is independent of $N$,  the number of windings 
of the spiral.
\end{lemma}

Without loss of generality we will assume the bound
$k$ on the number of isosceles pieces is a power 
of $4$, $k = 4^K = 2^{2K}$ and that it is at least 
16 times larger than that actual number of isosceles
pieces in the spiral.   We do this so that we can 
apply  Corollary \ref{narrow tubes} with the value $M= \sqrt{k}$ 
instead of $M =16 \sqrt{k}$. This will make notation 
slightly easier and does not affect the statement
of the  lemma.

For simplicity, we 
rescale so that the entrance and exit segments of the 
spiral have length $1$, i.e., the width of the tubes 
in the spiral is $w=1$.
The spiral is a topological 
annulus, with one bounded and one unbounded 
complementary component. The two ends of the spiral corresponding 
to these two regions  will be 
called the inner-end and outer-end respectively.
The spiral is a union of $N$ simple G-tubes joined 
end-to-end. We number these  consecutively 
starting at the inner-end and denote them $T_1, 
\dots, T_N$. We will call $T_1$ and $T_N$ the {\defit innermost} 
and {\defit outermost} tubes (or windings) respectively. 

If the number of windings satisfies  $1 < N \leq 100 $
then we can treat the spiral like a simple tube. 
We divide both  $T_1$ and $T_N$  into 
$\sqrt{k} =2^K$ parallel sub-tubes just as we did for simple 
G-tubes. We call these the {\defit narrow tubes}. 
Then every $P$-path entering the  innermost tube 
either hits a corner of one of the  narrow tubes, or it 
enters one of the narrower sub-tubes. In the latter case, it
can be Gabriel bent to hit a far  corner of this  narrow 
sub-tube. Similarly for paths entering the outermost tube.
The $\sqrt{k} $ corners of the narrow tubes are propagated 
as $P$-paths through the spiral until they hit the corners 
of the narrow tubes at the other end. This creates 
at most $O(k^{1.5})$ vertices in the spiral.
See Figure \ref{NarrowTubes2}.

\begin{figure}[htb]
\centerline{
\includegraphics[height=2.0in]{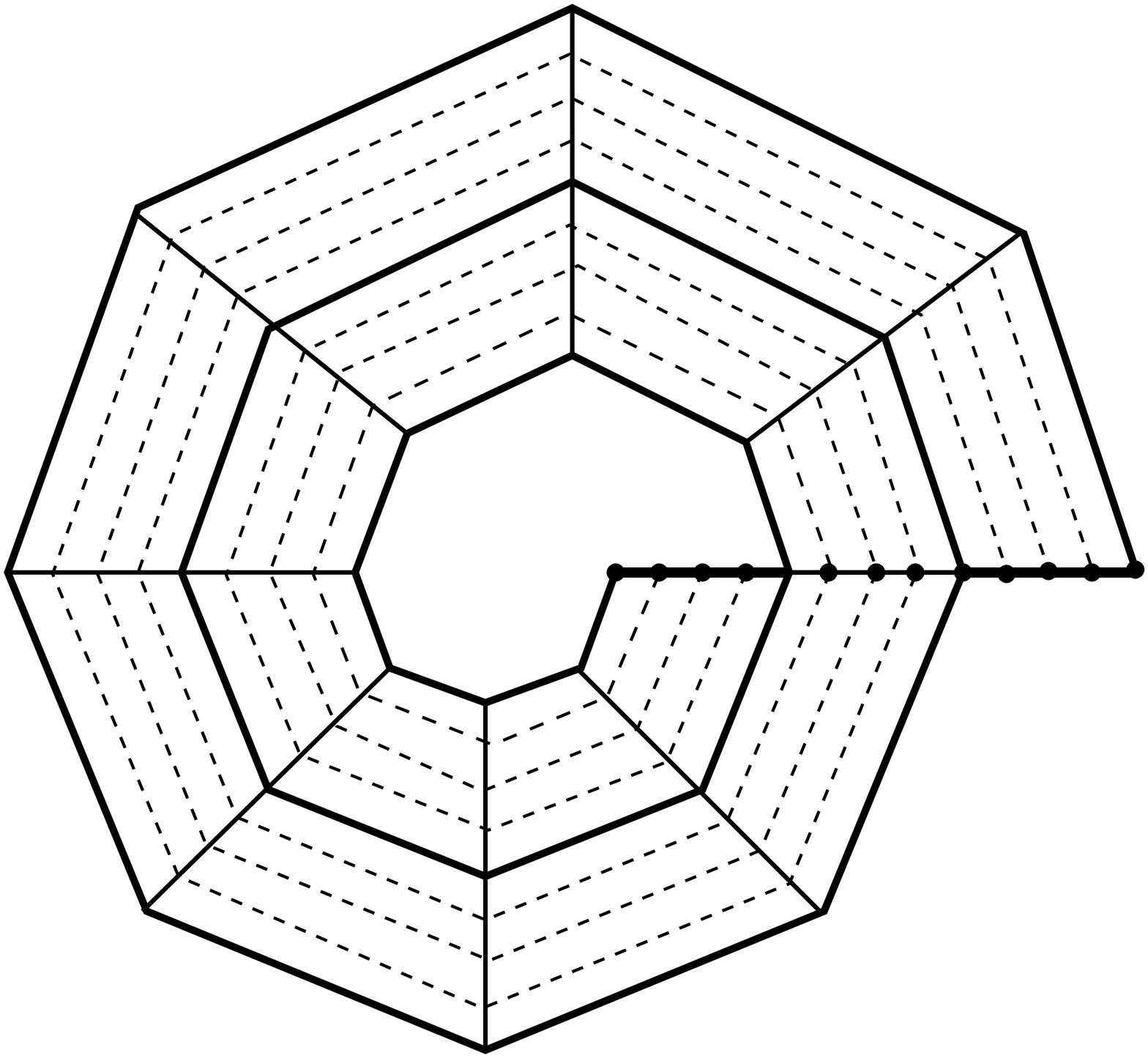}
 }
\caption{\label{NarrowTubes2}
For small spirals (less than 100 windings)
 we divide the tubes into $\sqrt{k}$
narrow sub-tubes. Every entering path can be Gabriel
bent to terminate within one winding.  Here 
we show $N=2$.
}
\end{figure}

For larger $N$ there are four different types of 
construction we will use, divided by certain ``phase
transitions''  in our ability to bend curves as $N$ 
increases. The transitions occur at $N \simeq 
k^{1/3}$, $k^{1/2}$ and $k$, as described in the next few sections.

\section{ $100< N \leq 16 k^{1/3}:$ Dyadic merging of paths}
 \label{dyadic merging}

We now consider the case {\bf $100< N \leq  16 k^{1/3} $}.
As above, we subdivide the innermost and outermost tubes 
into $\sqrt{k} = 2^K$ narrow, parallel sub-tubes, we terminate 
all entering paths within these narrow sub-tubes and 
we propagate the $2^K$ corners of the narrow sub-tubes
through the spiral. 
However, if we  propagate each  corner
of a narrow tube as a separate path 
all the way through the spiral we 
will generate too many new vertices. The 
solution is to  merge paths as they propagate 
through the spiral.

The general  idea is to take pairs of adjacent paths and bend them 
towards each other until they meet; this reduces the number 
of paths from $\sqrt{k}$ to $\sqrt{k}/2$. We then define new 
adjacent pairs and bend these towards each other, until 
they merge, creating $\sqrt{k}/4$ paths.  
We continue 
merging pairs of paths, reducing the number of paths by 
a factor of two each time, but needing more steps to 
accomplish this each time. Since each  pair of paths 
to be merged at some stage starts 
twice as far apart as pairs in the previous stage, we need 
a total displacement that is twice as long  as before.
Moreover, since 
the tubes are twice as wide, the allowable displacement 
is half as large per piece. Since the formula 
for displacement in Lemma \ref{bend in tube} involves
a factor of $L^2$, the length only needs to double to 
merge the next set of tubes.

  However, this does 
not mean we need twice as many windings. Since 
the windings of the spiral become longer as we move 
outwards, fewer windings are needed to achieve a
given length, so  the number of pieces crossed grows, but 
does not double, i.e., it grows at a geometric rate that 
is strictly less than $2$. This will allow us to obtain 
the desired bound.  We do the same construction 
starting at both ends of the spiral, so that when paths 
meet at the center of the spiral, the surviving paths from 
both ends match up. Next we give the details about how 
the merging and matching process works. 
See Figure \ref{Merging}.

If $N$ is even write $N = 2M$; otherwise write $N=2M+1$.
We will give the details for the merging process in 
the inner-half of the spiral, i.e., $ T_1 \cup 
\dots \cup T_M$. The same procedure is used in 
the outer half
$ T_{N-M} \cup \dots T_N$; these tubes are all longer 
than the corresponding tubes in the inner half of the 
spiral, and so the allowable displacements will be larger; 
thus any displacement that can be achieved in the inner
half can also (more easily) be attained in the outer half.
Thus we can assume the merged paths in  the inner half terminate 
at the same points as the merged paths in the outer half when 
$N$ is even and that they can be joined by $P$-paths in 
$T_{M+1}$ when $N$ is odd. 

\begin{figure}[htb]
\centerline{
\includegraphics[height=2.0in]{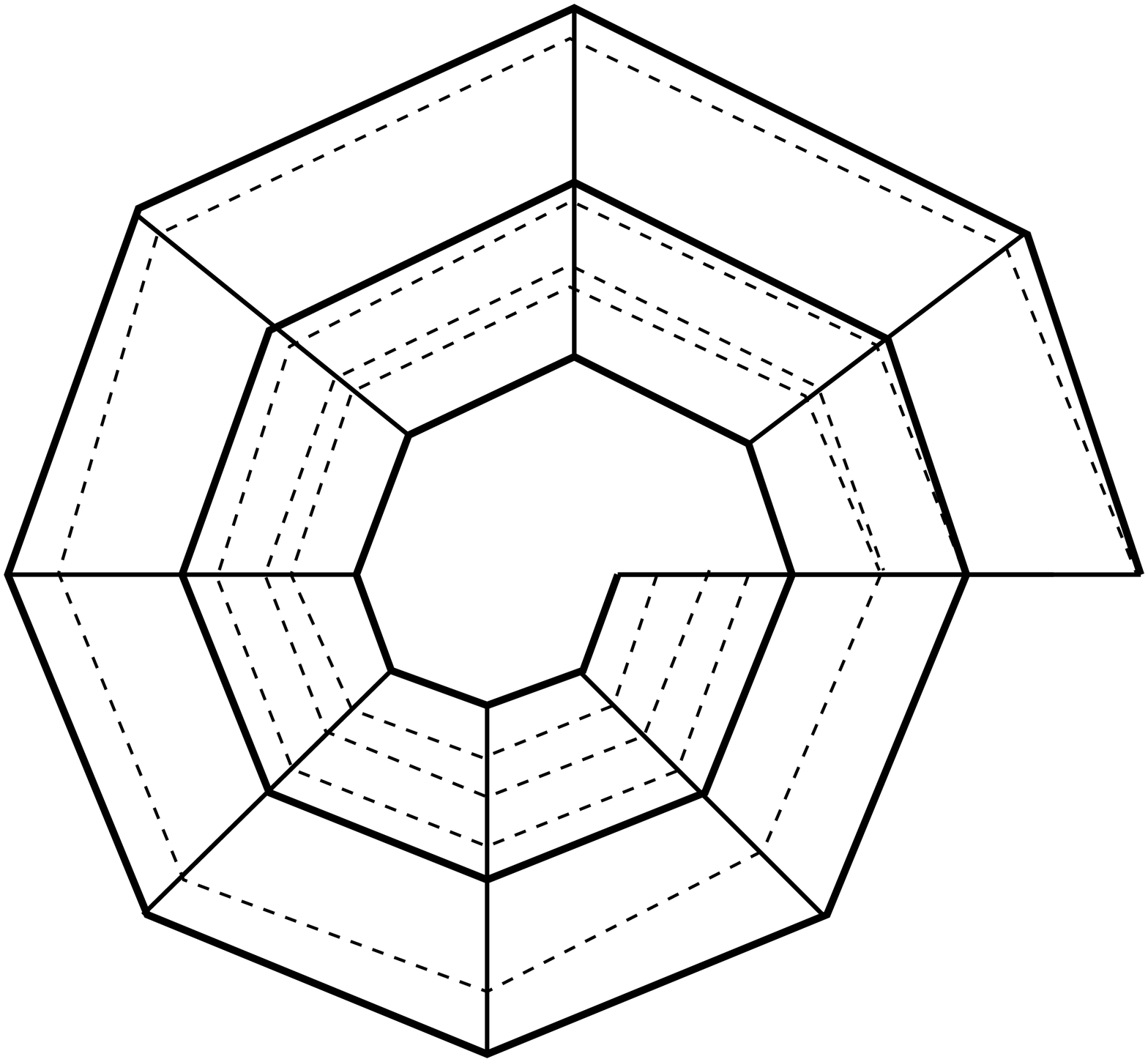}
 }
\caption{\label{Merging}
Between windings  $100$ and $16k^{1/3}$  we merge the paths
generated by the internal corners of the narrow tubes 
at the entrances to the spiral. Here we show four tubes 
merging into two tube in the first winding and the remaining 
two tubes merging into one in the second winding. (In general 
though, it will take more windings to merge wider tubes.)
}
\end{figure}

Let $\lambda = 2^{2/3} < 2$ and let $\lambda_j$ be the 
integer part of $\lambda^j$. Let $S_j = T_{\lambda_{j}+1} 
\cup \dots \cup T_{\lambda_{j+1}}$, i.e., this 
is  a sub-tube of the spiral that goes from winding $\lambda_j$
to winding $\lambda_{j+1}$. By the argument 
proving  Lemma \ref{pure spiral length}
this part of the tube has length  at least 
 \begin{eqnarray*}
  \sum_{i = \lambda_j+1}^{\lambda_{j+1}} (2i-1) 
   &\geq&     \lambda_{j+1}^2 -  (\lambda_{j}+1)^2 
   \geq  (\lambda^{j+1}-1)^2 -  (\lambda^{j}+1)^2 \\
   & = &  [\lambda^{2j+2}- 2 \lambda^{j+1}+1 
        -\lambda^{2j} - 2\lambda^j -1]   \\
    &=&   \lambda^{2j}  [\lambda^{2}- 2 \lambda^{-j+1}
        -1 - 2\lambda^{-j} ]  
\end{eqnarray*} 
An explicit calculation shows that for $j=5$ 
$$ \lambda^{2}- 2 \lambda^{-j+1} -1 - 2\lambda^{-j}  \approx 1.00644
  \geq 1 ,$$
and since this function is increasing in $j$, we deduce
that $S_j$ has length $\geq \lambda^{2j}$ 
if $j \geq 5$.  Moreover $\lambda^{10} = 2^{20/3} 
 \approx  101.594 \geq 100$.  

Suppose $S_j$ is divided into $2^{m}$ parallel sub-tubes
of width $ w = 2^{-m}$. By Lemma \ref{bend in tube} a path entering 
one of these sub-tubes can be Gabriel bent within the tube
to hit either of the far corners of the tube  if 
$$  w \leq \frac {L^2}{4pw} .$$
In this case $L  \geq  \lambda^{2j}$, $p = k \lambda^j$ 
(each piece is repeated in the tube $k \lambda^j $ times; 
once per spiral),
and $w = 2^{-m}$, so this  inequality becomes 
$$  2^{-m} \leq 
 \frac {(\lambda^{2j})^2}{4 \lambda^{j} k 2^{-m} }
$$
and this occurs iff 
$$  4 k 2^{-2m} \leq \lambda^{3j}.$$
Since  $2^{2j} = \lambda^{3j}$ and $k = 2^{2K}$, this 
is equivalent to  
$$    2^{-m} \leq 2^{j-1}/ \sqrt{ k}=2^{j-1-K}  .$$

So for $j\geq 5$
 we divide $S_j$ into $2^{K-j+1}$ parallel sub-tubes.
The tube $S_j$ shares an end with $S_{j-1}$. The  corners of
 the narrow sub-tubes of $S_j$ on this end 
form a subset of the corners of
the narrow sub-tubes of $S_{j-1}$ on the shared end.
The  remaining  corners are mid-points of  one end 
of the narrow tubes in $S_j$ and   can be propagated 
by Gabriel paths through the  narrow sub-tubes of $S_j$ so that
they hit far corners of these sub-tubes, merging another set
of paths. 

The total number of vertices created in $S_j$ is therefore
$O( \lambda^j \cdot 2^{K-j} \cdot k)$ (the number of windings
of the tube, times the number of paths propagating through the 
tube, times the number of pieces crossed in each winding).
The inner half of the spiral contains tubes $S_5, \dots, 
S_m$  as long as  $\lambda^m = 2^{2m/3} < M$, or
 $ m  \leq   \frac 32 \log_2 M$.
Summing over all $j \in [5,m]$ gives the upper bound 
$$ O(\sum_{j=5}^{\frac 32 \log_2 M} \lambda^j 2^{K-j} 2^{2K}  )
    =  O( 2^{3K} \sum_{j=5}^\infty 2^{\frac 23j -j})
    =O(2^{3K}) = O(k^{1.5}).$$
Therefore  if $N = O(k^{1/3})$, 
we can terminate every path entering 
the spiral using only $O(k^{1.5})$ new vertices in the spiral.

\section{ $16  k^{1/3}< N \leq 8 k^{1/2}:$ a single path}


If $ 16 k^{1/3}< N \leq  8 k^{1/2}$ 
 we  duplicate the construction 
of the previous section up to and including $S_{K+5}$ and
after this point we  simply let the single  remaining
path propagate as a $P$-path. Since the path winds
around the spiral at most $\sqrt{k}$ times, 
 at most $O(k^{1.5})$  new vertices
are created between $2k^{1/3}$ and $k^{1/2}$,
 and the path in the inner half eventually hits the analogous 
path from the outer half of the spiral.

\section{ $8 k^{1/2}< N \leq k:$  multiple closed curves}

When $j =8 \sqrt{k} $ we hit a new transition point: $S_j$  now has 
length $8\sqrt{k} = 2^{K+3}$, 
hence minimal length $ \geq 2\sqrt{k}$. Hence  Lemma \ref{G tube bending} 
says that if $j \geq 8 \sqrt{k}$, then 
the allowable displacement in $S_j$ is at least
$$ \frac {\tilde \ell(S_j)}{ 4k} = 
 \frac {(2 j)^2}{ 4 k}  \geq 1,$$
i.e., the allowable displacement is larger than the width 
of the tube, so we can connect opposite corners of $S_j$
by a Gabriel path (Gabriel assuming the tangent disks have 
diameter $1$, the width of the tube).

The idea is that for $j \geq 8 \sqrt{k}$ we cut the spiral 
by closed curves $\gamma_j$ in $S_j$. See Figure \ref{Rings}.
If we choose one such curve from every $S_j$ for 
$\sqrt{k} \leq j \leq k$  then 
we generate $\simeq k$ curves, each with  $k$  vertices, 
giving a total of $k^2$ new vertices. This will eventually 
lead to a $O(n^3)$ estimate in Theorem \ref{NonObtuse}
(still a polynomial bound, but larger than we want).
Instead, we will select a subsequence of tubes $\{S_{j_p}\}$
to contain closed curves. We want to choose at most 
$O(\sqrt{k})$  indices between $\sqrt{k}$ and $k$ and we 
want $\gamma_{j_p}$ to be 
Gabriel with respect to  annular region 
bounded between $\gamma_{j_{p-1}}$ and $\gamma_{j_{p+1}}$. 
The length of the $Q$-segments connecting these two
curves is exactly $ \delta_p  \equiv j_{p+1}  - j_{p-1}$
and the minimal-length of $S_{j_p}$ is at least 
$\frac 14 j_p$, so  by Lemma \ref{bend in tube}, 
(taking $L = j_p/4$, $p =  k$, $w= \delta_p$)
the Gabriel condition can be met if
$$ \frac {(j_p/4)^2} {4 k\delta_p} \geq 1,$$
or equivalently
$$ \delta_p \leq \frac {j_p^2}{64k} $$

\begin{figure}[htb]
\centerline{
\includegraphics[height=2.0in]{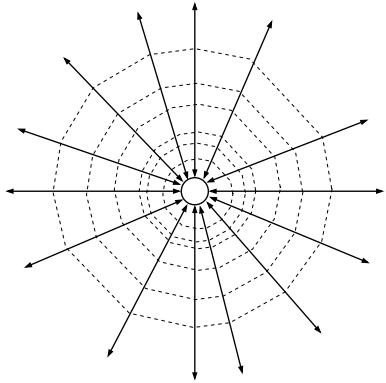}
 }
\caption{\label{Rings}
Between windings $k^{1/2}$ and $k$ we use closed 
Gabriel loops that start with unit spacing near
$k^{1/2}$  and eventually reach spacing $\simeq k$.
We shall see in the next section that beyond
radius $\simeq k$, no more loops are needed.
}
\end{figure}

 Suppose we have an  increasing  sequence of integers
$\{ j_i\}_{i=1}^P$ with
\begin{enumerate}
\item  $j_1   =  \sqrt{k} = 2^K$,
\item  $j_P   =  k = 2^{2K}$, 
\item  $P = O(\sqrt{k})$, 
\item $j_{i+1}-j_i \leq \frac 1{128k} j_i^2$.
\end{enumerate}
The fourth condition implies 
$$ |j_{i+1} - j_{i-1}  | \leq 
 |j_{i+1} - j_i|  + |j_{i} - j_{i-1}|  
\leq \frac 1{128k} (j_i^2 +j_{i-1}^2)
\leq \frac 1{64k} j_i^2 ,$$
since $j_{i-1} \leq j_i$.
So if we can find such a set of integers, we 
will have constructed the desired closed loops. 

We divide the interval $[\sqrt{k}, k] = [2^K, 2^{2K}]$ into 
geometrically increasing blocks 
$$ [2^K, 2^{K+1}],
 \dots, [2^{K+q}, 2^{K+q+1} ], 
  \dots, 
 [2^{2K-1} , 2^{2K}] .$$
In  the $q$th block we can take our sequence to be 
separated by gaps of size 
$ \frac 1{128 k} (2^q \sqrt{k})^2 = 2^{2q-7}$
and hence $2^{K+7-q}$ evenly spaced integers 
suffice to cross this block with the desired spacing.
Summing over all $q$'s shows that at most 
$$  \sum_{q=0}^K 2^{K+7-q} \leq
       2^7 \cdot   2^K = 128 \sqrt{k},$$
integers $j_p$ are needed. Thus at most $O(k^{1.5})$ new 
vertices are created in this phase of the construction.

\section{ $ N > k:$  the empty region} \label{empty}

In the previous section we created Gabriel loops in the 
spiral, assuming there were other loops nearby to prevent 
the Gabriel disks from getting too  large; thus we 
needed a sequence of larger and larger loops. However, once
$j \geq k$, we no longer need 
to limit the size of the Gabriel disks and this sequence 
of loops can end. More precisely,

\begin{lemma}
Suppose the spiral has $k$ pieces and the spiral 
has more than $k/4$ windings. Then 
each sub-tube $T_j$ with $j \geq  k/4$ is crossed by a 
closed loop  so that the vertices of the loop make 
each isosceles piece Gabriel.
\end{lemma} 

\begin{proof}
As above, we assume $T_j$ has width $1$.
Lemma \ref{G tube bending}  implies that the
 allowable displacement 
for an isosceles piece 
is at least $|a-b|^2/4R$, where $R$ is the distance to the 
vertex of the piece. We also have 
$ |a-b| =  2 R  \sin(\theta/2)$ where $\theta \in [0, \pi)$
is the angle of the piece. 
Since $\sin(\theta/2) \geq \sqrt{2} \theta /\pi$, 
the allowable displacement for such a piece is at least
$$    R \sin^2(\theta/2) \geq 
          R \frac  { 2}{\pi^2} \theta^2  .$$
  If we only bend the path for pieces 
where the vertex direction points towards the center of
the spiral, the sum of the angles is at least $2 \pi$, so 
by the Cauchy-Schwarz inequality,
the sum of the angles squared is at least $(2\pi)^2/k$.
Hence the total  allowable displacement in the tube $T_j$ 
 is at least $8R/k$ times the 
minimal-length of  $T_j$. In $T_j$ every piece with vertex 
towards the center of the spiral has $R \geq m-1 \geq m/2$,
(since all our pieces have angle $\leq \pi/2$, there are at
least $4$ of them)
so the allowable displacement is at least $4m/k \geq 1$, 
the width of $T_j$.
\end{proof} 

Up to this point, we have divided the spiral into 
inner and outer halves, but at this stage this is 
not necessary. The outermost tube  contains a Gabriel 
curve that connects opposite corners, forming  a loop
and every  entering path can be Gabriel propagated to 
hit this same corner. The region between this loop and 
the loop formed in $S_k$ is left empty.

The number of points created in a spiral constructed 
from $k$ isosceles pieces is $O(k^{1.5})$ with a 
constant that is bounded independent of $N$, the number
of windings. Since for each spiral $k =O(n)$ ($n$ 
recall that $n$ is the number of vertices in the PSLG)
and since there can only be $O(n)$ spirals (each is a 
return region and there are only $O(n)$ return regions),
the total number of points added due to all spirals is 
$O(n^{2.5})$.  This completes the proof of Theorem 
\ref{NonObtuse}.


\bibliography{nonobtuse}
\bibliographystyle{plain}
\end{document}